\documentclass[11pt]{article}
\pdfoutput=1

\usepackage{fullpage, graphics}
\usepackage{amsmath, amssymb, amsthm} 
\usepackage{subfig}  
\usepackage{hyperref}   
\usepackage{aliascnt}  
\usepackage{caption}
\usepackage{float}
\usepackage{xspace}

\usepackage{multirow}
\usepackage{makecell}

\hypersetup{colorlinks,
 citecolor=blue,
  linkcolor=black,
  urlcolor=black}

\usepackage{titling}
\preauthor{\begin{center}\large \scshape}
\postauthor{\par\end{center}}

\usepackage{mathpazo} 
\usepackage[scaled]{helvet} 
\usepackage{courier} 
\normalfont
\usepackage[T1]{fontenc}

\newenvironment{packed_item}{
\begin{itemize}
   \setlength{\itemsep}{1pt}
   \setlength{\parskip}{0pt}
   \setlength{\parsep}{0pt}
}
{\end{itemize}}

\newenvironment{packed_enum}{
\begin{enumerate}
   \setlength{\itemsep}{1pt}
  \setlength{\parskip}{0pt}
   \setlength{\parsep}{0pt}
}
{\end{enumerate}}

\newcommand{\ie}{{\em i.e.}\xspace}
\newcommand{\eg}{{\em e.g.}\xspace}

\allowdisplaybreaks

\usepackage{aliascnt}  	

\newtheorem{theorem}{Theorem}[section]          	
\newaliascnt{lemma}{theorem}				
\newtheorem{lemma}[lemma]{Lemma}              	
\aliascntresetthe{lemma}  					
\newaliascnt{conjecture}{theorem}			
\aliascntresetthe{conjecture}  				
\newaliascnt{remark}{theorem}				
\aliascntresetthe{remark}  					
\newaliascnt{corollary}{theorem}			
\newtheorem{corollary}[corollary]{Corollary}      
\aliascntresetthe{corollary}  				
\newaliascnt{definition}{theorem}			
\newtheorem{definition}[definition]{Definition}    
\aliascntresetthe{definition}  				
\newaliascnt{proposition}{theorem}			
\newtheorem{proposition}[proposition]{Proposition}  
\aliascntresetthe{proposition}  				
\newaliascnt{example}{theorem}			
\newtheorem{example}[example]{Example}  	
\aliascntresetthe{example}  				

\providecommand{\bx}[0]{\mathbf{x}}
\providecommand{\bm}[0]{\mathbf{m}}
\providecommand{\bM}[0]{\mathbf{M}}
\providecommand{\bh}[0]{\mathbf{h}}

\providecommand{\bu}[0]{\mathbf{u}}

\providecommand{\ba}[0]{\mathbf{a}}
\providecommand{\be}[0]{\mathbf{e}}
\providecommand{\E}[0]{\mathbb{E}}

\newcommand{\set}[1]{\{#1\}}                    
\newcommand{\setof}[2]{\{{#1}\mid{#2}\}}        

\newcommand{\mA}[0]{\mathcal{A}} 
\newcommand{\mB}[0]{\mathcal{B}}

\newcommand{\mM}[0]{\mathcal{M}}

\newcommand{\msg}[0]{{\mbox{msg}}}

\newcommand{\mpc}[0]{MPC}

\renewcommand{\P}{\mathbb{P}}
\newcommand{\contracted}{{\overline M}}

\newcommand{\vars}[1]{\text{vars}(#1)}

\newcommand{\lup}[0]{L^{upper}}
\newcommand{\llo}[0]{L^{lower}}

\newcommand{\proj}{{\sf \Pi}}

\begin{document}

\title{Communication Cost in Parallel Query Processing\thanks{This work is partially supported by NSF IIS-1115188, IIS-0915054, IIS-1247469 and CCF-1217099.}}
\author{ PAUL BEAME,  University of Washington \\
PARASCHOS KOUTRIS, University of Wisconsin-Madison \\
DAN SUCIU, University of Washington}
\date{}

\maketitle

\begin{abstract}
We study the problem of computing conjunctive queries over large databases
on parallel architectures without shared storage.
Using the structure of such a query $q$ and the skew in the data,
we study tradeoffs between the number of processors, the number of rounds
of communication, and the per-processor {\em load} -- the number of bits 
each processor can send or can receive in a single round -- that are required to
compute $q$.  Since each processor must store its received bits, the load
is at most the number of bits of storage per processor. 

When the data is free of skew, we obtain essentially tight upper and lower
bounds for one round algorithms and we show how the bounds degrade when there
is skew in the data.
In the case of skewed data, we show how to improve the algorithms
when approximate degrees of the (necessarily small number of) heavy-hitter
elements are available, obtaining essentially optimal algorithms
for queries such as skewed simple joins and skewed triangle join queries.

For queries that we identify as {\em tree-like}, we also
prove nearly matching upper and lower bounds for multi-round algorithms for
a natural class of skew-free databases.
One consequence of these latter lower bounds is that for any $\varepsilon>0$,
using $p$ processors to compute the connected components of a graph, or to
output the path, if any, between a specified pair of vertices of a graph with
$m$ edges and per-processor load that is $O(m/p^{1-\varepsilon})$ requires 
$\Omega(\log p)$ rounds of communication.

Our upper bounds are given by simple structured algorithms using MapReduce.
Our one-round lower bounds are proved in a very general model, which we
call the {\em Massively Parallel Communication (MPC)} model, that allows
processors to communicate arbitrary bits.
Our multi-round lower bounds apply in a restricted version of the MPC model
in which processors in subsequent rounds after the first communication round
are only allowed to send tuples.
\end{abstract}

\section{Introduction}
\label{sec:introduction}

  Most of the time spent during big data analysis today is allocated in data
  processing tasks, such as identifying relevant data, cleaning,
  filtering, joining, grouping, transforming, extracting features, and
  evaluating results~\cite{DBLP:conf/pods/Chaudhuri12,datasciencesurvey}.  These
  tasks form the main bottleneck in big data analysis, and it is a major
  challenge to improve the performance
  and usability of data processing tools.  The motivation for this
  paper comes from the need to understand the complexity of query
  processing in big data management.

  Query processing on big data is typically performed on a shared-nothing 
  parallel architecture. In this setting, the data is stored on a large number of 
  independent servers interconnected by a fast network.  The servers 
  perform local computations, and can also communicate with each other to
  exchange data. 
 Starting from MapReduce~\cite{DBLP:conf/osdi/DeanG04}, the last decade
 has seen the development of several massively parallel frameworks that support
 big data processing, including  PigLatin~\cite{DBLP:conf/sigmod/OlstonRSKT08},
  Hive~\cite{TSJSCALWM09}, Dremmel~\cite{DBLP:journals/pvldb/MelnikGLRSTV10} and
 Shark~\cite{Shark}. 

Unlike traditional query processing, the time complexity is no longer
dominated by the number of disk accesses.  Typically, a query is
evaluated by a sufficiently large number of servers such that the
entire data can be kept in main memory.  In these systems, the new
complexity parameter is the {\em communication cost}, which depends 
on both the amount of data being exchanged and the number of global 
synchronization barriers (rounds).

\paragraph{Contributions}

We define the {\em Massively Parallel Communication} (\mpc) model,
to analyze the tradeoff between the number of rounds and the amount of
communication required in a massively parallel computing
environment.  We include the number of servers $p$ as a
parameter, and allow each server to be infinitely powerful, subject
only to the data to which it has access.
The computation proceeds in {\em rounds}, where each round consists
of local computation followed by global exchange of data between all
servers. 

An algorithm in the \mpc\ model is characterized by the number of servers
$p$, the number of rounds $r$, and the maximum number of bits $L$, or 
{\em maximum load}, that each server receives at any round.
There are no other restrictions on communication between servers. 
Though the storage capacity of each server is not a separate parameter, since
a server needs to store the data it has received in order to operate
on it, the load $L$ is always a lower bound on the storage capacity of
each server.
An ideal parallel algorithm with input size $M$ would distribute the input
data equally among the $p$ servers, so each server would have a 
maximum load of $M/p$, and would perform the computation in a single round.
In the degenerate case where $L = M$, the entire data can be sent to a single
server, and thus there exists no parallelism.

The focus of the \mpc\ model on the communication load 
captures a key property of the system architectures assumed by MapReduce
and related programming abstractions\footnote{This focus is somewhat related
to the generic approach to communication in
Valiant's Bulk Synchronous Parallel (BSP)
model~\cite{DBLP:journals/cacm/Valiant90} which the \mpc\ model simplifies and 
strengthens.
We discuss the relationship of the \mpc\ model to a variety of MapReduce and
parallel models in~\autoref{sec:related}.}.
Since there is no restriction on the form of communication or the kinds of 
operations allowed for processing of local data, the lower bounds we
obtain in the \mpc\ model apply much more generally than those bounds
based on specific assumed primitives or communication structures such as 
those for MapReduce.

We establish both lower and upper bounds in the \mpc\ model for 
computing a full conjunctive query $q$, in three different settings. 

First, we restrict the computation to a single communication round and 
to input data without skew. In particular, given a query $q$ over
relations $S_1, S_2, \ldots$ such that that an input relation $S_j$
has size $M_j$ (in bits), we examine the minimum load $L$ for which it
is possible to compute $q$ in a single round. We show that any algorithm
that correctly computes $q$ requires a load
$$ L \geq \max_{\bu} \left(\frac{\prod_j M_j^{u_j}}{p} \right)^{1/\sum_j u_j}$$
where $\bu = (u_1, \dots, u_{\ell})$ is a fractional {\em edge packing} for the
hypergraph of $q$. 
Our lower bound applies to the strongest possible model in which servers can
encode any information in their messages, and have access
to a common source of randomness.  This is stronger than the
lower bounds
in~\cite{DBLP:journals/corr/abs-1206-4377,DBLP:conf/pods/KoutrisS11},
which assume that the units being exchanged are tuples.
We further show that a simple algorithm, which we call the HyperCube algorithm,
matches our lower bound for any conjunctive query when the input data has
no skew. As an example, for the triangle query $C_3 = S_1(x,y), S_2(y,z), S_3(z,x)$ 
with sizes $M = |S_1| = |S_2| = |S_3|$, we show that the lower bound for the load is 
$\Omega(M/p^{2/3})$, and the HyperCube algorithm can match this bound.

%
%

Second, we study how skew influences the computation. A value in the 
database is {\em skewed}, and is called a {\em heavy hitter} when it occurs 
with much higher frequency than some predefined threshold.
Since data distribution is typically done using hash-partitioning, 
unless they are handled differently from other values, all tuples 
containing a heavy hitter will be sent to the same server, causing it to be overloaded. 
The standard technique that handles skew consists of first detecting the heavy hitters, 
then treating them differently from the other values, e.g. by partitioning tuples with 
heavy hitters on the other attributes. 

In analyzing the impact of skew, we first provide bounds on the
behavior of algorithms that are not given special information about
heavy hitters and hence are limited in their ability to deal with
skew.  We then consider a natural model for handling skew which
assumes that at the beginning of the computation all servers know the
identity of all heavy hitters, and the (approximate) frequency of each
heavy hitter.  (It will be easy to see that there can only be a small
number of heavy hitters and this kind of information can be easily
obtained in advance from small samples of the input.)  Given these
statistics, we present upper and lower bounds for the maximum load for
full conjunctive queries. In particular, we present a general lower
bound that holds for any conjunctive query. We next give matching
upper bounds for the class of {\em star joins}, which are queries of
the form $q(z, x_1,\ldots,x_k) = S_1(z, x_1), S_2(z, x_2), \dots,
S_k(z, x_k)$ (this includes the case of the simple join query for
$k=2$), as well as the triangle query.

Third, we establish lower bounds for multiple communication rounds,
for a restricted version of the \mpc\ model, called \emph{tuple-based}
\mpc\ model.  The messages sent in the first round are still
unrestricted, but in subsequent rounds the servers can send only
tuples, either base tuples in the input tables, or join tuples
corresponding to a subquery; moreover, the destinations of each tuple
may depend only on the tuple content, the message received in the
first round, the server, and the round.  We note that any multi-step
MapReduce program is tuple-based, because in any map function the key
of the intermediate value depends only on the input tuple to the map
function.
Here, we prove that the number of rounds required is, essentially,
given by the depth of a query plan for the query, where each operator
is a subquery that can be computed in one round with the required load.  
For example, to compute a length $k$ chain query
$L_k$, if we want to achieve $L = O(M/p)$, the optimal computation is a bushy join
tree, where each operator is $L_2$ (a two-way join) and the optimal
number of rounds is $\log_2 k$. If $L = O(M/p^{1/2})$, then we can use
$L_4$ as operator (a four-way join), and the optimal number of rounds
is $\log_4 k$.  More generally, we show nearly matching upper and
lower bounds based on graph-theoretic properties of the query.
 

We further show that our results for conjunctive path queries imply that
any tuple-based \mpc\ algorithm with
load $L < M$ requires $\Omega(\log p)$ rounds to
compute the 
connected components of sparse
undirected graphs of size $M$ (in bits).
This is an interesting contrast to the results
of~\cite{DBLP:conf/soda/KarloffSV10}, which show that connected components
(and indeed minimum spanning trees) of undirected graphs can be computed in
only two rounds of
MapReduce provided that the input graph is sufficiently dense.


\begin{table*}
\centering
\renewcommand{\arraystretch}{2.0}
\newcommand{\specialcell}[2][c]{%
  \begin{tabular}[#1]{@{}c@{}}#2\end{tabular}}
\resizebox{\textwidth}{!} {
\begin{tabular}[c]{| c | c | c | c |} 
\hline
 Number of Rounds & Data Distribution & Upper Bound & Lower Bound  \\ \hline
\multirow{3}{*}{\makecell{1 round \\ (\mpc\ model)}} 
& no skew & \autoref{subsec:algorithms} & \autoref{subsec:oneround:lower}  \\ \cline{2-4}
& skew (oblivious) & \autoref{subsec:skew:without} & \autoref{subsec:skew:without}  \\ \cline{2-4}
& skew (with information) & \autoref{sec:skew:join} & \autoref{sec:skew:lower}  \\ \hline
\multirow{1}{*}{\makecell{multiple rounds \\ (tuple-based \mpc\ model)}} 
& no skew & \autoref{sec:multi:upper} & \autoref{sec:multi:lower} \\ \hline
\end{tabular} 
}
  \caption{Roadmap for the organization of the results presented in this article.} 
  \label{tab:roadmap}
\end{table*}

By being explicit about the number of processors, in the \mpc\ model we must
directly handle issues of load balancing and skew in task (reducer) sizes that 
are often ignored in MapReduce algorithms but are actually critical for
good performance (e.g., see~\cite{DBLP:conf/sigmod/KwonBHR12}).
When task sizes are similar, standard analysis shows that
hash-based load balancing works well.
However, standard bounds do not yield sharp results when there is significant
deviation in sizes.
In order to handle such situations, we prove a sharp Chernoff bound
for weighted balls in bins that is particularly suited to the analysis
of hash-based load balancing with skewed data.
This bound, which is given in \autoref{sec:hashing}, should be of independent
interest.

\paragraph{Organization} We start by presenting the \mpc\ model and defining
important notions  in~\autoref{sec:model}. 
In \autoref{sec:onestep}, we describe the upper and lower bounds for computation
restricted to one round and data without skew. We study the effect of data skew in
\autoref{sec:skew}. In \autoref{sec:multistep}, we present upper and lower bounds
for the case of multiple rounds. We conclude by discussing the related work in
\autoref{sec:related}. In~\autoref{tab:roadmap}, the reader can view a more detailed 
roadmap for the results of this article.

\section{Model}
\label{sec:model}

In this section, we present in detail the MPC model.

\subsection{Massively Parallel Communication}
\label{subsec:model}

In  MPC model,  computation is performed by $p$ servers, or processors, 
connected by a complete network of private channels. 
The computation proceeds in {\em steps}, or {\em rounds}, where each round consists of 
two distinct phases:
\begin{description}
\item[Communication Phase]  The servers exchange data, each by communicating with all other servers (sending and receiving data).
\item[Computation Phase] Each server performs only local computation.
\end{description}

The input data of size $M$ bits is initially uniformly partitioned among the $p$ servers, \ie 
each server stores $M/p$ bits of the data: this describes the way the data is typically 
partitioned in any distributed storage system. There are no assumptions on the
particular partitioning scheme. At the end of the execution, the output must be present
in the union of the $p$ processors. 

The execution of a parallel algorithm in the MPC model is captured by two
basic parameters:
\begin{description}
\item[The number of rounds $r$] This parameter denotes the number of 
synchronization barriers that an algorithm requires.
\item[The maximum load $L$] This parameter denotes the {\em maximum load} 
among all  servers at any round, where the load is the
amount of data received by a server during a particular round.
\end{description} 

Normally, the entire data is exchanged during the first communication
round, so the load $L$ is at least $M/p$.  On the other hand, the load
is strictly less than $M$: otherwise, if we allowed a load $L=M$, then
any problem can solved trivially in one round, by simply sending the
entire data to server 1, then computing the answer locally.  Our
typical loads will be of the form $M/p^{1-\varepsilon}$, for some $0
\leq \varepsilon < 1$ that depends on the query.  For a similar
reason, we do not allow the number of rounds to reach $r = p$, because
any problem can be solved trivially in $p$ rounds by sending at each
round $M/p$ bits of data to server 1, until this server accumulates
the entire data.  In this paper we only consider $r = O(1)$.

\paragraph{Input Servers}

As explained above, the data is initially distributed uniformly on the
$p$ servers; we call form of input {\em partitioned input}.  When
computing queries over a fixed relational vocabulary $S_1, \ldots,
S_\ell$, we consider an alternative model, where each relation $S_j$
is stored on a separate server, called an {\em input server}; during
the first round the $\ell$ input servers distribute their data to the
$p$ workers, then no longer participate in the computation.  The
input-server model is potentially more powerful, since the $j$'th
input server has access to the entire relation $S_j$, whose size is
much larger than $M/p$.  We state and prove all our lower bounds for
the input-server model.  This is w.l.o.g., because any algorithm in
the partitioned-input model with load $L$ can be converted into an
input-server algorithm with the same load, as follows.  Denote $f_j =
|S_j| / (\sum_i |S_i|)$ for all $j=1,\dots, \ell$: we assume these numbers
are known by all input servers, because we assume the statistics
$|S_j|$ known to the algorithm.  Then, each input server $j$ holding
the relation $S_j$ will partition $S_j$ into $f_j p$ equal fragments,
then will simulate $f_j p$ workers, each processing one of the
fragments.  Thus, our lower bounds for the input-sever model
immediately apply to the partitioned-input model.

\paragraph{Randomization}

The \mpc\ model allows randomization.  The random bits are
available to all servers, and are computed independently of the input
data.  The algorithm may fail to produce its output with a small
probability $\eta>0$, independent of the input.  For example, we use
randomization for load balancing, and abort the computation if the
amount of data received during a round would exceed the
maximum load $L$, but this will only happen with
exponentially small probability.

To prove lower bounds for randomized algorithms, we use Yao's
Lemma~\cite{yao83}.  We first prove bounds for {\em deterministic}
algorithms, showing that any algorithm fails with probability at
least $\eta$ over inputs chosen randomly from a distribution $\mu$. 
This implies, by Yao's Lemma, that every randomized algorithm with the 
same resource bounds will fail on some input (in the support of $\mu$) with 
probability at least $\eta$ over the algorithm's random choices.

\subsection{Conjunctive Queries}
\label{subsec:cq}

In this paper we consider a particular class of problems for the \mpc\
model, namely computing answers to conjunctive queries over a
database.  We fix an input vocabulary $S_1, \ldots, S_\ell$, where
each relation $S_j$ has a fixed arity $a_j$; we denote $a = \sum_{j
  =1}^{\ell} a_j$.  The input data consists of one relation instance
for each symbol.  

We consider full conjunctive queries (CQs) without self-joins, denoted
as follows:
\begin{equation} \label{eq:q}
  q(x_1,\ldots, x_k) = S_1(\bar x_1), \ldots, S_\ell(\bar x_\ell) 
\end{equation}

The query is {\em full}, meaning that every variable in the body
appears in the head (for example $q(x) = S(x,y)$ is not full), and
{\em without self-joins}, meaning that each relation name $S_j$
appears only once (for example $q(x,y,z) = S(x,y), S(y,z)$ has a
self-join).  The first restriction, to full conjunctive queries, is a
limitation: our lower bounds do not carry over to general conjunctive
queries (but the upper bounds do carry over).  The second restriction,
to queries without self-joins, is w.l.o.g.\footnote{To see this,
  denote $q'$ the query obtained from $q$ by giving distinct names to
  repeated occurrences of the same relations.  Any algorithm for $q'$
  is automatically an algorithm for $q$, with the same load.
  Conversely, any algorithm $A$ for $q$ can also be converted into an
  algorithm for $q'$, having the same load as $A$ has on an input that
  is $\ell$ times larger.  The algorithm $A$ is obtained as
  follows. For each atom $S(x,y,z,\ldots)$ occurring in the query,
  create a copy of the entire relation $S$ by renaming every tuple
  $(a,b,c,\ldots)$ into $((a,"x"), (b,"y"), (c,"z"), \ldots)$.  That
  is, each value $a$ in the first column is replaced by the pair
  $(a,"x")$ where $x$ is the variable occurring in that column, and
  similarly for all other columns.  This copy operation can be done
  locally by all servers, without communication.  Furthermore, each
  input relation is copied at most $\ell$ times.  Finally, run the
  algorithm $A'$ on the copied relations.}

The {\em hypergraph} of a query $q$ is defined by introducing one node
for each variable in the body and one hyperedge for each set of
variables that occur in a single atom. We say that a conjunctive query
is {\em connected} if the query hypergraph is connected. For example,
$q(x,y) = R(x),S(y)$ is not connected, whereas $q(x,y) = R(x), S(y),
T(x,y)$ is connected.  We use $\text{vars}(S_j)$ to denote the set of
variables in the atom $S_j$, and $\text{atoms}(x_i)$ to denote the set
of atoms where $x_i$ occurs; $k$ and $\ell$ denote the number of
variables and atoms in $q$, as in~\eqref{eq:q}. The {\em connected
  components} of $q$ are the maximal connected subqueries of $q$.


\paragraph{Characteristic of a Query} 

The {\em characteristic} of a conjunctive query $q$ as in \eqref{eq:q}
is defined as $\chi(q) = a - k - \ell + c$, where $a = \sum_j a_j$ is
the sum of arities of all atoms, $k$ is the number of variables,
$\ell$ is the number of atoms, and $c$ is the number of connected
components of $q$.\footnote{In the preliminary version of this
  paper~\cite{BKS13} we defined $\chi(q)$ with the opposite sign (as
  $-a + k + \ell -c$); we find the current definition more natural
  since now $\chi(q) \geq 0$ for every $q$.}

For a query $q$ and a set of atoms $M \subseteq \text{atoms}(q)$,
define $q/M$ to be the query that results from contracting the edges
in the hypergraph of $q$. As an example, if we define
$$ L_k = S_1(x_0, x_1), S_2(x_1, x_2), \dots , S_k(x_{k-1},x_k)$$
we have that $L_5/ \{S_2, S_4 \}= S_1(x_0, x_1), S_3(x_1,x_3),
S_5(x_3,x_5)$.

\begin{lemma} \label{lemma:chi} The characteristic of a query $q$
  satisfies the following properties:
\begin{enumerate}
\item[(a)] If $q_1, \ldots, q_c$ are the connected components of $q$, then
$\chi(q)=\sum_{i=1}^c \chi(q_i)$.
\item[(b)] For any $M \subseteq \text{atoms}(q)$, $\chi(q/M)=\chi(q) -\chi(M)$.
\item[(c)] $\chi(q) \geq 0$.  
\item[(d)] For any $M \subseteq \text{atoms}(q)$, $\chi(q) \geq
  \chi(q/M)$.
\end{enumerate}
\end{lemma}

\begin{proof}
  Property (a) is immediate from the definition of $\chi$, since the
  connected components of $q$ are disjoint with respect to variables
  and atoms.  Since $q/M$ can be produced by contracting according to
  each connected component of $M$ in turn, by property (a) and
  induction it suffices to show that property (b) holds in the case
  that $M$ is connected.  If a connected $M$ has $k_M$ variables,
  $\ell_M$ atoms, and total arity $a_M$, then the query after
  contraction, $q/M$, will have the same number of connected
  components, $k_M-1$ fewer variables, and the terms for the number of
  atoms and total arity will be reduced by $a_M-\ell_M$ for a total
  reduction of $a_M-k_M-\ell_M+1=\chi(M)$.  Thus, property (b)
  follows.

  By property (a), it suffices to prove (c) when $q$ is connected.  If
  $q$ is a single atom $S_j$ then $\chi(S_j) \geq 0$, since the number
  of variables is at most the arity $a_j$ of the atom.  If $q$ has
  more than one atom, then let $S_j$ be any such atom: then $\chi(q) =
  \chi(q/S_j) + \chi(S_j) \geq \chi(q/S_j)$, because $\chi(S_j) \geq
  0$.
  %
  Property (d) follows from (b) using the fact that $\chi(M) \geq 0$.
\end{proof}

For a simple illustration of property (b), consider the example above
$L_5 / \set{S_2, S_4}$, which is equivalent to $L_3$.  We have
$\chi(L_5) = 10 - 6 - 5 + 1 = 0$, and $\chi(L_3) = 6 - 4 - 3 + 1 = 0$,
and also $\chi(M) = 0$ (because $M$ consists of two disconnected
components, $S_2(x_1,x_2)$ and $S_4(x_3,x_4)$, each with
characteristic 0).  For a more interesting example, consider the query
$K_4$ whose graph is the complete graph with 4 variables:
$$K_4 = S_1(x_1,x_2), S_2(x_1,x_3), S_3(x_2,x_3), S_4(x_1, x_4), S_5(x_2,x_4), S_6(x_3,x_4)$$
and denote $M = \set{S_1,S_2,S_3}$.  Then $K_4/M = S_4(x_1,x_4),
S_5(x_1,x_4), S_6(x_1,x_4)$ and the characteristics are: $\chi(K_4) =
12 - 4 - 6 + 1 = 3$, $\chi(M) = 6 - 3 - 3 + 1 = 1$, $\chi(K_4/M) = 6 -
2 - 3 + 1 = 2$.

Finally, we define a class of queries that will be used later in the paper.

\begin{definition}
A conjunctive query $q$ is {\em tree-like} if $q$ is connected and
$\chi(q) = 0$.
\end{definition}

For example, the query  $L_k$ is tree-like; in fact, a query over a binary 
vocabulary is tree-like if and only if its hypergraph is a tree. 
Over non-binary vocabularies, if a query is tree-like then it is  acyclic,
but the converse does not hold: $q =
S_1(x_0,x_1,x_2),S_2(x_1,x_2,x_3)$ is acyclic but not tree-like. 
An important property of tree-like queries is that every connected subquery will
be also tree-like.

\paragraph{Fractional Edge Packing}

A {\em fractional edge packing} (also known as a {\em fractional
  matching}) of a query $q$ is any feasible solution
$\mathbf{u} = (u_1, \dots, u_{\ell})$ of the following linear
constraints:
\begin{align}
  & \forall i \in [k]:
  \sum_{j: i \in S_j} u_j \leq 1 \label{eq:cover:dual} \\
  & \forall j \in [\ell] : u_j \geq 0 \nonumber
\end{align}

The edge packing associates a non-negative weight $u_j$ to each
atom $S_j$ such that for every variable $x_i$, the sum of the weights for
the atoms that contain $x_i$ do not exceed 1. If all inequalities are
satisfied as equalities by a solution to the LP, we say that the
solution is {\em tight}.  The dual notion is a {\em fractional vertex
  cover} of $q$, which is a feasible solution $\mathbf{v} = (v_1,
\ldots, v_k)$ to the following linear constraints:
\begin{align*}
  & \forall j \in [\ell]:
\sum_{i: i \in S_j} v_i \geq 1 \\
 & \forall i \in [k]: v_i \geq 0
\end{align*}
At optimality, $\max_{\mathbf{u}} \sum_j u_j = \min_{\mathbf{v}}
\sum_i v_i$; this quantity is denoted $\tau^*$ and is called the {\em
  fractional vertex covering number of $q$}.

\begin{example}
  An edge packing of the query $L_3 = S_1(x_1, x_2),S_2(x_2,x_3),
  S_3(x_3, x_4)$ is any solution to $u_1 \leq 1$, $u_1+u_2 \leq 1$,
  $u_2 + u_3 \leq 1$ and $u_3 \leq 1$. In particular, the solution
  $(1,0,1)$ is a tight edge packing; it is also an optimal packing,
  thus $\tau^* = 2$.
\end{example}


We also need to refer to the {\em fractional edge cover}, which is a
feasible solution $\mathbf{u} = (u_1, \dots, u_{\ell})$ to the system
above where $\leq$ is replaced by $\geq$ in Eq.(\ref{eq:cover:dual}).
Every tight fractional edge packing is a tight fractional edge cover,
and vice versa.  The optimal value of a fractional edge cover is
denoted $\rho^*$.  The fractional edge packing and cover have no
connection, and there is no relationship between $\tau^*$ and
$\rho^*$.  For example, for $q=S_1(x,y),S_2(y,z)$, we have $\tau^* =
1$ and $\rho^* = 2$, while for $q = S_1(x),S_2(x,y),S_3(y)$ we have
$\tau^* = 2$ and $\rho^* = 1$.  The two notions coincide, however,
when they are tight, meaning that a tight fractional edge cover is
also a tight fractional edge packing and vice versa.  The fractional
edge cover has been used recently in several papers to prove bounds on
query size and the running time of a sequential algorithm for the
query~\cite{DBLP:conf/focs/AtseriasGM08,DBLP:conf/pods/NgoPRR12}; for
the results in this paper we need the fractional packing.

\subsection{Entropy}
 
Let us fix a finite probability space. For random variables $X$
and $Y$, the {\em entropy} and the {\em conditional entropy} are 
defined respectively as follows:
\begin{align}
  H(X) = & - \sum_x P(X=x)  \log P(X=x) \\
  H(X \mid Y) = & \sum_y P(Y=y) H(X \mid Y=y)  \label{eq:cond:ent:v}
\end{align}
The entropy satisfies the following basic inequalities:
\begin{align}  
  H(X \mid Y) &\leq  H(X) \nonumber \\
  H(X, Y) &=  H(X \mid Y) + H(Y) \label{eq:cond:ent}
\end{align}
Assuming additionally that $X$ has a support of size $n$:
\begin{align}
  H(X) &\leq  \log n  \label{eq:entropy:size}
\end{align}

\subsection{Friedgut's Inequality}

Friedgut~\cite{friedgut2004hypergraphs} introduces
the following class of inequalities.  Each inequality is described by
a hypergraph, which in our paper corresponds to a query, so we will
describe the inequality using query terminology.  Fix a query $q$ as
in \eqref{eq:q}, and let $n > 0$.  For every atom $S_j(\bar x_j)$ of
arity $a_j$, we introduce a set of $n^{a_j}$ variables $w_{j}(\ba_j) \geq  0$, 
where $\ba_j \in [n]^{a_j}$. If $\ba \in [n]^a$, we denote by $\ba_j$ the
vector of size $a_j$ that results from projecting on the variables of the
relation $S_j$.
Let $\mathbf{u} = (u_1, \dots, u_{\ell})$ be a 
fractional {\em edge cover} for $q$.  Then:
\begin{align}
  \sum_{\ba \in [n]^k} \prod_{j=1}^{\ell} w_{j}( \ba_j) \leq & \prod_{j=1}^{\ell}
  \left(\sum_{\ba_j \in [n]^{a_j}}  w_{j} (\ba_j)^{1/u_j}\right)^{u_j} \label{eq:friedgut}
\end{align}
We illustrate Friedgut's inequality on the queries $C_3$ and $L_3$:
\begin{align}
C_3(x,y,z) = S_1(x,y),S_2(y,z),S_3(z,x) \nonumber\\
L_3(x,y,z,w) = S_1(x,y),S_2(y,z),S_3(z,w) \label{eq:ccc}
\end{align}
Consider the cover $(1/2,1/2,1/2)$ for $C_3$, and the cover $(1,0,1)$ for $L_3$.
Then, we obtain the following inequalities, where $\alpha,\beta,\gamma$ stand for
$w_1,w_2,w_3$ respectively:
\begin{align*}
  \sum_{x,y,z \in [n]}\kern -1em \alpha_{xy}\cdot \beta_{yz} \cdot \gamma_{zx} \leq &
  \sqrt{\sum_{x,y \in [n]} \alpha_{xy}^2 \sum_{y,z \in [n]} \beta_{yz}^2
    \sum_{z,x \in [n]} \gamma_{zx}^2} \\
 \kern -2em \sum_{\kern +2em x,y,z,w \in [n]}\kern -2.2em  \alpha_{xy}\cdot \beta_{yz} \cdot \gamma_{zw} \leq &
  \sum_{x,y \in [n]} \alpha_{xy}\,\cdot\, \max_{y,z \in [n]}
    \beta_{yz}\,\cdot\,\sum_{z,w \in [n]} \gamma_{zw}
\end{align*}
where we used the fact that $\lim_{u\rightarrow 0} (\sum
\beta_{yz}^{\frac{1}{u}})^u = \max \beta_{yz}$.

Friedgut's inequalities immediately imply a well known result
developed in a series of
papers~\cite{DBLP:conf/soda/GroheM06,DBLP:conf/focs/AtseriasGM08,DBLP:conf/pods/NgoPRR12}
that gives an upper bound on the size of a query answer as a function
on the cardinality of the relations.  For example in the case of
$C_3$, consider an instance $S_1, S_2, S_3$, and set $\alpha_{xy} = 1$ if
$(x,y) \in S_1$, otherwise $\alpha_{xy}=0$ (and similarly for
$\beta_{yz},\gamma_{zx}$).  We obtain then $ |C_3| \leq \sqrt{|S_1| \cdot |S_2|
  \cdot |S_3|}$.  Note that all these results are expressed in terms
of a fractional edge {\em cover}.  When we apply Friedgut's inequality
in Section~\ref{sec:onestep} to a fractional edge {\em packing}, we 
ensure that the packing is tight.

\section{One Communication Step without Skew}
\label{sec:onestep}

In this section, we consider the case where the data has no skew, and
the computation is restricted to a single communication round.

We will say that a database is a {\em matching database} if each relation 
has degree bounded by 1 (\ie the frequency of each value is exactly 1 for each
relation). Our lower bounds in this section will hold for such matching
databases. The upper bound, and in particular the load analysis for the
algorithm, hold not only for matching databases, but in general for
databases with a small amount of skew, which we will formally define
in~\autoref{subsec:algorithms}.

We assume that all input servers know the cardinalities $m_1, \ldots,
m_\ell$ of the relations $S_1, \ldots, S_\ell$.  We denote
$\mathbf{m}= (m_1, \ldots, m_\ell)$ the vector of cardinalities, and
$\bM = (M_1, \ldots, M_\ell)$ the vector of the sizes expressed in
bits, where $M_j = a_j m_j \log n$, and $n$ is the size of the domain
of each attribute.

\subsection{The HyperCube Algorithm}
\label{subsec:algorithms}

We describe here an algorithm that computes a conjunctive query in one
step.  Such an algorithm was introduced by Afrati and
Ullman~\cite{DBLP:conf/edbt/AfratiU10} for MapReduce, is similar to an
algorithm by Suri and Vassilvitskii~\cite{DBLP:conf/www/SuriV11} to
count triangles, and also uses ideas that can be traced back to
Ganguly~\cite{DBLP:journals/jlp/GangulyST92} for parallel processing
of Datalog programs.  We call this the HyperCube (HC) algorithm,
following~\cite{BKS13}.

The HC algorithm initially assigns to each variable $x_i$, where $i=1,
\dots, k$, a {\em share} $p_i$, such that $\prod_{i=1}^k p_i =
p$. Each server is then represented by a distinct point $\mathbf{y}
\in \mathcal{P}$, where $\mathcal{P} = [p_1] \times \dots \times
[p_2]$; in other words, servers are mapped into a $k$-dimensional
hypercube. The HC algorithm then uses $k$ independently chosen hash
functions $h_i: [n] \rightarrow [p_i]$ and sends each tuple $t$ of
relation $S_j$ to all servers in the destination subcube of $t$:
\begin{align}
 \mathcal{D}(t) = \setof{\mathbf{y} \in \mathcal{P}}
 {\forall m = 1,\dots, a_j : h_{i_m}(t[i_m]) = \mathbf{y}_{i_m}}
\end{align}
During the computation phase, each server locally computes the query
$q$ for the subset of the input that it has received.  The correctness
of the HC algorithm follows from the observation that, for every
potential tuple $(a_1, \dots, a_k)$, the server $(h_1(a_1), \dots,
h_k(a_k))$ contains all the necessary information to decide whether it
belongs in the answer or not.

\begin{example}
  We illustrate how to compute the triangle query $C_3(x_1,x_2,x_3) =
  S_1(x_1, x_2),S_2(x_2, x_3),S_3(x_3, x_1)$.  Consider the shares
  $p_1=p_2=p_3=p^{1/3}$. Each of the $p$ servers is uniquely
  identified by a triple $(y_1,y_2,y_3)$, where $y_1, y_2, y_3\in
  [p^{1/3}]$.  In the first communication round, the input server
  storing $S_1$ sends each tuple $S_1(\alpha_1,\alpha_2)$ to all
  servers with index $(h_1(\alpha_1), h_2(\alpha_2), y_3)$, for all
  $y_3 \in [p^{1/3}]$: notice that each tuple is replicated $p^{1/3}$
  times.  The input servers holding $S_2$ and $S_3$ proceed similarly
  with their tuples.  After round 1, any three tuples $S_1(\alpha_1,
  \alpha_2)$, $S_2(\alpha_2, \alpha_3)$, $S_3(\alpha_3, \alpha_1)$
  that contribute to the output tuple
  $C_3(\alpha_1,\alpha_2,\alpha_3)$ will be seen by the server
  $\mathbf{y} = (h_1(\alpha_1), h_2(\alpha_2), h_3(\alpha_3))$: any
  server that detects three matching tuples outputs them.
\end{example}

\paragraph{Analysis of the HC algorithm} 
Let $R$ be a relation of arity $r$.
For a tuple $J$ over a subset of the 
attributes $[r]$, define $d_{J}(R) =  |\sigma_{J}(R)|$  as the degree
of the tuple $J$ in relation $R$. A matching database restricts the
degrees such that for every tuple $t$ over $U \subseteq [r]$, we 
have $d_{J}(R) =1$.
Our first analysis of the HC algorithm in~\cite{BKS13} was only for the special case 
of {\em matching databases}. 

Here, we analyze the behavior of the HC algorithm for larger degrees. 
Our analysis is based on the following lemma about hashing, which we 
prove in detail in~\autoref{sec:hashing}. 

\begin{lemma}
\label{lemma:hashing} 
  Let $R(A_1, \dots, A_r)$ be a relation of arity $r$ of size $m$. 
  Let $p_1, \ldots, p_r$ be integers and let $p = \prod_i
  p_i$. Suppose
  that we hash each tuple $(a_1, \ldots, a_r)$ to the bin
  $(h_1(a_1), \ldots, h_r(a_r))$, where $h_1, \ldots, h_r$ are
  independent and perfectly random hash functions. Then:
  \begin{packed_enum}
  \item The expected load in every bin is $m/p$.
  \item Suppose that for every tuple $J$ over $U \subseteq [r]$ we have
 $d_{J}(R) \leq  \frac{\beta^{|U|} m}{\prod_{i \in U}p_i}$ for some constant $\beta > 0$.
  Then the probability that the maximum load exceeds $O(m/p)$
  is exponentially small in $p$.
   \end{packed_enum}
\end{lemma}

Using the above lemma, we can now prove the following statement on the behavior
of the HC algorithm.

\begin{corollary}
\label{cor:hashing}
Let $\mathbf{p} = (p_1, \dots, p_k)$ be the shares of the HC algorithm.
Suppose that for every relation $S_j$ and every tuple $J$ over $U \subseteq [a_j]$ we have
 $d_{J}(S_j) \leq  \frac{\beta^{|U|} m_j}{\prod_{i \in U}p_i}$ for some constant $\beta > 0$.
Then with high probability the maximum load per server  is
  $$O \left( \max_j \frac{M_j}{\prod_{i: i \in S_j} p_i} \right)$$ 
\end{corollary}

\paragraph{Choosing the Shares} Here we discuss how to compute the
shares $p_i$ to optimize the expected load per server.  Afrati and
Ullman compute the shares by optimizing the total load $\sum_j
m_j/\prod_{i: i \in S_j} p_i$ subject to the constraint $\prod_i p_i =
1$, which is a non-linear system that can be solved using Lagrange
multipliers.  Our approach is to optimize the maximum load {\em per
  relation}, $L = \max_j m_j/\prod_{i: i \in S_j} p_i$; the total load
per server is $\leq \ell L$.  This leads to a linear optimization
problem, as follows.  First, write the shares as $p_i = p^{e_i}$ where
$e_i \in [0,1]$ is called the {\em share exponent} for $x_i$, denote
$\lambda = \log_p L$ and $\mu_j = \log_p M_j$ (we will assume
w.l.o.g. that $M_j \geq p$, hence $\mu_j \geq 1$ for all $j$).  Then,
we optimize the LP:
\begin{align}
\text{minimize}   \quad & \lambda \nonumber \\
\text{subject to} \quad
                  &  \sum_{i \in [k]} -e_i \geq -1 \nonumber \\
\quad \forall j \in [\ell]:  & \sum_{i \in S_j} e_i + \lambda \geq \mu_j \nonumber \\
\quad \forall i \in [k]: & e_i \geq 0, \quad \lambda \geq 0 \label{eq:primal:lp}
\end{align}

\begin{theorem}[Upper Bound]
  For a query $q$ and $p$ servers, with statistics $\bM$, let $\be =
  (e_1, \dots, e_k)$ be the optimal solution to \eqref{eq:primal:lp}
  and $e^*$ its objective value.
 
Let $p_i = p^{e_i}$ and  suppose that for every relation $S_j$ and every 
tuple $J$ over $U \subseteq [a_j]$ 
we have $d_{J}(S_j) \leq  \frac{\beta^{|U|} m_j}{\prod_{i \in U}p_i}$ for some constant $\beta > 0$. Then the HC algorithm with shares $p_i$ achieves $O(\lup)$
  maximum load with high probability, where $\lup= p^{e^*}$.
\end{theorem}

A special case of interest is when all cardinalities $M_j$ are equal,
therefore $\mu_1 = \ldots = \mu_\ell = \mu$.  In that case, the
optimal solution to Eq.(\ref{eq:primal:lp}) can be obtained from an
optimal fractional vertex cover $\mathbf{v^*}=(v_1^*, \ldots, v_k^*)$
by setting $e_i = v_i^* / \tau^*$ (where $\tau^* = \sum_i v_i^*$).  To
see this, we note that any feasible solution
$(\lambda,e_1,\ldots,e_k)$ to Eq.(\ref{eq:primal:lp}) defines the
vertex cover $v_i = e_i/(\mu-\lambda)$, and in the opposite direction
every vertex cover defines the feasible solution $e_i = v_i/(\sum_i
v_i)$, $\lambda = \mu - 1/(\sum_i v_i)$; further more, minimizing
$\lambda$ is equivalent to minimizing $\sum_i v_i$.  Thus, when all
cardinalities are equal to $M$, at optimality $\lambda^* = \mu -
1/\tau^*$, and $\lup = M/p^{1/\tau^*}$.

We illustrate more examples in~\autoref{subsec:lower:eq:upper}.

\subsection{The Lower Bound}
\label{subsec:oneround:lower}

In this section, we prove a lower bound on the maximum load per server
over databases with statistics $\mathbf{M}$.

Fix a query $q$ and a fractional edge packing $\bu$ of $q$.  Denote:
\begin{align}
  L(\bu,\bM,p) = & \left(\frac{\prod_{j=1}^{\ell} M_j^{u_j} }{p} \right)^{1/\sum_j u_j} \label{eq:lump}
\end{align}
Further denote $\llo = \max_{\bu} L(\bu,\bM,p)$, where $\bu$ ranges
over all edge packings for $q$.  In this section, we will prove that
Eq.(\ref{eq:lump}) is a lower bound for the load of any algorithm
computing the query $q$, over a database with statistics $\bM$.  We
will prove in ~\autoref{subsec:lower:eq:upper} that $\llo=\lup$,
showing that the upper bound and lower bound are tight.  To gain some
intuition behind the formula (\ref{eq:lump}), consider the case when
all cardinalities are equal, $M_1 = \ldots = M_\ell = M$.  Then
${\llo}= M/p^{1/\sum_j u_j}$, and this quantity is maximized when
$\bu$ is a maximum fractional edge packing, whose value is $\tau^*$,
the fractional vertex covering number for $q$.  Thus, $\llo =
M/p^{1/\tau^*}$, which is the same expression as $\lup$.

To prove the lower bound, we will define a probability space from
which the input databases are drawn.  Notice that the cardinalities of
the $\ell$ relations are fixed: $m_1, \ldots, m_\ell$.  We first
choose a domain size $n \geq \max_j m_j$, to be specified later, and
choose independently and uniformly each relation $S_j$ from all
matchings of $[n]^{a_j}$ with exactly $m_j$ tuples. We call this the
{\em matching probability space}.  Observe that the probability space
contains only databases with relations without skew (in fact all
degrees are exactly 1).  We write $\E[|q(I)|]$ for the expected number
of answers to $q$ under the above probability space.

\begin{theorem}[Lower Bound] 
  \label{th:lower:uniform} 
  Fix statistics $\mathbf{m}$, and consider any deterministic MPC
  algorithm that runs in one communication round on $p$ servers. Let
  $\bu$ be any fractional edge packing of $q$.  If $s$ is any server
  and $L_s$ is its load, then server $s$ reports at most
  $$
  \frac{L_s^{\sum_j u_j}}{(\sum_j u_j/4)^{\sum_j u_j}
   \prod_{j=1}^{\ell} M_j^{u_j} } \cdot \E[|q(I)|]
  $$
  answers in expectation, where $I$ is a randomly chosen from the
  matching probability space with statistics $\mathbf{m}$ and domain
  size $n = (\max_j m_j)^2$.  Therefore, the $p$ servers of the
  algorithm report at most
   $$\left( \frac{4 L}{ (\sum_j u_j) \cdot L(\bu,\bM,p)} \right)^{\sum_j u_j} \cdot \E[|q(I)|]$$
   answers in expectation, where $L = \max_{s \in [p]} L_s$ is the
   maximum load of all servers.

   Furthermore, if if all relations have equal size $m_1 = \ldots =
   m_\ell = m$ and arity $a_j \geq 2$, then one can choose $n = m$,
   and strengthen the number of answers reported by the $p$ servers
   to:
   $$\left( \frac{L}{ (\sum_j u_j) \cdot L(\bu,\bM,p)} \right)^{\sum_j u_j} \cdot \E[|q(I)|]$$
\end{theorem}

Therefore, if $\bu$ is any fractional edge packing, then $(\sum_j u_j)
\cdot L(\bu,\bM,p) / 4$ is a lower bound for the load of any algorithm
computing $q$.  Up to a constant factor, the strongest such lower
bound is given by $\bu^*$, the optimal solution for
Eq.(\ref{eq:lump}), since for any $\bu$, we have $(\sum_j u_j) \cdot
L(\bu,\bM,p)/4 \leq [(\sum_j u_j) / (\sum_j u_j^*)] \cdot [(\sum_j
u_j^*) \cdot \llo / 4]$, and $(\sum_j u_j) / (\sum_j u_j^*) \leq
\tau^* = O(1)$ (since $\sum_j u_j \leq \tau^*$ and, at optimality,
$\sum_j u_j^* \geq 1$).

Before we prove the theorem, we show how to extend it to a lower bound
for any randomized algorithm.  For this, we start with a lemma that we
also need later.

\begin{lemma} 
\label{lem:expected_size} 
The expected number of answers to $q$ is $\E[|q(I)|] = n^{k-a}
\prod_{j=1}^{\ell} m_j$.  In particular, if $n = m_1 = \cdots =
m_\ell$ then $\E[|q(I)|] = n^{c-\chi(q)}$, where $c$ is the number of
connected components of $q$.
\end{lemma}

\begin{proof}
  For any relation $S_j$, and any tuple $\mathbf{a}_j \in [n]^{a_j}$, 
  the probability that $S_j$ contains $\mathbf{a}_j$
  is $P(\mathbf{a_j} \in S_j) = m_j / n^{a_j}$.  Given a tuple $\mathbf{a}
  \in [n]^k$ of the same arity as the query answer, let
  $\mathbf{a}_j$ denote its projection on the variables in $S_j$.  Then:
  \begin{align*}
    \E[|q(I)|]  
    &  = \sum_{\ba \in [n]^k}  P(\bigwedge_{j=1}^{\ell} (\mathbf{a}_j \in S_j)) 
     = \sum_{\ba \in [n]^k} \prod_{j=1}^{\ell} P(\mathbf{a}_j \in S_j)  \\
    & = \sum_{\ba \in [n]^k} \prod_{j=1}^{\ell} m_j n^{-a_j} 
    = n^{k-a} \prod_{j=1}^{\ell} m_j
\end{align*}
\end{proof}

We now can prove a lower bound for the maximum load of any randomized 
algorithm, on a fixed database instance.

\begin{theorem}
\label{thm:probability-LB}
  Consider a connected query $q$ with fractional vertex covering
  number $\tau^*$.  Fix some database statistics $\bM$.  Let $A$ be
  any one round, randomized \mpc\ algorithm $A$ for $q$, with maximum
  load $L \leq \delta \cdot \llo$, for some constant $\delta < 1/(4
  \cdot 9^{\tau^*})$.  Then there exists an instance $I$ such that the
  randomized algorithm $A$ fails to compute $q(I)$ correctly with
  probability $>1 - 9(4\delta)^{1/\tau^*} = \Omega(1)$
\end{theorem}

\begin{proof}
  We use Yao's lemma, which in our setting says the following.
  Consider any probability space for database instances $I$. If any
  deterministic algorithm fails with probability $\geq 1-\delta$ to
  compute $q(I)$ correctly, over random inputs $I$, then there exists
  an instance $I$ such that every randomized algorithm $A$ fails with
  probability $\geq 1-\delta$ to compute $q(I)$ correctly over the
  random choices of $A$.  To apply the lemma, we need to choose the right
  probability space over database instances $I$.  The space of random
  matchings is not useful for this purpose, because for a connected
  query with a large characteristic $\chi(q)$, $\E[|q(I)|] = O(1/n)$
  and therefore $P(q(I) \neq \emptyset) = O(1/n)$, which means that a
  naive deterministic algorithm that always returns the empty answer
  with fail with a very small probably, $O(1/n)$.  Instead, denoting
  $\mu = \E[|q(I)|]$, we define $C_\alpha$ the event $|q(I)| > \alpha
  \mu$, where $\alpha > 1$ is some constant.  We will apply Yao's
  lemma to the probability space of random matchings conditioned on
  $C_\alpha$.  

  We prove that, for $\alpha=1/3$, any deterministic algorithm $A$
  fails to compute $q(I)$ correctly with probability $\geq
  1-9(4\delta)^{1/\tau^*}$, over random matchings conditioned on
  $C_{1/3}$.  Let $\bu^*$ be an edge packing that maximizes
  $L(\bu,\bM,p)$, and denote
  $f = \left( \frac{4 L}{(\sum_j u_j^*) \cdot \llo} \right)^{\sum_j
    u_j^*}$. \autoref{lem:expected_size} implies that that, for any
  one-round deterministic algorithm with load $\leq L$, $\E[|A(I)|]
  \leq f \E[|q(I)|]$.  We prove the following in \autoref{sec:prob-bounds}:

  \begin{lemma} \label{lem:instance:choice}
    If $A$ is a deterministic algorithm for $q$ (more precisely:
    $\forall I, A(I) \subseteq q(I)$), such that, over random
    matchings, $\E[|A(I)|] \leq f \E[|q(I)|]$ for some constant $f <
    1/9$, then, denoting $\texttt{fail}$ the event $A(I) \neq q(I)$,
    we have
    \begin{align*}
      P(\texttt{fail} | C_{1/3}) \geq & 1 - 9f
    \end{align*}
  \end{lemma}

  The proof of the theorem follows from Yao's lemma and the fact that $f
  \leq \left(\frac{4\delta}{\sum_j u_j^*}\right)^{1/\sum_j u_j^*} \leq
  (4\delta)^{1/\tau^*}$ because $\sum_j u_j^* \leq \tau^*$ and, at
  optimality, $\sum_j u_j^* \geq 1$.
\end{proof}

In the rest of this section, we give the proof of~\autoref{th:lower:uniform}. 

Let us fix some server $s \in [p]$, and let $\msg(I)$ denote the
function specifying the message the server receives on input $I$.
Recall that, in the input-sever model, each input relation $S_j$ is
stored at a separate input server, and therefore the message received
by $s$ consists of $\ell$ separate message $\msg_j = \msg_j(S_j)$, for
each $j=1,\dots, \ell$.  One should think of $\msg_j$ is a bit string.  Once
the server $s$ receives $\msg_j$ it ``knows'' that the input relation
$S_j$ is in the set $\setof{S_j}{\msg_j(S_j) = \msg_j}$.
%
%
%
This justifies the following definition: given a message $\msg_j$,
{\em the set of tuples known by the server} is:
\begin{align*}
  K_{\msg_j}(S_j)=\setof{t\in [n]^{a_j}}{\mbox{ for all instances }S_j\subseteq[n]^{a_j}, \msg_j(S_j)=\msg_j\Rightarrow t\in S_j}
\end{align*}
where $a_j$ is the arity of $S_j$. 

Clearly, an algorithm $A$ may output a tuple $\ba \in [n]^k$ as answer
to the query $q$ iff, for every $j$, $\ba_j \in K_{\msg_j}(S_j)$ for
all $j=1,\dots, \ell$, where $\ba_j$ denotes the projection of $\ba$ on the
variables in the atom $S_j$.

We will first prove an upper bound for each $|K_{\msg_j}(S_j)|$ in
Section~\ref{subsec:relation_bound}.  Then in
Section~\ref{subsec:onestepB} we use this bound, along with Friedgut's
inequality, to establish an upper bound for $|K_\msg(q)|$ and hence
prove \autoref{th:lower:uniform}.

\subsubsection{Bounding the Knowledge of Each Relation}
\label{subsec:relation_bound}

Let us fix a server $s$, and an input relation $S_j$.  Recall that
$M_j = m_j\log n$ denotes the number of bits necessary to encode
$S_j$.  An algorithm $A$ may use few bits, $\mM_j$, by exploiting the
fact that $S_j$ is a uniformly chosen $a_j$-dimensional matching.
There are precisely $\binom{n}{m_j}^{a_j} (m_j !)^{a_j-1}$ different
$a_j$-dimensional matchings of arity $a_j$ and size $m_j$ and thus the
number of bits $N$ necessary to represent the relation is given by the
entropy:
\begin{align}
  \mM_j = H(S_j) = a_j \log \binom{n}{m_j} + (a_j-1) \log (m_j!)  \label{eq:mm}
\end{align}
We will prove later that $\mM_j = \Omega(M_j)$.
The following lemma provides a bound on the expected knowledge
$K_{m_j}(S_j)$ the server may obtain from $S_j$:

\begin{lemma} \label{lem:entropy_ratio} Suppose that the size of $S_j$
  is $m_j \leq n/2$ (or $m_j=n$), and that the message $\msg_j(S_j)$
  has at most $f_j \cdot \mM_j$ bits. Then $\E[|K_{\msg_j}(S_j)|] \le
  2 f_j \cdot m_j$ (or $\leq f_j \cdot m_j$), where the expectation is
  taken over random choices of the matching $S_j$.
\end{lemma}

It says that, if the message $\msg_j$ has only a fraction $f_j$ of the
bits needed to encode $S_j$, then a server receiving this message
knows, in expectation, only a fraction $2f_j$ of the $m_j$ tuples in
$S_j$.  Notice that the bound holds only in expectation: a specialized
encoding may choose to use very few bits to represent a particular
matching $S_j \subseteq [n]^{a_j}$: when a server receives that
message, then it knows all tuples in $S_j$, however then there will be
fewer bit combinations left to encode the other matchings $S_j$.

\begin{proof} 
  The entropy $H(S_j)$ in Eq.(\ref{eq:mm}) has two parts,
  corresponding to the two parts needed to encode $S_j$: for each
  attribute of $S_j$ we need to encode a subset $\subseteq [n]$ of
  size $m_j$, and for each attribute except one we need to encode a
  permutation over $[m_j]$.  Fix a value $\msg_j$ of the message
  received by the server from the input $S_j$, and let $k =
  |K_{\msg_j}(S_j)|$. Since $\msg_j$ fixes precisely $k$ tuples of
  $S_j$, the conditional entropy $H(S_j | \msg_j(S_j) = \msg_j)$ is:
  \begin{align}
    \log|\setof{S_j}{\msg_j(S_j)=\msg_j}| & \le a_j \log  \binom{n-k}{m_j - k} + (a_j-1) \log ((m_j-k))!) \nonumber
  \end{align}
  We will next show that
  \begin{align}
    \log|\setof{S_j}{\msg_j(S_j)=\msg_j}| \le \left(1- \frac{k}{2m_j} \right)  \mM_j \label{eq:ratiobound}
  \end{align}
In other words, we claim that the entropy has decreased by at least a
fraction $k/(2m_j)$.  We show this by proving that each of the two
parts of the entropy decreased by that amount:

\begin{proposition}
$\log \left( (m-k)! \right) \leq \left(1-\frac{k}{m} \right) \log (m!) $
\end{proposition}
 
\begin{proof}
Since $\log(x)$ is an increasing function, it holds that
$\sum_{i=1}^{m-k} (\log(i)/(m-k)) \leq \sum_{i=1}^m (\log(i)/m)$, which is
equivalent to:
$$ \frac{\log ((m-k)!)}{\log (m!)} \leq \frac{m-k}{m}$$
This proves the claim.
\end{proof}

\begin{proposition} \label{lem:binom_inequality}
For any $k \leq m \leq n/2$, or $k \leq m = n$:
\begin{align*}
\log \binom{n-k}{m-k} \leq \left(1-\frac{k}{2m} \right) \log \binom{n}{m}
\end{align*} 
\end{proposition}
  
 \begin{proof}
   If $m=n$ then the claim holds trivially because both sides are 0,
   so we assume $m \leq n/2$.  We have:
$$ \frac{ \binom{n-k}{m-k}}{\binom{n}{m}} = \frac{m \cdot (m-1) \cdots (m-k+1)}{n \cdot (n-1) \cdots (n-k+1)} \leq \left(\frac{m}{n} \right)^{k}$$
and therefore:
$$
\log \binom{n-k}{m-k} \leq \log \binom{n}{m} - k\log(n/m) =
\left(1-\frac{k \log(n/m)}{\log \binom{n}{m}} \right) \log \binom{n}{m}
$$

To conclude the proof, it suffices to show that $ \log \binom{n}{m}
\leq 2m \log(n/m)$. For this, we use the bound $\log \binom{n}{m} \leq
n H(m/n)$, where $H(x) = -x\log(x) - (1-x)\log(1-x)$ is the {\em
  binary entropy}.  Denote $f(x) = -x\log(x)$, therefore $H(x) = f(x)
+ f(1-x)$.  Then we have $f(x) \leq f(1-x)$ for $x \leq 1/2$, because
the function $g(x) = f(x) - f(1-x)$ is concave (by direct calculation,
$g''(x) = -1/x + 1/(1-x) \leq 0$ for $x \in [0,1/2]$), and
$g(0)=g(1/2) = 0$, meaning that $g(x) \geq 0$ on the interval $x \in
[0,1/2]$.  Therefore, $H(x) \leq 2f(x)$, and our claim follows from:
\begin{align*}
  \log \binom{n}{m} \leq n H(m/n) \leq 2 n f(m/n) = 2 m \log (n/m)
\end{align*}
This concludes the proof of \autoref{lem:binom_inequality}.
\end{proof} 
  
Now we will use Eq.(\ref{eq:ratiobound}) to complete the proof of
\autoref{lem:entropy_ratio}.  We apply the chain rule for entropy,
$H(S_j,\msg_j(S_j)) = H(\msg_j(S_j)) + H(S_j | \msg_j(S_j))$, then use
the fact that $H(S_j,\msg_j(S_j)) = H(S_j)$ (since $S_j$ completely
determines $\msg_j(S_j)$) and apply the definition of $H(S_j |
\msg_j(S_j))$:
  \begin{align}
   H(S_j) 
  &=H(\msg_j(S_j))+ \sum_{\msg_j}P(\msg_j(S_j)=\msg_j)\cdot H(S_j|\msg_j(S_j)=\msg_j)\nonumber\\
  &\le f_j \cdot  H(S_j) +  \sum_{\msg_j}P(\msg_j(S_j)=\msg_j)\cdot H(S_j|\msg_j(S_j)=\msg_j)&
  \nonumber\\
  &\le f_j \cdot  H(S_j)+
      \sum_{\msg_j}P(\msg_j(S_j)=\msg_j)\cdot(1-\frac{|K_{\msg_j}(S_j)|}{ 2 m_j}) H(S_j)&
      \nonumber\\
  &= f_j \cdot  H(S_j)+ 
      (1-\sum_{\msg_j}P(\msg_j(S_j)=\msg_j)\ \frac{ |K_{\msg_j}(S_j)|}{2 m_j}) H(S_j)\nonumber\\
  &= f_j \cdot H(S_j)+  (1-\frac{ \E [|K_{\msg_j(S_j)}(S_j)|]}{2 m_j}) H(S_j) \label{eq:entropy}
  \end{align}
  where the first inequality follows from the assumed upper bound on
  $|\msg_j(S_j)|$, the second inequality follows by
  \eqref{eq:ratiobound}, and the last two lines follow by definition.
  Dividing both sides of \eqref{eq:entropy} by $H(S_j)$ since $H(S_j)$
  is not zero and rearranging we obtain the required statement.
  \end{proof}

\subsubsection{Bounding the Knowledge of the Query}
\label{subsec:onestepB}

We use now \autoref{lem:entropy_ratio} to derive an upper bound on the
number of answers to $q(I)$ that a server $s$ can report.  Recall that
\autoref{lem:entropy_ratio} assumed that the message $\msg_j(S_j)$ is
at most a fraction $f_j$ of the entropy of $S_j$.  We do not know the
values of $f_j$, instead we know that the entire $\msg(I)$ received by
the server $s$ (the concatenation of all $\ell$ messages
$\msg_j(S_j)$) has at most $L$ bits.  For each relation $S_j$, define
$$f_j = \frac{\max_{S_j\subseteq[n]^{a_j}} |\msg_j(S_j)|}{\mM_j} .$$
Thus, $f_j$ is the largest fraction of bits of $S_j$ that the server
receives, over all choices of the matching $S_j$.
We immediately derive an upper bound on the $f_j$'s.  We have
$\sum_{j=1}^{\ell} \max_{S_j} |\msg_j(S_j)| \leq L$, because each
relation $S_j$ can be chosen independently, which implies
$\sum_{j=1}^{\ell} f_j \mM_j \leq L$.

For $\ba_j \in [n]^{a_j}$, let $w_j({\ba_j})$ denote the probability
that the server knows the tuple $\ba_j$.  In other words $w_j(\ba_j) =
P(\ba_j \in K_{\msg_j(S_j)}(S_j))$, where the probability is over the
random choices of $S_j$.

\begin{lemma} \label{lem:bounds} 
For any relation $S_j$:
\begin{itemize}
{\setlength\itemindent{25pt} \item[(a)] $\forall \ba_j \in [n]^{a_j}: w_j(\ba_j) \leq m_j/n^{a_j}$, and}
{\setlength\itemindent{25pt}\item[(b)] $\sum_{\ba_j \in [n]^{a_j}} w_j({\ba_j}) \leq 2 f_j \cdot m_j$.}
\end{itemize}
\end{lemma}

\begin{proof}
  To show (a), notice that $w_j(\ba_j) \leq P(\ba_j \in S_j) = m_j/ n^{a_j}$,
  while (b) follows from the fact $\sum_{\ba_j \in [n]^{a_j}} w_j({\ba_j}) =
  \mathbf{E}[|K_{\msg_j(S_j)}(S_j)|] \leq 2 f_j \cdot m_j$ (\autoref{lem:entropy_ratio} ).
\end{proof}

Since the server receives a separate message for each relation
$S_j$, from a distinct input server, the events $\ba_1 \in K_{\msg_1}(S_1),
\ldots, \ba_\ell \in K_{\msg_\ell}(S_\ell)$ are independent, hence:
\begin{align*}
  \mathbf{E}[|K_{\msg(I)}(q)|] = \sum_{\ba \in [n]^k} P(\ba \in K_{\msg(I)}(q)) = \sum_{\ba \in [n]^k} \prod_{j=1}^{\ell}
  w_j({\ba_j})
\end{align*}
We now use  Friedgut's
inequality. Recall that in order to apply the inequality, we need to 
find a fractional edge cover.  Let us pick any fractional edge packing
$\mathbf{u} = (u_1, \ldots, u_\ell)$.  Given $q$, defined as in \eqref{eq:q},
 consider the {\em extended query}, which has a new unary
 atom for each variable $x_i$:
\begin{align*}
  q'(x_1,\ldots, x_k) = S_1(\bar x_1), \ldots, S_\ell(\bar x_\ell),
  T_1(x_1), \ldots, T_k(x_k)
\end{align*}
For each new symbol $T_i$, define 
$  u_i' = 1 - \sum_{j: x_i \in \text{vars}(S_j)} u_j$.
Since $\mathbf{u}$ is a packing,  $u_i' \geq 0$. 
Let us define $\mathbf{u}'=(u_1', \ldots, u_k')$.

\begin{lemma}
(a) The assignment $(\mathbf{u},\mathbf{u}')$ is both a tight
fractional edge packing and a tight fractional edge cover for $q'$. 
(b) $\sum_{j=1}^\ell a_j u_j + \sum_{i=1}^k u'_i = k$
\end{lemma}

\begin{proof}
  (a) is straightforward, since for every variable $x_i$ we have $u_i'+\sum_{j:
    x_i \in \text{vars}(S_j)} u_j = 1$.  Summing up:
  \begin{align*}
k=\sum_{i=1}^k \left( u_i'+\sum_{j: x_i \in \text{vars}(S_j)} u_j \right) = 
\sum_{i=1}^k u_i' + \sum_{j=1}^\ell a_j u_j 
  \end{align*}
which proves (b).  
\end{proof}

We will apply Friedgut's inequality to the extended query $q'$. 
 Set the variables $w(-)$ used in Friedgut's
inequality as follows:
\begin{align*}
  w_j(\ba_j) = & \P(\ba_j \in K_{\msg_j(S_j)}(S_j)) \mbox{ for $S_j$,
    tuple $\ba_j \in [n]^{a_j}$} \\
  w_i'(\alpha) = & 1\kern 1.1in  \mbox{ for $T_i$, value $\alpha \in [n]$}
\end{align*}

Recall that, for a tuple $\ba \in [n]^k$ we use $\ba_j \in [n]^{a_j}$
for its projection on the variables in $S_j$; with some abuse, we
write $\ba_i \in [n]$ for the projection on the variable $x_i$.
Assume first that $u_j>0$, for $j=1,\dots, \ell$.  Then:
\begin{align*}
\allowdisplaybreaks
\mathbf{E}[|K_\msg(q)|] 
& =  \sum_{\ba \in [n]^k} \prod_{j=1}^{\ell} w_j({\ba_j}) \\
& = \sum_{\ba \in [n]^k} \prod_{j=1}^{\ell} w_j({\ba_j})\prod_{i=1}^k w_i'({\ba_i})\\
&\leq \prod_{j=1}^{\ell} \left( \sum_{\ba \in [n]^{a_j}} w_j({\ba})^{1/u_j} \right)^{u_j}
\prod_{i=1}^{k} \left( \sum_{\alpha \in [n]} w'_i(\alpha)^{1/u_i'} \right)^{u_i'} \\
 &= \prod_{j=1}^{\ell} \left( \sum_{\ba \in [n]^{a_j}} w_j({\ba})^{1/u_j} \right)^{u_j} \prod_{i=1}^k n^{u_i'}
\end{align*}
Note that, since $w'_i(\alpha) = 1$ we have $w'_i(\alpha)^{1/u_i'} =
1$ even if $u_i'=0$.  Write $w_j({\ba})^{1/u_j} = w_j({\ba})^{1/u_j-1}
w_j({\ba})$, and use~\autoref{lem:bounds} to obtain:
\begin{align*}
\sum_{\ba \in [n]^{a_j}} w_j({\ba})^{1/u_j} 
& \leq (m_j / n^{a_j})^{1/u_j-1} \sum_{\ba \in [n]^{a_j}} w_j({\ba}) \\
& \leq (m_j  n^{-a_j})^{1/u_j-1}  2 f_j \cdot m_j \\
&  = 2 f_j \cdot m_j^{1/u_j} \cdot n^{( a_j -a_j/u_j)}
\end{align*}
Plugging this in the bound, we have shown that:
\begin{align}
\E[|K_\msg(q)|]& \leq \prod_{j=1}^{\ell} (2 f_j \cdot m_j^{1/u_j} \cdot n^{( a_j -a_j/u_j)})^{u_j} \cdot \prod_{i=1}^k n^{u_i'} \nonumber\\
  &= \prod_{j=1}^{\ell } (2f_j)^{u_j} \cdot \prod_{j=1}^{\ell} m_j \cdot n^{(\sum_{j=1}^{\ell}a_j u_j - a)} \cdot n^{\sum_{i=1}^k u_i'} \nonumber\\
  &= \prod_{j=1}^{\ell } (2f_j)^{u_j} \cdot \prod_{j=1}^{\ell} m_j \cdot  n^{-a +(\sum_{j=1}^{\ell}a_j u_j+ \sum_{i=1}^k u_i')} \nonumber\\
  &= \prod_{j=1}^{\ell } (2f_j)^{u_j} \cdot \prod_{j=1}^{\ell} m_j \cdot  n^{k-a} \nonumber\\
   & = \prod_{j=1}^{\ell } (2f_j)^{u_j} \cdot  \E[|q(I)|]  \label{eq:end}
\end{align}
If some $u_j=0$, then we can derive the same lower bound as follows:
We can replace each $u_j$ with $u_j+\delta$ for any $\delta>0$ still
yielding an edge cover. 
Then we have $\sum_j a_j u_j + \sum_i u_i' = k + a\delta$,
and hence an extra factor $n^{a\delta}$ multiplying the term
$n^{\ell+k-a}$ in \eqref{eq:end}; however, we
obtain the same upper bound since, in the limit as $\delta$ approaches 0,
this extra factor approaches 1.

Let $f_q = \prod_{j=1}^{\ell } (2f_j)^{u_j}$; the final step is to upper bound the quantity $f_q$
using the fact that $\sum_{j=1}^{\ell} f_j \mM_j \leq L$. Recall that $u = \sum_j u_j$, then:
\begin{align*}
 f_q  & = \prod_{j=1}^{\ell} (2f_j)^{u_j} 
  = \prod_{j=1}^{\ell} \left( \frac{f_j \mM_j}{u_j} \right)^{u_j}
     \prod_{j=1}^{\ell} \left( \frac{2 u_j}{\mM_j} \right)^{u_j} \\
& \leq \left( \frac{\sum_{j=1}^{\ell} f_j \mM_j}{\sum_j u_j} \right)^{\sum_j u_j}
  \prod_{j=1}^{\ell} \left( \frac{2 u_j}{\mM_j} \right)^{u_j} \\
& \leq \left( \frac{L}{\sum_j u_j} \right)^{\sum_j u_j}
  \prod_{j=1}^{\ell} \left( \frac{2 u_j}{\mM_j} \right)^{u_j} \\
& = \prod_{j=1}^{\ell} \left( \frac{2 L}{u \cdot \mM_j} \right)^{ u_j}
  \prod_{j=1}^{\ell} \left( u_j \right)^{u_j} \\
& \leq  \prod_{j=1}^{\ell} \left( \frac{2 L}{u \cdot \mM_j} \right)^{ u_j}
\end{align*}
Here, the first inequality comes from the weighted version of the Arithmetic Mean-Geometric
Mean inequality. The last inequality holds since $u_j \leq 1$ for any $j$. 

Finally, we need a lower bound on the number of bits $\mM_j$ needed to represent relation
$S_j$. Indeed:

\begin{proposition}
\label{prop:bit-representation}
The number of bits $\mM_j$ needed to represent $S_j$ are:
\begin{itemize}
{\setlength\itemindent{25pt} \item[(a)] 
If $n \geq m_j^2$, then $\mM_j \leq M_j/2$}
{\setlength\itemindent{25pt}\item[(b)] 
If $n = m_j$ and $a_j \geq 2$, then  $\mM_j \leq M_j/4$}
\end{itemize}
\end{proposition}

\begin{proof}
For the first item, we have:
$$ \mM_j \geq a_j \log \binom{n}{m_j}  \geq a_j m_j \log (n /m_j)
\geq (1/2) a_j m_j \log(n) = M_j/2$$

For the second item, we have:
$$\mM_j \geq (a_j-1) \log(m_j !) \geq \frac{a_j-1}{2} m_j \log(m_j) 
\geq \frac{(a_j-1)}{2a_j}  M_j \geq M_j/4$$
where the last inequality comes from the assumption that $a_j \geq 2$.
%
\end{proof}

Applying the above bound on $\mM_j$, we complete the proof of~\autoref{th:lower:uniform}.
Recall that our $L$ denotes the load of an arbitrary server, which was denoted
$L_i$ in the statement of the theorem. 

%

\subsection{Proof of Equivalence}

\label{subsec:lower:eq:upper}

Let $pk(q)$ be the {\em extreme points} of the convex polytope defined by the fractional edge
packing constraints in~\eqref{eq:cover:dual}.  Recall that the
vertices of the polytope are feasible solutions $\bu_1, \bu_2, \ldots$,
with the property that every other feasible solution $\bu$ to the LP
is a convex combination of these vertices.  Each vertex can be
obtained by choosing $m$ out of the $k+\ell$ inequalities in
\eqref{eq:cover:dual}, transforming them into equalities, then
solving for $\bu$. 
Thus, it holds that $|pk(q)| \leq \binom{k+\ell}{m}$.  We
prove here:

\begin{theorem} \label{th:lump} For any vector of statistics $\bM$ and number of
  processors $p$ , we have:
  \begin{align*}
    \llo = \lup = \max_{\bu \in pk(q)}  L(\bu,\bM,p)
  \end{align*}
\end{theorem}

\begin{proof} 
Recall that $\lup =p^{e^*}$, 
where $e^*$ is the optimal solution to the {\em
    primal} LP problem \eqref{eq:primal:lp}.  Consider its {\em
    dual} LP:
\begin{align}
  \text{maximize}   \quad & \sum_{j \in [\ell]} \mu_j f_j - f \nonumber \\
  \text{subject to} \quad
  &  \sum_{j \in [\ell]} f_j \leq 1 \nonumber \\
  \quad \forall i \in [k]: & \sum_{j:i \in S_j} f_j -f \leq 0 \nonumber \\
  \quad \forall j \in [\ell]: & f_j \geq 0, \quad f \geq 0  \label{eq:dual:lp}
\end{align}

By the primal-dual theorem, its optimal solution is also $e^*$.
Writing $u_j = f_j/f$ and $u = 1/f$, we transform it into the
following non-linear optimization problem:
\begin{align}
  \text{maximize}   \quad &  \frac{1}{u} \cdot  \left( \sum_{j \in [\ell]}\mu_j u_j - 1 \right) \nonumber \\
  \text{subject to} \quad
  &  \sum_{j \in [\ell]} u_j \leq u \nonumber \\
  \quad \forall i \in [k]: & \sum_{j:i \in S_j} u_j \leq 1 \nonumber \\
  \quad \forall j \in [\ell]: & u_j \geq 0 \label{eq:dual:lp:2}
\end{align}

Consider optimizing the above non-linear problem.  Its optimal
solution must have $u = \sum_j u_j$, otherwise we simply replace $u$
with $\sum_j u_j$ and obtain a feasible solution with at least as good
objective function (indeed, $\mu_j \geq 1$ for any $j$, and hence 
$\sum_j \mu_j u_j \geq \sum_j u_j \geq 1$, since any optimal $\bu$ will have sum at least 1).
Therefore, the optimal is given by a fractional
edge packing $\bu$.  Furthermore, for any packing $\bu$, the objective
function $\sum_j \frac{1}{u} \cdot (\mu_j u_j - 1)$ is $\log_p
L(\bu, \bM, p)$.  To prove the theorem, we show that (a) $e^* =
u^*$ and (b) the optimum is obtained when $\bu \in pk(q)$.  This
follows from:

\begin{lemma} \label{lem:convex:mapping}
  Consider the function $F : \mathbf{R}^{k+1} \rightarrow
  \mathbf{R}^{k+1}$: $F(x_0, x_1, \ldots, x_k) = (1/x_0, x_1/x_0,
  \ldots, x_k/x_0)$.  Then:
  \begin{packed_item}
  \item $F$ is its own inverse, $F = F^{-1}$.
  \item $F$ maps any feasible solution to the system
    \eqref{eq:dual:lp} to a feasible solution to
    \eqref{eq:dual:lp:2}, and conversely.
  \item $F$ maps a convex set to a convex set.
  \end{packed_item}
\end{lemma}

\begin{proof}
  If $y_0 = 1/x_0$ and $y_j = x_j/x_0$, then obviously $x_0 = 1/y_0$
  and $x_j = y_j/y$.  The second item can be checked directly.  For
  the third item, it suffices to prove that $F$ maps a convex
  combination $\lambda \mathbf{x} + \lambda' \mathbf{x}'$ where
  $\lambda+\lambda'=1$ into a convex combination $\mu F(\mathbf{x}) +
  \mu' F(\mathbf{x}')$, where $\mu+\mu'=1$.  Assuming $\mathbf{x} =
  (x_0, x_1, \ldots, x_k)$ and $\mathbf{x}' = (x_0', x_1', \ldots,
  x_k')$, this follows by setting $\mu = x_0 / (\lambda x_0 + \lambda
  x_0')$ and $\mu' = x_0' / (\lambda x_0 + \lambda x_0')$.
\end{proof}

This completes the proof of \autoref{th:lump}. \qed
\end{proof}

\subsection{Discussion}
\label{sec:onestep-discuss}

We present here examples and applications of the theorems
proved in this section.  

\paragraph{The Speedup of the HyperCube}
Denote $\bu^*$ the fractional edge packing that maximizes
$L(\bu, \bM, p)$~\eqref{eq:lump}.  When the number of servers
increases, the load decreases at a rate of $1/p^{1/\sum_j u^*_j}$,
which we call the {\em speedup} of the HyperCube algorithm.  We call
the quantity $1/\sum_j u^*_j$ the {\em speedup exponent}. We have seen
that, when all cardinalities are equal, then the speedup exponent is
$1/\tau^*$, but when the cardinalities are unequal then the speedup
exponent may be better.

\begin{example}
\label{ex:triangle}
  Consider the triangle query 
  $$C_3 = S_1(x_1,x_2), S_2(x_2,x_3),S_3(x_3,x_1)$$ 
  and assume the relation sizes are $M_1, M_2, M_3$.  Then, $pk(C_3)$
  has five vertices, and each gives a different value for
  $L(\bu,\bM,p) = (M_1^{u_1}M_2^{u_2}M_3^{u_3}/p)^{1/(u_1+u_2+u_3)}$:
\begin{center}  
\[ \begin{array}{|c|c|}
\hline
  \bu & L(\bu, \bM, p)  \\ \hline
  (1/2, 1/2, 1/2) & (M_1M_2M_3)^{1/3} / p^{2/3}   \\ \hline
   (1, 0, 0) & M_1 / p \\ \hline
(0, 1, 0) & M_2 / p \\ \hline
(0, 0, 1) & M_3 / p \\ \hline
(0,0,0)  & 0 \\ \hline
\end{array} \]
\end{center}
(The last row is justified by the fact that $L(\bu, \bM, p) \leq
\max(M_1,M_2,M_3)/p^{1/(u_1+u_2+u_3)} \rightarrow 0$ when $u_1+u_2+u_3
\rightarrow 0$.)
%
The load of the HC algorithm is given by the largest of these
quantities, in other words, the optimal solution to the LP
\eqref{eq:primal:lp} that gives the load of the HC algorithm can be
given in closed form, as the maximum over these five expressions.  To
compute the speedup, suppose $M_1 < M_2=M_3 = M$.  Then there are two
cases.  When $p \leq M/M_1$, the optimal packing is $(0,1,0)$ (or
$(0,0,1)$) and the load is $M/p$.  HyperCube achieves linear speedup
by computing a standard join of $S_2 \Join S_3$ and broadcasting the
smaller relation $S_1$; it does this by allocating shares $p_1=p_2=1$,
$p_3=p$.  When $p > M/M_1$ then the optimal packing is $(1/2,1/2,1/2)$
the load is $(M_1M_2M_3)^{1/3}/p^{2/3}$, and the speedup decreases to
$1/p^{2/3}$.
\end{example}

The following lemma sheds some light into how the HyperCube algorithm exploits unequal
cardinalities.

\begin{lemma}
  Let $q$ be a query, over a database with statistics $\bM$, and let
  $\bu^* = \text{argmax}_\bu L(\bu, \bM, p)$, and $L= L(\bu^*, \bM,
  p)$.  Then:
  \begin{enumerate}
  \item If for some $j$, $M_j < L$, then $u^*_j = 0$.
  \item Let $M = \max_k M_k$.  If for some $j$, $M_j < M/p$, then
    $u^*_j = 0$.
  \item When $p$ increases, the speedup exponent remains constant or
    decreases, eventually reaching $1/\tau^*$.
  \end{enumerate}
\end{lemma}

\begin{proof}
We prove the three items of the lemma.

\noindent \textbf{(1)} If we modify a fractional edge packing $\bu$ by setting
    $u_j=0$, we still obtain a fractional edge packing.  We claim that
    the function $f(u_j) = L(\bu, \bM, p)$ is strictly decreasing in
    $u_j$ on $(0,\infty)$: the claim implies the lemma because $f(0) >
    f(u_j)$ for any $u_j > 0$.  The claim follows by noticing that
    $f(u_j) = p^{(u_j \log_p M_j + b)/(u_j + c)}$ where $a,b,c$ are
    positive constants, hence $f$ is
    monotone on $u_j \in (0,\infty)$, and $f(u_j) = L > M_j =
    f(\infty)$, implying that it is monotonically decreasing.
    
\noindent \textbf{(2)} This follows immediately from the previous item by noticing that
    $M/p \leq L$; to see the latter, let $k$ be such that $M_k = M$,
    and let $\bu$ be the packing $u_k=1$, $u_j=0$ for $j\neq k$.  Then
    $M/p = L(\bu,\bM,p) \leq L(\bu^*,\bM,p) = L$.

\noindent \textbf{(3)} Consider two edge packings $\bu, \bu'$, denote $u = \sum_j
    u_j$, $u' = \sum_j u_j'$, and assume $u < u'$.  Let $f(p) = L(\bu,
    \bM, p)$ and $g(p) = L(\bu',\bM,p)$.  We have $f(p) = c/p^{1/u}$
    and $g(p) = c'/p^{1/u'}$, where $c, c'$ are constants independent
    of $p$.  Then $f(p) < g(p)$ if and only if
    $p > (c/c')^{1/(1/u-1/u')}$, since $1/u - 1/u'>0$. Thus,
    as $p$ increases from $1$ to $\infty$, initially we have $f(p) < g(p)$,
    then $f(p) > g(p)$, and the crossover point is $(c/c')^{1/(1/u-1/u')}$.  
    Therefore, the value $\sum_j u^*_j$ can never decrease,
    proving the claim.  To see that the speedup exponent reaches
    $1/\tau^*$, denote $\bu^*$ the optimal vertex packing (maximizing
    $\sum_j u_j$) and let $\bu$ be any edge packing s.t. $u = \sum_j
    u_j < \tau^*$.  Then, when $p^{1/u - 1/\tau^*} > (\prod_j
    M_j^{u_j^*})^{1/\tau^*}/(\prod_j M_j^{u_j})^{1/u}$, we have
    $L(\bu^*,\bM,p)>L(\bu,\bM,p)$.
\end{proof}

The first two items in the lemma say that, if $M$ is the size of the
largest relation, then the only relations $S_j$ that matter to the
HC algorithm are those for which $M_j \geq M/p$; any smaller
relation will be broadcast by the HC algorithm.  The last item
says that the HC algorithm can take advantage of unequal
cardinalities and achieve speedup better than $1/p^{1/\tau^*}$,
\eg by allocating fewer shares to the smaller relations, or even
broadcasting them. As $p$ increases, the speedup decreases
until it reaches $1/p^{1/\tau^*}$.

\paragraph{Space Exponent} 
Let $|I| =\sum_j M_j $ denote the size of the input
database. Sometimes it is convenient to study algorithms whose maximum
load per server is given as $L = O(|I|/p^{1-\varepsilon})$, where $0
\leq \varepsilon < 1$ is a constant parameter $\varepsilon$ called the
{\em space exponent} of the algorithm.  The lower bound given by
\autoref{th:lower:uniform} can be interpreted as a lower bound on the
space exponent. To see this, consider the special case, when all
relations have equal size $M_1= \ldots = M_\ell = M$; then the load
can also be written as $L = O(M/p^{1-\varepsilon})$, and, denoting
$\bu^*$ the optimal fractional edge packing, we have $\sum_j u_j^* =
\tau^*$ and $L(\bu^*, \bM, p) = M / p^{1/\tau^*}$.
~\autoref{th:lower:uniform} implies that any algorithm with a fixed
space exponent $\varepsilon$ will report at most as many answers:
\begin{align*}
  O(\left(\frac{L}{L(\bu^*,\bM,p)}\right)^{\tau^*}) \cdot \E[|q(I)|] =
  O(p^{\tau^*[\varepsilon - (1-1/\tau^*)]}) \cdot \E[|q(I)|]
\end{align*}
Therefore, if the algorithm has a space exponent $\varepsilon < 1 -
1/\tau^*$, then, as $p$ increases, it will return a smaller fraction
of the expected number of answers.  This supports the intuition that
achieving parallelism becomes harder when $p$ increases: an algorithm
with a small space exponent may be able to compute the query correctly
when $p$ is small, but will eventually fail, when $p$ becomes large
enough.

\begin{table*}
  \centering
  \renewcommand{\arraystretch}{1.8}
\resizebox{\textwidth}{!}
{
  \begin{tabular}[c]{|l|c|c|c|} \hline
    Conjunctive Query     & Share         &   Value & Lower Bound for  \\
                &  Exponents &   $\tau^*(q)$ & Space Exponent \\ \hline
$C_k(x_1,\ldots,x_k) = \bigwedge_{j=1}^{k} S_j(x_j,x_{(j \bmod k)+1}) $  
 & $\frac{1}{k}, \dots, \frac{1}{k}$& $k/2$ & $1- 2/k$ \\ \hline
$T_k(z, x_1, \ldots, x_k)= \bigwedge_{j=1}^k S_j(z,x_j) $   & $1, 0,\dots, 0$ & $1$ & $0$ \\ \hline
$L_k(x_0, x_1,\ldots, x_k) =  \bigwedge_{j=1}^k S_j(x_{j-1},x_j)$   
  & $0, \frac{1}{\lceil k/2 \rceil},0,\frac{1}{\lceil k/2 \rceil}, \dots$& $\lceil k/2 \rceil$ & $1 - 1/\lceil k/2 \rceil$\\ \hline
$B_{k,m}(x_1, \ldots, x_k) = \bigwedge_{I \subseteq [k], |I|=m} S_{I}(\bar x_{I})$  
 & $\frac{1}{k}, \dots, \frac{1}{k}$ & $k/m$ & $1- m/k$\\ \hline
  \end{tabular}
}
  \caption{Query examples: $C_k=$ cycle query, $L_k=$ linear
    query, $T_k=$ star query, and $B_{k,m}=$ query with $\binom{k}{m}$ relations, where each relation
    contains a distinct set of $m$ out of the $k$ head variables. The share exponents presented
    are for the case where the relation sizes are equal.
   }
  \label{tab:queries}
\end{table*}

\paragraph{Replication Rate}
Given an algorithm that computes a conjunctive query $q$, let $L_s$ be the load of server $s$, where $s = 1, \dots, p$. The {\em replication rate} $r$ of the algorithm, defined in \cite{DBLP:journals/corr/abs-1206-4377}, is $ r = \sum_{s=1}^p L_i / |I|$. In other words, the replication rate computes how many times on average each input bit is communicated. The authors in~\cite{DBLP:journals/corr/abs-1206-4377} discuss the tradeoff between $r$ and the maximum load in the case where the number of servers is not given, but can be chosen optimally. We show next how we can apply our lower bounds to obtain a lower bound for the tradeoff between the replication rate and the maximum load.

\begin{corollary}
\label{th:map-reduce}
Let $q$ be a conjunctive query with statistics $\bM$. Any algorithm that computes $q$ with maximum load $L$, where $L \leq M_j$ for every $S_j$\footnote{if $L > M_j$, we can send the whole relation to any processor without cost} must have replication rate
$$ r \geq \frac{c L}{\sum_j M_j} \max_{\bu} \prod_{j=1}^{\ell} \left( \frac{M_j}{L} \right)^{u_j} $$ 
where $\bu$ ranges over all fractional edge packings of $q$ and $c = \max_\bu (\sum_j u_j/4)^{\sum_j u_j}$.
\end{corollary}

\begin{proof}
  Let $f_s$ be the fraction of answers returned by server $s$, in
  expectation, where $I$ is a randomly chosen matching database with statistics
  $\bM$.  Let $\bu$ be an edge packing for $q$ and $c(\bu) = (\sum_j u_j/4)^{\sum_j u_j}$; by  \autoref{th:lower:uniform}, 
   $ f_s \leq  \frac{L_s^{\sum_j u_j}}{c(\bu) \prod_j M_j^{u_j}}$.
  Since we assume all answers are returned,
  \begin{align*}
    1 \leq & \sum_{s=1}^p f_s = \sum_{s=1}^p \frac{L_s^{\sum_j u_j}}{c(\bu) \prod_j M_j^{u_j}}
      \leq  \frac{L^{\sum_j u_j-1} \sum_{s=1}^p L_s}{c(\bu) \prod_j M_j^{u_j}}
      =   \frac{L^{\sum_j u_j-1} r |I|}{c(\bu) \prod_j M_j^{u_j}}
  \end{align*}
  where we used the fact that $\sum_j u_j \geq 1$ for the optimal $\bu$. The claim follows by 
   noting that $|I| = \sum_j M_j$.
\end{proof}

In the specific case where the relation sizes are all equal to $M$, the above corollary
tells us that the replication rate must be $r =
\Omega((M/L)^{\tau^*-1})$. Hence, the ideal
case where $r = o(1)$ is achieved only when the maximum vertex cover number $\tau^*$
is equal to 1 (which happens if and only if a variable occurs in every atom of the query).

\begin{example}
Consider again the triangle query $C_3$ and assume that all sizes 
are equal to $M$. In this case, the edge packing that maximizes the lower bound is 
$(1/2,1/2,1/2)$, and $\tau^* = 3/2$. 
Thus, we obtain an $\Omega(\sqrt{M/L})$ bound for the replication rate for the triangle
query.
\end{example}

\section{Handling Data Skew in One Communication Step}
\label{sec:skew}

In this section, we discuss how to compute queries in the MPC model in the presence of skew. 
We first start by presenting an example where the HC algorithm that uses the optimal
shares from \eqref{eq:primal:lp} fails to work when the data has skew, even though
it is asymptotically optimal when the relations are of low degree.

\begin{example} \label{ex:skew:presence}
  Let $q(x,y,z) = S_1(x,z),S_2(y,z)$ be a simple join query, where both
  relations have cardinality $m$ (and size in bits $M$). The optimal
  shares are  $p_1=p_2 = 1$, and $p_3 = p$.  
  This allocation of shares corresponds to a standard parallel hash-join algorithm, where
  both relations are hashed on the join variable $z$.
  When the data has no skew, the maximum load is $O(M/p)$ with high probability. 
  
  However, if the relation has skew, the maximum load can be as large as $O(M)$. This
  occurs in the case where all tuples from $S_1$ and $S_2$ have the same value for
  variable $z$.
\end{example}

As we can see from the above example, the problem occurs when the input data contains values
with high  frequency of occurrence, which we call {\em outliers}, or {\em heavy hitters}.
We will consider two different scenarios when handling data skew. 
In the first scenario, in \autoref{subsec:skew:without}, we assume that the algorithm has no information about the data apart from the size of the relations.

 In the second scenario, presented in \autoref{subsec:skew:with}, we assume that the algorithm knows about the outliers in our data. All the results in this section are limited to single-round algorithms.

\subsection{The HyperCube Algorithm with Skew} 
\label{subsec:skew:without}

We answer the following question: what are the optimal shares for the HC algorithm such
that the maximum load is minimized over all possible distributions of input data? In other
words, we limit our treatment to the HyperCube algorithm, but we consider data that can
heavily skewed, as in~\autoref{ex:skew:presence}. Notice that the HC algorithm is oblivious
of the values that are skewed, so it cannot be modified in order to handle these cases separately.
Our analysis is based on the following
lemma about hashing, which we prove in detail in~\autoref{sec:hashing}.

\begin{lemma}
\label{cor:hashing:2}
 Let $R(A_1, \dots, A_r)$ be a relation of arity $r$ with $m$ tuples.
 Let $p_1, \ldots, p_r$ be integers and $p = \prod_i
  p_i$. Suppose that we hash each tuple $(a_1,
  \ldots, a_r)$ to the bin $(h_1(a_1), \ldots, h_r(a_r))$, where
  $h_1, \ldots, h_r$ are independent and perfectly random hash 
  functions. Then, the probability that the maximum load exceeds $O(m/ (\min_i p_i))$
  is exponentially small in $m$.
\end{lemma}
 
\begin{corollary}
Let $\mathbf{p} = (p_1, \dots, p_k)$ be the shares of the HC algorithm.
  For any relations, with high probability the maximum load per server  is
  $$O \left( \max_j \frac{M_j}{\min_{i: i \in S_j} p_i} \right)$$ 
\end{corollary}

The abound bound is tight: we can always construct an instance for given shares
such that the maximum load is at least as above. Indeed, for a relation $S_j$ with
$i = \arg \min_{i \in S_j} p_i$, we can construct an instance with a single value for any attribute
other than $x_i$, and $M_j$ values for $x_i$. In this case, the hashing will be across only one
dimension with $p_i$ servers, and so the maximum load has to be at least $M_j / p_i$ for the
relation $S_j$.

As in the previous section, if $L$ denotes the maximum load per server, we must
have that $M_j / \min_{i \in S_j} p_i \leq L$.
Denoting $\lambda = \log_p L$ and $\mu_j = \log_p M_j$, the load is optimized 
by the following LP:
\begin{align}
\text{minimize}   \quad & \lambda \nonumber \\
\text{subject to} \quad
                  &  \sum_{i \in [k]} -e_i \geq -1 \nonumber \\
\quad \forall j \in [\ell]:  & h_j + \lambda \geq \mu_j \nonumber \\
\quad \forall j \in [\ell], i \in S_j:  & e_i - h_j \geq 0 \nonumber \\
\quad \forall i \in [k]: & e_i \geq 0, \quad \forall j \in [\ell]: h_j \geq 0 \quad \lambda \geq 0 \label{eq:lp2}
\end{align}

Following the same process as in the previous section, we can obtain the dual of the above LP,
and after transformations obtain the following non-linear program with the same optimal objective function:
\begin{align}
  \text{maximize}   \quad &  \frac{\sum_{j \in [\ell]}\mu_j u_j - 1}{\sum_{j \in [\ell]} u_j} \nonumber \\
  \text{subject to} \quad
    \forall i \in [k] : &  \sum_{j:i \in S_j} w_{ij} \leq 1 \nonumber \\
  \quad \forall j \in [\ell]: & u_j \leq \sum_{i \in S_j} w_{ij} \nonumber \\
  \quad \forall j \in [\ell]: & u_j \geq 0 \nonumber \\
  \quad \forall i \in [k], j \in [\ell]: & w_{ij} \geq 0
\end{align}

\subsection{Skew with Information} 
\label{subsec:skew:with}

We discuss here the case where there is additional information known about
skew in the input database. We will present a general lower bound for 
arbitrary conjunctive queries, and show an algorithm that matches the bound
for {\em star queries}
$$q = S_1(z,x_1), S_2(z, x_2), \dots, S_\ell(z, x_\ell)$$
which are a generalization of the join query. In~\cite{BeameKS14} we show how
our algorithmic techniques (the \textsc{BinHC} algorithm)
can be used to compute arbitrary conjunctive queries; 
however, there is a substantial gap between the upper and lower bounds in
the general case.
 
We first introduce some necessary notation.  
For each relation $S_j$ with $|S_j| = m_j$, and each assignment $h \in [n]$ for
 variable $z$, we define its {\em frequency} as 
$m_j(h) = |\sigma_{z = h}(S_j)|$. We will be
interested in assignments that have high frequency, which we call
{\em heavy hitters}. 
In order to design algorithms that take skew into account, we will assume that
every input server knows the assignments with frequency $\geq m_j/p$
for every relation $S_j$, along with their frequency. Because each relation 
can contain at most $p$ heavy hitters, the total number over all
relations will be $O(p)$. Since we are considering cases where
the number of servers is much smaller than the data, an $O(p)$ amount of
information can be easily stored in the input server.

To prove the lower bound, we will make a stronger assumption about
the information available to the input servers.
Given a conjunctive query $q$, fix a set of variables $\bx$ and let $d = |\bx|$
Also, let $\bx_j = \bx \cap \vars{S_j}$ for every relation $S_j$, and $d_j = |\bx_j|$.
A {\em statistics of type $\bx$, or $\bx$-statistics} is a vector $m = (m_1, \dots, m_\ell)$,
where $m_j$ is a function $m_j: [n]^{\bx_j} \rightarrow \mathbb{N}$.
We associate with $m$ the function $m : [n]^{\bx} \rightarrow
(\mathbb{N})^\ell$, where $m(\bh) = (m_1(\bh_1), \ldots, m_\ell(\bh_\ell))$, 
and $\bh_j$ denotes the restriction of the tuple $\bh$ to the variables in $\bx_j$. 
We say that an instance of $S_j$ satisfies the statistics if for any tuple $\bh_j \in [n]^{\bx_j}$ ,
its frequency is precisely $m_j(\bh_j )$.
When $\bx = \emptyset$, then $m$ simply consists of $\ell$ numbers, each representing the cardinality of a relation; thus, a $\bx$-statistics generalizes the cardinality statistics. 
Recall that we use upper case $\bM = (M_1, \ldots,
M_\ell)$ to denote the same statistics expressed in bits,
i.e. $M_j(\bh) = a_j m_j(\bh) \log (n)$.  

In the particular case of the star query, we will assume that the input servers know
the $z$-statistics; in other words, for every assignment $h \in [n]$ of variable $z$,
we know that its frequency in relation $S_j(z, x_j)$ is precisely $m_j(h)$.
Observe that In this case the cardinality of $S_j$ is
$|S_j| = \sum_{h \in [n]} m_j(h)$.
%

\subsubsection{Algorithm for Star Queries}
\label{sec:skew:join}

The algorithm uses the same principle popular in virtually all
parallel join implementations to date: identify the heavy hitters and
treat them differently when distributing the data.  
However, the analysis and optimality proof is new, to the best of our knowledge.

Let $H$ denote the set of heavy hitters in all relations. Note that
$|H| \leq \ell p$. 
The algorithm will deal with the tuples that have no heavy hitter values
({\em light tuples}) by running the vanilla HC algorithm, which runs with shares
$p_z=p$ and $p_{x_j} = 1$ for every $j=1, \dots, \ell$.
For this case, the load analysis of~\autoref{sec:hashing} will give us a 
maximum load of $\tilde{O}(\max_j M_j/p)$ with high probability,
where $\tilde{O}$ hides a polylogarithmic factor that depends on $p$. 
For heavy hitters, we will have to adapt its function as follows.

To compute $q$, the algorithm must compute for each $h \in H$
the subquery 
$$q[h / z] = S_1(h, x_1), \dots S_k(h, x_k)$$ 
which is equivalent to computing the Cartesian product $q_z = S_1'(x_1), \dots, S_k'(x_k)$, where
$S_1'(x_1) = S_1(h,x_1)$ and $S_2'(x_2) = S_2(h,x_2)$, and each relation $S_j'$ has cardinality $m_j(h)$ (and size in bits $M_j(h)$). 
We call $q_z$ the {\em residual query}. 
The algorithm will allocate $p_h$ servers to compute $q[h/z]$ for each $h \in H$, 
such that $\sum_{h \in H} p_h = \Theta(p)$. 
Since the unary relations have no skew, they will be of low degree and thus 
the maximum load $L_h$ for each $h$ is given by
$$ 
L_h = O \left( \max_{\bu \in pk(q_z)} L(\bu, \bM(h), p_h)  \right)
$$
For the star query, we have $pk(q_z) = \{0,1\}^\ell \setminus (0,0,\dots, 0)$.
At this point, since $p_h$ is not specified, it is not clear which edge packing in $pk(q_z)$ 
maximizes the above quantity for each $h$.
To overcome this problem, we further refine the assignment of servers to heavy hitters:
we allocate $p_{h,\bu}$ servers to each $h$ and each $\bu \in pk(q_z)$, such that 
$p_h = \sum_{\bu} p_{h,\bu}$.
Now, for a given $\bu \in pk(q_z)$, we can evenly distribute the load among the heavy
hitters by allocating servers proportionally to the "heaviness" of executing the residual query.
In other words we want $p_{h, \bu} \sim \prod_j M_j(h)^{u_j}$ for every $h \in H$.
Hence, we will choose:
$$ p_{h,\bu} = 
\left\lceil p \cdot \frac{\prod_j M_j(h)^{u_j})}
{\sum_{h' \in H} \prod_j M_j(h')^{u_j}}  \right\rceil$$
Since $\lceil x \rceil \leq x+1$, and $|H| \leq \ell p$, we can compute that 
the total number of servers we need is at most $(\ell +1) \cdot |pk(q_z)| \cdot p$, which is $\Theta(p)$. Additionally, the maximum load $L_h$ for every $h \in H$ will be 
\begin{align*}
L_h = O \left(  \max_{\bu \in pk(q_z)} 
\left( \frac{ \sum_{h \in H} \prod_j M_j(h)^{u_j} }{p} \right)^{1/(\sum_j u_j)} \right)
\end{align*}
Plugging in the values of $pk(q_z)$, we obtain the following upper bound on the algorithm
for the heavy hitter case:
\begin{align} \label{eq:join:load}
O \left(  \max_{I \subseteq [\ell]} 
\left( \frac{ \sum_{h \in H} \prod_{j \in I} M_j(h)}{p} \right)^{1/|I|} \right)
\end{align}
Observe that the terms depend on the frequencies of the heavy hitters, and can be
much larger than the bound $O(\max_j M_j/p)$ we obtain from the light hitter case. 
In the extreme, a single heavy hitter $h$ with
$m_j(h) = m_j$ for $j=1, \dots, \ell$ will demand maximum load 
equal to $O ((\prod_j M_j /p)^{1/\ell})$.

\subsubsection{Algorithm for Triangle Query}

We show here how to compute the triangle query $ C_3 = R(x,y), S(y,z), T(z,x)$
when all relation sizes are equal to $m$ (and have $M$ bits).
As with the star query, the algorithm will deal with the tuples that have no heavy hitter values, 
\ie the frequency is less than $m/p^{1/3}$, by running the vanilla HC algorithm.
For this case, the load analysis of~\autoref{sec:hashing} will give us a 
maximum load of $\tilde{O}(M/p^{2/3})$ with high probability.

Next, we show how to handle the heavy hitters. We distinguish two cases.

\paragraph{Case 1} In this case, we handle the tuples that have values with 
frequency $\geq m/p$ in at least two variables. 
Observe that we did not set the heaviness threshold to $m/p^{1/3}$, for
reasons that we will explain in the next case.

Without loss of generality, suppose that both $x,y$ are heavy in at least one of the two relations
they belong to. The observation is that there at most $p$ such heavy values for each variable,
and hence we can send all tuples of $R(x,y)$ with both $x,y$ heavy (at most $p^2$) to
all servers. Then, we essentially have to compute the query $S'(y,z), T'(z,x)$, where $x$ and $y$
can take only $p$ values. We can do this by computing the join on $z$; since the frequency of $z$ will be at most $p$ for each relation, the maximum load from the join
computation will be $O(M/p)$.

\paragraph{Case 2} In this case, we handle the remaining output: this includes the tuples
where one variable has frequency $\geq m/p^{1/3}$, and the other variables are light,
\ie have frequency $\leq m/p$. Without loss of generality, assume that we
want to compute the query $q$ for the values of $x$ that are heavy in either $R$ or $T$.
Observe that there are at most $2p^{1/3}$ of such heavy hitters. If $H_x$ denotes the set
of heavy hitter values for variable $x$, the residual query $q[h/x]$ for each $h \in H$ is:
$$ q[h/x] = R(h,y), S(y,z), T(z,h)$$
which is equivalent to computing the query $q_x = R'(y), S(y,z), T'(z)$ with cardinalities
$m_R(h),m, m_T(h)$ respectively. As before, we allocate $p_h$ servers to compute $q[h/x]$
for each $h \in H$. 
If there is no skew, the maximum load $L_h$ is given by the following formula:
$$ 
L_h = O \left( \max \left( \frac{M}{p_h}, \sqrt{\frac{M_R(h) M_T(h)}{p_h}} \right)  \right)
$$
Notice now that the only cause of skew for $q_x$ may be that $y$ or $z$ are heavy in
$S(y,z)$. However, we assumed that the frequencies for both $y,z$ are $\leq m/p$, so
there will be no skew (this is why we set the heaviness threshold for Case 1 to $m/p$
instead of $m/p^{1/3}$).

We can now set $p_h = p_{h,1} + p_{h,2}$ (for each of the quantities in the max
expression), and choose the allocated servers similarly to how we chose for the
star queries:
\begin{align*}
p_{h,1} & = \left\lceil p \cdot \frac{M_S(h)}{M}  \right\rceil 
\quad \quad \quad
p_{h,2}  = \left\lceil p \cdot \frac{ M_R(h) M_T(h)}
{\sum_{h \in H_x} M_R(h) M_T(h)}  \right\rceil
\end{align*}

We now get a load of:
$$ L =  O \left( \max \left( \frac{M}{p}, \sqrt{\frac{\sum_h M_R(h) M_T(h)}{p}} \right)  \right)$$
Summing up all the cases, we obtain that  the load of the 1-round algorithm for
computing triangles is:
$$ L =  \tilde{O} \left( \max \left( \frac{M}{p^{2/3}}, \sqrt{\frac{\sum_h M_R(h) M_T(h)}{p}} 
, \sqrt{\frac{\sum_h M_R(h) M_S(h)}{p}} , \sqrt{\frac{\sum_h M_S(h) M_T(h)}{p}}  \right)  \right)$$

\subsubsection{Lower Bound}
\label{sec:skew:lower}

The lower bound we present here holds for any conjunctive query, and generalizes 
the lower bound in \autoref{th:lower:uniform}, which was over databases
with cardinality statistics $\bM = (M_1, \dots, M_\ell)$,
to databases with a fixed degree sequence. If the degree sequence is
skewed, then the new bounds can be stronger, proving that skew in
the input data makes query evaluation harder.

Let us fix statistics $\bM$ of type $\bx$.
We define $q_\bx$ as the
{\em residual query}, obtained by removing all variables $\bx$, and
decreasing the arities of $S_j$ as necessary (the new arity of relation $S_j$
is $a_j - d_j$).  Clearly, every fractional edge packing of $q$ is
also a fractional edge packing of $q_\bx$, but the converse does not
hold in general.  If $\bu$ is a fractional edge packing of $q_\bx$,
we say that $\bu$ {\em saturates} a variable $x_i \in \bx$, if
$\sum_{j: x_i \in vars(S_j)} u_j \geq 1$; we say that $\bu$ saturates $\bx$ if
it saturates all variables in $\bx$.  For a given $\bx$ and $\bu$ that saturates
$\bx$, define
  \begin{align}
  L_\bx(\bu, \bM, p) = 
  \left( \frac{\sum_{\bh \in [n]^{\bx}} \prod_j M_j(\bh_j)^{u_j})}{p} \right)^{1/\sum_j u_j}\label{eq:lx}
\end{align}

\begin{theorem} 
  \label{th:lower:skew} 
  Fix statistics $\bM$ of type $\bx$ such that $a_j > d_j$ for every relation $S_j$. 
  Consider any deterministic MPC algorithm that runs in one communication round on
  $p$ servers and has maximum load $L$ in bits.  Then, for any fractional edge 
  packing $\bu$ of $q$ that saturates $\bx$, we must have
   $$L \geq \min_j \frac{(a_j-d_j)}{4 a_j} \cdot L_\bx(\bu, \bM, p).$$
\end{theorem}

Note that, when $\bx = \emptyset$ then $L_\bx(\bu, \bM,p) = L(\bu,\bM,p)$,
as defined in \eqref{eq:lump}. However, our theorem does not imply 
\autoref{th:lower:uniform}, since it does not give a lower
bound on the expected size of the algorithm output as a fraction of the 
expected output size.

\begin{proof}
For $\bh \in [n]^{\bx}$ and $\ba_j \in S_j$, we write
$\ba_{j} \parallel \bh$ to denote that the tuple $\ba_j$ from $S_j$ matches
with $\bh$ at their common variables $\bx_j$, and denote $(S_j)_\bh$ the
subset of tuples $\ba_j$ that match $\bh$: $(S_j)_\bh
=\setof{\ba_j}{\ba_j \in S_j, \ba_j \parallel \bh}$.  Let $I_{\bh}$ denote the
restriction of $I$ to $\bh$, in other words $I_\bh = ((S_1)_\bh,
\ldots, (S_\ell)_\bh)$.

We pick the domain $n$ such that $n = (\max_j \{m_j\})^2$ and
construct a probability space for instances $I$ defined by the
$\bx$-statistics $\bM$ as follows. 
For a fixed tuple $\bh \in [n]^{\bx}$, the
restriction $I_\bh$ is a uniformly chosen instance over all matching
databases with cardinality vector $\bM(\bh)$, which is precisely the probability space
that we used in the proof of \autoref{th:lump}.  In particular, for
every $\ba_j \in [n]^{\bx_j}$, the probability that $S_j$ contains
$\ba_j$ is $P(\ba_j \in S_j) = m_j(\bh_j) / n^{a_j-d_j}$.
\autoref{lem:expected_size} immediately gives:
$$ \E[|q(I_{\bh})|] =   n^{k-d} \prod_{j=1}^{\ell} \frac{m_j(\bh_j)}{ n^{a_j-d_j}}$$

Let us fix some server and let $\msg(I)$ be the message the server receives on input $I$.
As in the previous section, let $K_{\msg_j}(S_j)$ denote the set of tuples from relation $S_j$
{\em known by the server}. Let $w_j({\ba_j}) = P(\ba_j \in K_{\msg_j(S_j)}(S_j))$, 
where the probability is over the random choices of $S_j$.  This is upper bounded
by $P(\ba_j \in S_j)$:
\begin{align}
  w_j(\ba_j) \leq m_j(\bh_j) / n^{a_j-d_j}, \quad \mbox{if $\ba_j \parallel \bh$}
   \label{eq:w:n}
\end{align}

We derive a second upper bound by exploiting the fact that the server
receives a limited number of bits, in analogy with \autoref{lem:entropy_ratio}:
\begin{lemma}
\label{lem:L-bound}
Let $S_j$ a relation with $a_j >d_j$.
Suppose that the size of $S_j$ is $m_j \leq n/2$ (or $m_j = n$), and that the 
message $\msg_j(S_j)$ has at most $L$ bits. Then, we have
$\E[|K_{\msg_j}(S_j)|] \leq \frac{4 L}{(a_j-d_j) \log(n)}$.
\end{lemma}

Observe that in the case where $a_j = d_j$ for some relation $S_j$, the 
$\bx$-statistics fix all the tuples of the instance for this particular relation, 
and hence  $\E[|K_{\msg_j}(S_j)|] = m_j$.

\begin{proof}
We can express the entropy $H(S_j)$ as follows:
\begin{align} \label{eq:entropy:one}
 H(S_j) &=H(\msg_j(S_j))+ \sum_{\msg_j}P(\msg_j(S_j)=\msg_j) \cdot H(S_j \mid \msg_j(S_j)=\msg_j) \nonumber\\
 &\le L +  \sum_{\msg_j}P(\msg_j(S_j)=\msg_j) \cdot H(S_j \mid \msg_j(S_j)=\msg_j) 
\end{align}

For every $\bh \in [n]^{\bx}$, let $K_{\msg_j}((S_j)_{\bh})$ denote the known tuples that belong 
in the restriction of $S_j$ to $\bh$. 
Following the proof of~\autoref{lem:entropy_ratio},
and denoting by $\mathcal{M}_j(\bh_j)$ the number of bits
necessary to represent  $(S_j)_{\bh}$, we have:
\begin{align*}
 H(S_j \mid \msg_j(S_j)=\msg_j)  
 & \leq \sum_{\bh \in [n]^{\bx}} \left(1-\frac{|K_{\msg_j}((S_j)_{\bh})|}{2 m_j(\bh_j)} \right) \mM_j(\bh_j)\\
 & = H(S_j) - \sum_{\bh \in [n]^{\bx}} \frac{|K_{\msg_j}((S_j)_{\bh})|}{2 m_j(\bh_j)}  \mM_j(\bh_j) \\
 & \leq H(S_j) - \sum_{\bh \in [n]^{\bx}} \frac{|K_{\msg_j}((S_j)_{\bh})|}{2 m_j(\bh_j)} m_j(\bh_j)\frac{a_j-d_j}{2} \log (n)  \\ 
 & = H(S_j) - (1/4) \cdot |K_{\msg_j}(S_j)| (a_j-d_j)\log (n) 
\end{align*}
where the last inequality comes from \autoref{prop:bit-representation}.
Plugging this in~\eqref{eq:entropy:one}, and solving for $\E[|K_m(S_j)|] $:
\begin{align*}
\E[|K_{\msg_j}(S_j)|]  \leq \frac{4 L}{(a_j-d_j)\log(n)}
\end{align*}
This concludes our proof.
\end{proof}

Let $q_{\bx}$ be the residual query, and recall that $\bu$ is a
fractional edge packing that saturates $\bx$.  Define the {\em
  extended query} ${q_{\bx}}'$ to consist of $q_\bx$, where we add a
new atom $S_i'(x_i)$ for every variable $x_i \in
vars(q_\bx)$.  Define $u'_i = 1 - \sum_{j: i \in S_j} u_j$.
In other words, $u'_i$ is defined to be the slack at the variable
$x_i$ of the packing $\bu$.  The new edge packing
$(\mathbf{u},\mathbf{u}')$ for the extended query $q_\bx'$ has no more
slack, hence it is both a tight fractional edge packing and a tight
fractional edge cover for $q_{\bx}$.  By adding all equalities of the
tight packing we obtain:
$$\sum_{j=1}^{\ell} (a_j-d_j) u_j + \sum_{i=1}^{k-d} u'_i = k-d$$

We next compute how many output tuples from $q(I_{\bh})$ will be known
in expectation by the server. Note that $q(I_{\bh}) = q_{\bx}(I_{\bh})$, and thus:
\begin{align*}
\E[|K_{\msg}(q(I_{\bh}))|]  & = \E[|K_{\msg}(q_{\bx}(I_{\bh}))|] \\ 
& =  \sum_{ \ba \parallel \bh} \prod_{j=1}^{\ell} w_j({\ba_j}) \\
& =  \sum_{\ba \parallel \bh} \prod_{j=1}^{\ell} w_j({\ba_{j}}) 
         \prod_{i=1}^{k-d} w_i'({\ba_i})\\      
 & \leq   \prod_{i=1}^{k-d} n^{u_i'} \cdot 
  \prod_{j =1}^{\ell} \left( \sum_{\ba_j \parallel \bh} w_j({\ba_j})^{1/u_j} \right)^{u_j} 
\end{align*}
By writing $w_j(\ba_j)^{1/u_j} = w_j(\ba_j)^{1/u_j-1} w_j({\ba_j})$ for $\ba_j \parallel \bh$, we can bound the sum in the above quantity as follows:
\begin{align*}
 \sum_{\ba_j \parallel \bh} w_j({\ba_j})^{1/u_j} 
  \leq \left( \frac{m_j(\bh_j)}{ n^{a_j-d_j}} \right)^{1/u_j-1} \sum_{\ba_j \parallel \bh} w_j({\ba_j}) 
\quad  = (m_j(\bh_j)   n^{d_j-a_j})^{1/u_j-1}  L_j(\bh) \\
\end{align*}
where $L_j(\bh) = \sum_{\ba_j \parallel \bh} w_j({\ba_j})$. Notice that for every relation $S_j$, we have
$\sum_{\bh_j \in [n]^{\bx_j}} L_j(\bh_j)  = \sum_{\ba_j \in [n]^{a_j}} w_j(\ba_j)$.
We can now write:
\begin{align}
\E[|K_{\msg}(q(I_{\bh}))|] 
& \leq  n^{\sum_{i=1}^{k-d} u_i'}  \prod_{j=1}^{\ell} \left( L_{j}(\bh)   m_j(\bh_j)^{1/u_j-1}  n^{(d_j-a_j)(1/u_j-1)} \right)^{u_j}  \nonumber\\
%
%
& =  \prod_{j=1}^{\ell} L_{j}(\bh)^{u_j} \cdot \prod_{j=1}^{\ell} m_j(\bh_j)^{-u_j} \cdot \E [|q(I_{\bh})|] \label{eq:lastline}
\end{align}

Summing over all $p$ servers, we obtain that the expected number of answers that can be output for $q(I_{\bh})$ is at most $p \cdot \E[|K_{\msg}(q(I_{\bh}))|]$. If some $\bh \in [n]^{\bx}$ this number is not at least $\E[|q(I_{\bh})|]$, the algorithm will fail to compute $q(I)$. Consequently, for every $\bh$ we must have that 
$ \prod_{j=1}^{\ell} L_{j}(\bh_j)^{u_j} 
\geq  (1/p) \cdot \prod_{j =1}^{\ell} m_j(\bh_j)^{u_j}$.
Summing the inequalities for every $\bh \in [n]^{\bx}$:
\begin{align*}
\frac{1}{p} \cdot \sum_{\bh \in [n]^\bx} \prod_{j=1}^{\ell} m_j(\bh_j)^{u_j}
& \leq \sum_{\bh \in [n]^\bx} \prod_{j=1}^{\ell} L_j(\bh_j)^{u_j} \\
& \leq \prod_{j=1}^{\ell} \left( \sum_{\bh_j \in [n]^{\bx_j}} L_j(\bh_j) \right)^{u_j}
&\mbox{by Friedgut's inequality} \\
& \leq \prod_{j=1}^{\ell} \left( \frac{4L}{(a_j -d_j) \log(n)} \right)^{u_j}
&\mbox{by \autoref{lem:L-bound}} 
\end{align*}
Solving for $L$ and using the fact that $M_j = a_j m_j \log(n)$, we obtain that for any edge packing $\bu$ that 
saturates $\bx$, 
\begin{align*}
 L \geq  \left(\min_j \dfrac{a_j-d_j}{4 a_j} \right) \cdot \left( \frac{\sum_{\bh \in [n]^\bx} \prod_{j} M_j(\bh_j)^{u_j}}{p} \right)^{1/\sum_j u_j}  
\end{align*}
which concludes the proof. \qed
\end{proof}

To see how~\autoref{th:lower:skew} applies to the star query, we assume that the input
servers know $z$-statistics $\bM$; in other words, for every assignment $h \in [n]$ 
of variable $z$, we know that its frequency in relation $S_j$ is $m_j(h)$.
Then, for any edge packing $\bu$ that saturates $z$, we obtain a lower bound of
$$ L \geq (1/8) \cdot \left( \frac{\sum_{h \in [n]} \prod_{j} M_j(h)^{u_j}}{p} \right)^{1/\sum_j u_j}  $$
Observe that the set of edge packings that saturate $z$ and maximize the above quantity is
$\{0,1\}^{\ell} \setminus (0, \dots, 0)$. Hence, we obtain a lower bound 
$$ L \geq (1/8) \cdot  \max_{I \subseteq [\ell]} \left( \frac{\sum_{h \in [n]} \prod_{j \in I} M_j(h)}{p} \right)^{1/|I|}  $$

\section{Multiple Communication Steps}
\label{sec:multistep}

In this section, we discuss the computation of queries in the MPC model
in the case of multiple steps. We will establish both upper and lower bounds on 
the number of rounds needed to compute a query $q$. 

To prove our results, we restrict both the structure of the input and the
type of computation in the MPC model. In particular, our multi-round
algorithms will process only queries where the relations are {\em of equal size}
and the data has {\em no skew}. 
Additionally, our lower bounds are proven for a restricted version of the \mpc\
model, called the the {\em tuple-based \mpc\ model}, which limits the way
communication is performed.

\subsection{An Algorithm for Multiple Rounds}
\label{sec:multi:upper}

In~\autoref{sec:onestep}, we showed that in the case where all relations
have size equal to $M$ and are matching databases (\ie the degree of
any value is exactly one), we can compute a conjunctive query $q$ in
one round with maximum load
$$ L = O(M/p^{1/\tau^*(q)})$$ 
where $\tau^*(q)$ denotes the fractional vertex covering number of $q$.
Hence, for any $\varepsilon \geq 0$, a conjunctive query $q$ with
$\tau^*(q) \leq 1/(1-\varepsilon)$ can be computed in one round in the
\mpc\ model with load $L= O(M/p^{1-\varepsilon})$;
recall from \autoref{sec:onestep-discuss} that we call the parameter 
$\varepsilon$ the {\em space exponent}. 

We define now the class of queries $\Gamma^r_\varepsilon$ using 
induction on $r$. For $r=1$, we define
$$\Gamma_\varepsilon^1 = \setof{q}{\tau^*(q) \leq 1/(1-\varepsilon)}$$
For $r > 1$, we define
$\Gamma_\varepsilon^{r}$ to be the set of all conjunctive queries $q$
constructed as follows.  Let $q_1, \ldots, q_m \in
\Gamma_\varepsilon^{r-1}$ be $m$ queries, and let $q_0 \in
\Gamma_\varepsilon^1$ be a query over a different vocabulary $V_1,
\ldots, V_m$, such that $|\text{vars}(q_j)|=\text{arity}(V_j)$ for all
$j\in [m]$.  Then, the query $q = q_0[q_1/V_1, \ldots, q_m/V_m]$,
obtained by substituting each view $V_j$ in $q_0$ with its definition
$q_j$, is in $\Gamma_\varepsilon^{r}$.  In other words,
$\Gamma^r_\varepsilon$ consists of queries that have a {\em query
  plan} of depth $r$, where each operator is a query computable in one
step with maximum load $O(M/p^{1-\varepsilon})$.
The following proposition is now straightforward.
 
\begin{proposition} \label{prop:multistep} 
Every conjunctive query $q \in \Gamma_\varepsilon^r$ with input a matching
database where each relation has size $M$ can be computed by 
an \mpc\ algorithm in $r$ rounds with maximum load $L = O(M/p^{1-\varepsilon})$.
\end{proposition}

We next present two examples that provide some intuition on the structure of the
queries in the class $\Gamma_\varepsilon^r$.

\begin{example} \label{ex:queryplan} Consider the
  query $L_k$ in \autoref{tab:queries} with $k=16$; we can construct a query plan
  of depth $r=2$ and load $L = O(M/p^{1/2})$ (with space exponent $\varepsilon=1/2$). 
   The first step computes in parallel four queries,
  $v_1 = S_1, S_2, S_3, S_4$, \ldots, $v_4 = S_{13}, S_{14}, S_{15},
  S_{16}$.  Each query  is isomorphic to $L_4$, therefore
  $\tau^*(q_1)=\cdots=\tau^*(q_4) = 2$ and thus each can be computed in one
  step with load $L = O(M/p^{1/\tau^*(q_1)}) = O(M/p^{1/2})$.  
  The second step computes the query $q_0 = V_1,
  V_2, V_3, V_4$, which is also isomorphic to $L_4$. 
\end{example}  

We can generalize the above approach for any query $L_k$.
For any $\varepsilon \geq 0$, let $k_{\varepsilon}$ be the
largest integer such that $L_{k_{\varepsilon}} \in \Gamma_1^\varepsilon$.
In other words, $\tau^*(L_{k_{\varepsilon}}) \leq
  1/(1-\varepsilon)$ and so we choose $k_{\varepsilon} = 2 \lfloor 1/(1-\varepsilon)
  \rfloor$.  Then, for any $k \geq k_{\varepsilon}$, $L_k$ can be
  computed using $L_{k_{\varepsilon}}$ as a building block at each
  round: the plan will have a depth of $\lceil \log_{k_{\varepsilon}} (k)
   \rceil$ and will achieve a load of $L = O(M/p^{1-\varepsilon})$.

\begin{example}
  Consider the query $SP_k =
  \bigwedge_{i=1}^k R_i(z,x_i), S_i(x_i, y_i)$. Since $\tau^*(SP_k) =
  k$, the one round algorithm can achieve a load of $O(M/p^{1/k})$.
  
  However, we can construct a query plan of depth 2 for $SP_k$ with
  load $O(M/p)$, by computing the joins $q_i =
  R_i(z,x_i), S_i(x_i, y_i)$ in the first round and in the second
  round joining all $q_i$ on the common variable $z$.  
\end{example}

We next present an upper bound on the number of rounds needed to compute any
query if we want to achieve a given load $L = O(M/p^{1-\varepsilon})$; in other words,
we ask what is the minimum number of rounds for which we can achieve a space
exponent $\varepsilon$.

 Let $\text{rad}(q) = \min_u \max_v d(u,v)$ denote the {\em
  radius} of a query $q$, where $d(u,v)$ denotes the distance between
two nodes in the hypergraph of $q$.  For example, $\text{rad}(L_k) = \lceil
k/2 \rceil$ and $\text{rad}(C_k) = \lfloor k/2 \rfloor$.

\begin{lemma} 
\label{lem:multitree}
Fix $\varepsilon \geq 0$, let $k_{\varepsilon} = 2 \lfloor
1/(1-\varepsilon) \rfloor$, and let $q$ be any connected query.  Define
\begin{align*}
r(q)= \begin{cases}
\lceil \log_{k_{\varepsilon}} (\text{rad}(q)) \rceil +1 
& \mbox{if $q$ is tree-like,} \\
 \lfloor \log_{k_{\varepsilon}} (\text{rad}(q))\rfloor +2
& \mbox{otherwise.}
\end{cases}
\end{align*}
 Then, $q$ can be computed in
$r(q)$ rounds on any matching database input with relations of size $M$
with maximum load $L = O(M/p^{1-\varepsilon})$.
\end{lemma}

\begin{proof}
By definition of $\text{rad}(q)$, there exists some node
$v \in \text{vars}(q)$,
such that the maximum distance of $v$ to any other node in the hypergraph of
$q$ is at most $\text{rad}(q)$. 
If $q$ is tree-like then we can decompose $q$ into a set of at most
$|\text{atoms}(q)|^{\text{rad}(q)}$
(possibly overlapping) paths ${\cal P}$ of length $\leq \text{rad}(q)$,
each having $v$ as one endpoint.
Since it is essentially isomorphic to $L_{\ell}$,
a path of length $\ell \leq \text{rad}(q)$ can be computed 
in at most  $\lceil \log_{k_{\varepsilon}} (\text{rad}(q)) \rceil$
rounds using the query plan from \autoref{prop:multistep} together with 
repeated use of the one-round HyperCube algorithm for paths of length 
$k_\varepsilon$.
Moreover, all the paths in ${\cal P}$ can be computed in parallel, because
 $|{\cal P}|$ is a constant depending only on $q$.
Since every path will contain variable $v$, we can compute the join of
all the paths in one final round with load $O(M/p)$.

The only difference for general connected queries is that $q$ may also contain
atoms that join vertices at distance $\text{rad}(q)$ from $v$ that are
not on any of the paths of length $\text{rad}(q)$ from $v$: these
can be covered using paths of length $\text{rad}(q)+1$ from $v$.
To get the final formula, we apply the equality 
$\lceil \log_{a}(b+1) \rceil = \lfloor \log_a(b) \rfloor +1$, 
which holds for positive integers $a,b$.
\end{proof}

As an application of the above lemma, \autoref{tab:complexity} shows
the number of  rounds required by different types of
queries.

\begin{table}
  \centering
  \begin{tabular}{|c|c|c|c|} \hline
  q & $\varepsilon$           & $r$                     & $r = f(\varepsilon)$ \\
    query    & space exponent       & rounds to achieve          & rounds/space \\ 
   &  for 1 round &  load $O(M/p)$   & tradeoff \\
    \hline

$C_k$ & $1-2/k$ & $\lceil \log k \rceil$ & $\sim \frac{\log k}{\log(2/(1-\varepsilon))}$ \\ \hline
$L_k$ & $1-\frac{1}{\lceil k/2 \rceil}$ & $\lceil \log k \rceil$ & $\sim \frac{\log k}{\log(2/(1-\varepsilon))}$ \\ \hline
$T_k$ & $0$ & 1 & NA \\ \hline
$SP_k$ & $1-1/k$ & 2 & NA \\ \hline
  \end{tabular}

  \caption{The tradeoff between space and communication rounds for several queries.}
  \label{tab:complexity}
\end{table}

\subsection{Lower Bounds for Multiple Rounds}
\label{sec:multi:lower}

To show lower bounds for the case of multiple rounds, we will need to
restrict the communication in the \mpc\ model; to do this, we define a
restriction of the \mpc\ model that we call the {\em tuple-based \mpc\
model.}

\subsubsection{Tuple-Based \mpc} 
\label{subsec:multiround:upper}
In the general \mpc\ model, we did not have any restrictions on the messages
sent between servers at any round. In the tuple-based \mpc\ model, we will impose
some structure on how we can communicate messages.

Let $I$ be the input database instance, $q$ be the query we want to compute, and
$\mA$ an algorithm. For a server $s \in [p]$, we denote by $\msg^1_{j \rightarrow s}(\mA,I)$ 
the message sent during round 1 by the input server for $S_j$ to the server $s$,
and by $\msg^k_{s \rightarrow s'}(\mA,I)$ the message sent to server $s'$  from server $s$ 
at round $k \geq 2$. Let $\msg^1_s(\mA,I) = (\msg^1_{1 \rightarrow s}(\mA,I), \ldots,
\msg^1_{\ell \rightarrow s}(\mA,I))$ and 
$\msg^k_s(\mA,I) = (\msg^k_{1 \rightarrow s}(\mA,I), \ldots,
\msg^k_{p \rightarrow s}(\mA,I))$ for any round $k \geq 2$.

Further, we define $\msg^{\leq k}_{s}(\mA,i)$ to be the vector of
messages received by server $s$ during the first $k$ rounds,
and $\msg^{\leq k}(\mA, i)=(\msg^{\leq k}_{1}(\mA,i),\ldots, \msg^{\leq k}_{p}(\mA,i))$.

Define a {\em join tuple} to be any tuple in $q'(I)$, where $q'$ is
any connected subquery of $q$.
An algorithm $\mA$ in the {\em tuple-based \mpc\ model} has the following
two restrictions on communication during rounds $k \geq 2$, for every server $s$
\begin{itemize}
\item the message $\msg^k_{s \rightarrow s'}(\mA,I)$ 
is a set of join tuples.
\item
for every join tuple $t$, the server $s$ decides whether to include
$t$ in $\msg^k_{s \rightarrow s'}(\mA,I)$
 based only on the parameters $t, s, s', r$, 
and the messages $\msg^1_{j \rightarrow s}(\mA,I)$ for all $j$ such that $t$ 
contains a base tuple in $S_j$.
\end{itemize}

The restricted model still allows unrestricted communication during
the first round; the information $\msg^1_s(\mA,I)$ received by server $s$ in the
first round is available throughout the computation.  However, during
the following rounds, server $s$ can only send messages consisting of
join tuples, and, moreover, the destination of these join tuples
can depend only on the tuple itself and on $\msg^1_s(\mA,I)$.  

The restriction of communication to join tuples (except for the first round
during which arbitrary, e.g. statistical, information can be sent) is 
natural and the tuple-based \mpc\ model captures a wide variety of algorithms
including those based on MapReduce.  
(Indeed, MapReduce is closer to an even more restricted version in which
communication in the first round is also limited to sending tuples.)
Since the servers can perform arbitrary inferences
based on the messages that they receive, even a limitation to messages that
are join tuples starting in the second round, without a restriction on how
they are routed, would still essentially have been equivalent to the fully
general \mpc\ model. For example, any server wishing to send a sequence of
bits to
another server can encode the bits using a sequence of tuples that the two
exchanged in previous rounds, or (with slight loss in efficiency) using the 
understanding that the tuples themselves are not important, but some
arbitrary fixed Boolean function of those tuples is the true message being
communicated.  This explains the need for the condition on routing tuples
that the tuple-based \mpc\ model imposes.

\subsubsection{A Lower Bound}

We present here a general lower bound for connected conjunctive queries in the
tuple-based \mpc\ model.

We first introduce a combinatorial object associated with every query $q$, 
called the $(\varepsilon,r)$-plan, which is central to the construction of the multi-round 
lower bound. We next define this notion, and also discuss how we can construct 
such plans for various classes of queries.

Given a query $q$ and a set $M\subseteq\text{atoms}(q)$, recall that
$q/M$ is the query that results from contracting the edges $M$
in the hypergraph of $q$. Also, we define  
$\contracted =   \text{atoms}(q) \setminus M$.

\begin{definition} \label{def:m} 
Let $q$ be a connected conjunctive
query.  A set $M\subseteq\text{atoms}(q)$ is {\em $\varepsilon$-good
for $q$} if it satisfies the following two properties:
\begin{enumerate}
\item Every connected subquery of $q$ that is in $\Gamma_{\varepsilon}^1$
contains at most one atom in $M$. 
\item $\chi(\contracted) = 0$ (and thus $\chi(q/ \contracted)=\chi(q)$
 by~\autoref{lemma:chi}).
\end{enumerate}

For $\varepsilon \in [0,1)$ and integer $r \geq 0$, 
an {\em $(\varepsilon,r)$-plan} $\cal M$ is a sequence
$M_1,\ldots, M_r$, with $M_0=\text{atoms}(q)\supset M_1
\supset \cdots M_r$ such that
(a) for $j = 0, \dots, r-1$, $M_{j+1}$ is $\varepsilon$-good for
$q/\contracted_j$,  and
(b) $q/\contracted_r\;\notin \Gamma^1_\varepsilon$.
\end{definition}

We provide some intuition about the above definition with the next two lemmas,
which shows how we can obtain such a plan for the query $L_k$ and $C_k$
respectively.

\begin{lemma} \label{lem:er-plan:line}
The query $L_k$ admits an $(\varepsilon, \lceil \log_{k_{\varepsilon}}(k) \rceil-2)$-plan
for any integer $k > k_\varepsilon = 2 \lfloor 1/(1-\varepsilon)\rfloor$.
\end{lemma}

\begin{proof}
We will prove using induction that for every integer $r \geq 0$, 
if $k \ge k_{\varepsilon}^{r+1}+1$ 
then $L_k$ admits an $(\varepsilon, r)$-plan.
This proves the lemma, because then for a given $k$ the smallest
integer $r$ we can choose for the plan is
$r = \lfloor \log_{k_{\varepsilon}}( k-1) \rfloor-1 = \lceil \log_{k_{\varepsilon}}(k) \rceil-2$.
 For the base case $r=0$, we have
that $k \geq k_\varepsilon+1$, and observe that 
$L_k/\contracted_0 = L_k \notin \Gamma^1_\varepsilon$. 
  
For the induction step, let $k_0 \ge  k_\varepsilon^{r+1}+1$; 
then from the inductive hypothesis we have that for every
$k \ge k_\varepsilon^{r} +1$ 
the query $L_k$ has an $(\varepsilon, r-1)$-plan.
Define $M$ to be the set of atoms where we include every 
$k_{\varepsilon}$-th atom $L_{k_0}$, starting with $S_1$; in other words, 
$M = \{ S_1, S_{k_\varepsilon+1}, S_{2k_\varepsilon+1},\ldots \}$. 
Observe now that 
$L_{k_0}/\contracted = S_1(x_0,x_1), S_{k_\varepsilon+1}(x_1,
x_{k_\varepsilon+1}), S_{2k_\varepsilon+1}(x_{k_\varepsilon+1},
x_{2k_\varepsilon+1}), \ldots$, which is isomorphic to 
$L_{\lceil k_0/k_{\varepsilon} \rceil}$. 

We will show first that $M$ is $\varepsilon$-good for $L_{k_0}$.
Indeed, $\chi(L_{k_0}/\contracted ) = \chi(L_{\lceil k_0/k_{\varepsilon}}) = \chi(L_{k_0})$ and
thus property (2) is satisfied.
Additionally, recall that $\Gamma^1_\varepsilon$ consists of queries for
which $\tau^*(q) \leq 1/(1-\varepsilon)$; thus the connected subqueries of
$L_{k_0}$ that are in $\Gamma^1_\varepsilon$ are precisely queries of
the form $S_j(x_{j-1},x_j),S_{j+1}(x_j,x_{j+1}),\ldots,
S_{j+k-1}(x_{j+k-2},x_{j+k-1})$, where $k \leq k_\varepsilon$.  
By choosing $M$ to contain every $k_\varepsilon$-atom, no such subquery in
$\Gamma^1_\varepsilon$ will contain more than one atom from $M$ and thus
property (1) is satisfied as well.

Finally, we have that $\lceil k_0/k_{\varepsilon} \rceil \geq \lceil k_{\varepsilon}^r + 1/  k_{\varepsilon} \rceil = k_{\varepsilon}^r +1$
and thus from the inductive hypothesis the query $L_{k_0}/\contracted$  admits an
$(\varepsilon, r-1)$-plan. By definition, this implies a sequence $M_1, \dots, M_{r-1}$;
the extended sequence $M, M_1, \dots, M_{r-1}$ will now be an $(\varepsilon, r)$-plan for $L_{k_0}$.
\end{proof}

\begin{lemma} \label{lem:er-plan:cycle}
The query $C_k$ admits an 
$(\varepsilon, \lfloor \log_{k_{\varepsilon}} (k/(m_{\varepsilon}+1)) \rfloor )$-plan,
for every integer $k > m_{\varepsilon} = \lfloor  2/(1-\varepsilon)\rfloor$.
\end{lemma}

\begin{proof}
The proof is similar to the proof for the query $L_k$, since we can
observe that any set $M$ of atoms that are (at least)
$k_\varepsilon$ apart along any cycle $C_k$ is a $\varepsilon$-good
set for $C_k$ and further $C_k/\contracted$ is isomorphic to $C_{\lfloor
k/k_\varepsilon\rfloor}$. 
The only difference is that the base case for $r=0$ is that $k \geq m_{\varepsilon}+1$.
Thus, the inductive step is that for every integer $r \geq 0$, if 
$k \ge k_{\varepsilon}^r (m_{\varepsilon}+1)$ 
then $C_k$ admits an $(\varepsilon, r)$-plan.
\end{proof}

The above examples of queries show how we can construct $(\varepsilon, r)$-plans.
We next present the main theorem of this section, which tells us how we can use
such plans to obtain lower bounds on the number of communication rounds needed
to compute a conjunctive query.

\begin{theorem}[Lower Bound for Multiple Rounds] 
\label{th:multiround} 
Let $q$ be a conjunctive query that admits an $(\varepsilon,r)$-plan. 
For every randomized algorithm in the tuple-based \mpc\ model that
computes $q$ in $r+1$ rounds and with load $L \leq c M/p^{1-\varepsilon}$ 
for a sufficiently  small constant $c$, there exists an instance $I$ with relations of size $M$
where the algorithm fails to compute $q$ with probability $\Omega(1)$.
\end{theorem}

The constant $c$ in the above theorem depends on the query $q$ and
the parameter $\varepsilon$. To state the precise expression the constant $c$, 
we need some additional definitions.

\begin{definition}
Let $q$ be a conjunctive query and $\cal M$ be an $(\varepsilon,r)$-plan for $q$.
We define $\tau^*(\cal M)$ to be the minimum of
$\tau^*(q/\contracted_r)$ and  $\tau^*(q')$, where $q'$
ranges over all connected subqueries of $q/\contracted_{j-1}$, $j\in
[r]$, such that $q' \not\in \Gamma^1_\varepsilon$.
\end{definition}

\begin{proposition}
Let $q$ be a conjunctive query and $\cal M$ be an $(\varepsilon,r)$-plan for $q$.
Then, $\tau^*(\mathcal{M})>1/(1-\varepsilon)$.
\end{proposition}

\begin{proof}
For every $q' \not\in \Gamma^1_\varepsilon$, we have by definition that 
$\tau^*(q')> 1/(1-\varepsilon)$. Additionally, by the definition of an 
$(\varepsilon,r)$-plan, we have that $\tau^*(q/\contracted_r)
>1/(1-\varepsilon)$.
\end{proof}

Further, for a given query $q$ let us define the following sets:
\begin{align*}
\mathcal{C}(q) & = \setof{q'}{q' \text{ is a connected subquery of } q} \\
\mathcal{C}_\varepsilon(q) & = \setof{q'}{q' \notin \Gamma^1_{\varepsilon},\  q' \text{ is a connected subquery of } q} \\
\mathcal{S}_\varepsilon(q) & = \setof{q'}{q' \notin \Gamma^1_{\varepsilon},\  
q' \text{ is a minimal connected subquery of } q}.
\end{align*}
and let
\begin{align*}
\beta(q,\mathcal{M})=
\left( \frac{1}{\tau^*(q/\contracted_r)} \right)^{\tau^*(\mathcal{M})} +
\sum_{k=1}^{r} \sum_{q' \in \mathcal{S}_\varepsilon(q/\contracted_{k-1})} 
\left( \frac{1}{\tau^*(q')} \right)^{\tau^*(\mathcal{M})}
\end{align*}

We can now present the precise statement.

\begin{theorem} \label{th:strong-multiround} 
If $q$ has an $(\varepsilon,r)$-plan $\cal M$ then any deterministic 
tuple-based \mpc\ algorithm running in $r+1$ rounds with maximum load $L$
reports at most 
$$ \beta(q,\mathcal{M})\cdot \left( \frac{(r+1)L}{M} \right)^{\tau^*(\mathcal{M})} p \cdot  \E[|q(I)|] $$
correct answers in expectation over a uniformly at random chosen matching database $I$
where each relation has size $M$.
\end{theorem}

The above theorem implies~\autoref{th:multiround} by following the same proof as
in~\autoref{thm:probability-LB}. Indeed, for $L \leq c M / p^{1-\varepsilon}$
we obtain that the output tuples will be at most $f \cdot  \E[|q(I)|] $, where
$f = \beta(q,\mathcal{M})\cdot ((r+1) c)^{\tau^*(\mathcal{M})}$. If we choose 
the constant $c$ such that $f < 1/9$, we can apply \autoref{lem:instance:choice}
to show that for any randomized algorithm we can find an instance $I$ where it
will fail to produce the output with probability $\Omega(1)$.

In the rest of this section, we present the proof of~\autoref{th:strong-multiround}. Let $\mA$
be an algorithm that computes $q$ in $r+1$ rounds. The
intuition is as follows. Consider an $\varepsilon$-good set $M$; then any
matching database $i$ consists of two parts, $i=(i_M, i_\contracted)$,\footnote{We 
will use $i$ to denote a fixed matching instance, as opposed to $I$ that denotes a
random instance.}
where $i_M$ are the relations for atoms in $M$, and $i_\contracted$
are all the other relations. We show that, for a fixed instance
$i_\contracted$, the algorithm can be used to compute
$q/\contracted(i_M)$ in $r+1$ rounds; however, the first round is almost
useless, because the algorithm can discover only a tiny number of join
tuples with two or more atoms $S_j \in M$ (since every subquery $q'$
of $q$ that has two atoms in $M$ is not in $\Gamma^1_\varepsilon$).  This
shows that the algorithm can compute most of the answers in $q/\contracted(i_M)$ 
in only $r$ rounds, and we repeat the argument until a one-round algorithm remains.

To formalize this intuition, we need some notation.  
For two relations $A,B$ we write $A \ltimes B$, called
the {\em semijoin},  to denote the set of tuples in $A$ for which there is a tuple
in $B$ that has equal values on their common variables. We also write $A \rhd B$, called
the {\em antijoin}, to denote the set of tuples in $A$ for which no tuple in $B$ has equal
values on their common variables.

Let $\mA$ be a deterministic algorithm with $r+1$ rounds, $k\in [r+1]$ a
round number, $s$ a server, and $q'$ a subquery of $q$.  We define:
\begin{align*}
    K^{\mA,k, s}_\msg(q') & =  \{\ba' \in [n]^{\text{vars}(q')}\mid \forall \mbox{ matching database }i,
      \msg^{\leq k}_{s}(\mA, i)=\msg \Rightarrow \ba' \in q'(i)\}\\
    K^{\mA,k}_\msg(q') & = \bigcup_{s=1}^p K^{\mA,k,s}_{\msg_s}(q')
\end{align*}
Using the above notation, $K^{\mA,k,s}_{\msg^{\leq k}_{s}(\mA,i)}(q')$ and $K^{\mA,k}_{\msg^{\leq k}(\mA,i)}(q')$ denote the set of join tuples from $q'$ known at round $k$ by server $s$, 
and by all servers, respectively, on input $i$.
Further, $ \mA(i) =K^{\mA,r+1}_{\msg^{\leq r+1}(\mA,i)}(q)$ is w.l.o.g. the final 
answer of the algorithm $\mA$ on input $i$.
Finally, let us define 
\begin{align*}
J^{\mA,q}(i)&= \bigcup_{q' \in \mathcal{C}(q)} K^{\mA,1}_{\msg^{\leq 1}(\mA,i)}(q') \\
J^{\mA,q}_\varepsilon(i)&= \bigcup_{q' \in \mathcal{C}_\varepsilon(q)} K^{\mA,1}_{\msg^{\leq 1}(\mA,i)}(q') 
\end{align*}
$J^{\mA,q}_\varepsilon(i)$ is precisely the set of join tuples known
after the first round, but restricted to those that correspond to subqueries that are
not computable in one round; thus, the number of tuples in
$J^{\mA,q}_\varepsilon(i)$ will be small. 

We can now state the two lemmas we need as building blocks  to
prove \autoref{th:strong-multiround}.

\begin{lemma} \label{lemma:contraction} 
Let $q$ be a query, and $M$ be
  any $\varepsilon$-good set for $q$.  If $\mA$ is an algorithm with
  $r+1$ rounds for $q$, then for any matching database $i_\contracted$
  over the atoms of $\contracted$, there exists an algorithm $\mA'$ with
  $r$ rounds for $q/\contracted$ 
  using the same number of processors and the same total number of bits
  of communication
  received per processor such that, for every matching
  database $i_M$ defined over the atoms of $M$:
  \begin{align*}
    |\mA(i_M, i_\contracted)|\le |q(i_M,i_\contracted) \ltimes
    J^{\mA,q}_\varepsilon(i_M,i_\contracted)|+ |\mA'(i_M)|.
  \end{align*}
\end{lemma}

In other words, the algorithm returns no more answers than the (very
few) tuples in $J^{\mA,q}_\varepsilon$, plus what another algorithm $\mA'$ 
that we define next computes for $q/\contracted$ using {\em one fewer} round.

\begin{proof} 
We call $q/\contracted$ the {\em contracted} query. 
While the original query $q$ takes as input the complete
database $i = (i_M, i_\contracted)$, the input to the contracted
query is only $i_M$. Observe also that for different matching databases
$i_{\contracted}$, the lemma produces different algorithms $\mA'$.
We fix now a matching database $i_{\contracted}$.

The construction of $\mA'$ is based on the following two constructions, which we
call {\em contraction} and {\em retraction}. \\

\noindent {\bf Contraction.}  We first show how to use the algorithm $\mA$
to derive an algorithm $\mA^{c}$ for $q/\contracted$ that uses the same
number of rounds as $\mA$.

For each connected component $q_c$ of $\contracted$, we choose a
representative variable $z_c\in\text{vars}(q_c)$.
The query answer $q_c(i_\contracted)$  is a
matching instance, since $q_c$ is tree-like (because $\chi(\contracted) = 0$).  
Denote $\bm{\sigma} = \setof{\sigma_x}{x \in \vars{q}}$, where, for every variable
$x \in \vars{q}$, $\sigma_x : [n] \rightarrow [n]$ is the following permutation. 
If $x \not \in \vars{\contracted}$, then $\sigma_x$ is defined as
the identity, \ie $\sigma_x(a) = a$ for every $a \in [n]$.
Otherwise, if $q_c$ is the unique connected component such that $x \in \vars{q_c}$
and $\ba \in q_c(i_\contracted)$ is the unique tuple such that $\ba_x = a$, 
we define $\sigma_x(a) = \ba_{z_c}$.
In other words, we think of $\bm{\sigma}$ as permuting the domain of each variable 
$x \in \vars{q}$. Observe that $\bm{\sigma}$ is known to all servers, since $i_\contracted$
is a fixed instance.

It holds that $\bm{\sigma}(q(i)) = q(\bm{\sigma}(i))$, and
$\bm{\sigma}(i_\contracted) = \mathbf{id}_\contracted$, 
where $\mathbf{id}_\contracted$ is the identity
matching database (where each relation in $\contracted$ is
$\set{(1,1,\ldots),(2,2,\ldots),\ldots}$). 
Therefore,\footnote{We assume $\text{vars}(q/\contracted) \subseteq \text{vars}(q)$;
for that, when we contract a set of nodes of the hypergraph, we
replace them with one of the nodes in the set.}
\begin{align*}
q/\contracted(i_M) = \bm{\sigma}^{-1}(\Pi_{\vars{q/\contracted}}(q(\bm{\sigma}(i_M),\mathbf{id}_\contracted)))
\end{align*}

Using the above equation, we can now define the algorithm $\mA^c$ that computes
the query $q/\contracted(i_M)$.  First, each
input server for $S_j \in M$ replaces $S_j$ with $\bm{\sigma}(S_j)$. 
Second, we run $\mA$ unchanged, substituting all
relations $S_j \in \contracted$ with the identity. Finally, we apply
$\bm{\sigma}^{-1}$ to the answers and return the output. Hence, we have:
\begin{align}
      \mA^c(i_M) = \bm{\sigma}^{-1}(\Pi_{\vars{q/\contracted}}(\mA(\bm{
      \sigma}(i_M), \mathbf{id}_\contracted))) \label{eq:contracted}
\end{align}

\noindent  {\bf Retraction.} Next, we transform $\mA^c$ into a new algorithm
$\mA^r$, called the {\em retraction} of $A^c$, that takes as input $i_{M}$ as
follows.
\begin{packed_item}
\item In round 1, each input server for $S_j$ sends
(in addition to the messages sent by $\mA^c$) every tuple in $\ba_j \in
  S_j$ to all servers $s$ that eventually receive $\ba_j$.  In other
  words, the input server sends $t$ to every $s$ for which there
  exists $k \in [r+1]$ such that $\ba_j \in
  K^{\mA^c,k,s}_{\msg^{\leq k}_{s}(\mA^c, i_M)}(S_j)$. This is possible because of
  the restrictions in the tuple-based \mpc\
  model: all destinations of  $\ba_j$ depend only on $S_j$,
  and hence can be computed by the input server.  
  Note that this does not increase the total number of bits received by any
  processor, though it does mean that more communication will be performed
  during the first round.
\item  In round $2$, $\mA^r$ sends {\em no tuples}. 
\item In rounds $k \geq 3$, $\mA^r$ sends a join tuple $t$ from server $s$ to server $s'$ 
if server $s$ {\em knows} $t$ at round $k$, and also algorithm $\mA^c$ sends 
$t$ from $s$ to $s'$ at round $k$.
\end{packed_item}

Observe first that the algorithm $\mA^r$ is correct, in the sense that the output $\mA^r(i_M)$
will be a subset of $q/\contracted(i_M)$. We now need to quantify how many tuples $\mA^r$
misses compared to the contracted algorithm $\mA^c$.
Let $\mathcal{Q}_M = \setof{q'}{ q' \mbox{ subquery of } q/\contracted, |q'| \geq 2 }$, and
define:
$$J_+^{\mA^c}(i_M) = \bigcup_{q' \in \mathcal{Q}_M} K^{\mA^c,1}_{\msg^1(\mA^c,i)}(q').
$$

The set $J_+^{\mA^c}(i_M)$ is exactly the set of non-atomic tuples known 
by $\mA^c$ right after round 1: these are also the tuples that the new algorithm $\mA^r$
will choose not to send during round 2. 

\begin{lemma}
$(\mA^c(i_M) \rhd J_+^{\mA^c}(i_M)) \subseteq \mA^r(i_M)$
\end{lemma}

\begin{proof}
We will prove the statement by induction on the number of rounds: 
for any subquery $q'$ of $q/\contracted$,
if server $s$ knows $t \in (q'(i_M) \rhd J_+^{\mA^c}(i_M))$ at round $k$ for algorithm $\mA^c$,
then server $s$ knows $t$ at round $k$ for algorithm $\mA^r$ as well.

For the induction base, in round 1 we have by construction that 
$K^{\mA^c,1,s}_{\msg^{1}_s(\mA^c, i_M)}(S_j)
\subseteq K^{\mA^r,1,s}_{\msg^{1}_s(\mA^r, i_M)}(S_j)$ for every $S_j \in M$, 
and thus any tuple $t$ (join or atomic) 
that is known by server $s$ for algorithm $\mA^c$ will be also known for algorithm $\mA^r$.

Consider now some round $k+1$ and a tuple $t \in (q'(i_M) \rhd J_+^{\mA^c}(i_M))$ known
by server $s$ for algorithm $\mA^c$. If $q'$ is a single relation, the statement is correct since 
by construction all atomic tuples are known at round 1 for algorithm $\mA^r$. 
Otherwise $q' \in \mathcal{Q}_M$. 
Let $t_1, \dots, t_m$ be the subtuples at server $s$ from which tuple $t$ is constructed, where
$t_j \in q_j(i_M)$ for every $j = 1, \dots, m$. Observe that $t_j \in (q_j(i_M) \rhd J_+^{\mA^c}(i_M))$. Thus, if $t_i$ was known at round $k$ by some server $s'$ for algorithm $\mA^c$, 
by the induction hypothesis it would be known by server $s'$ for algorithm $\mA^r$ as well,
and thus it would have been communicated to server $s$ at round $k+1$.
\end{proof}


From the above lemma it follows that:
\begin{align}
\mA^c(i_M) \subseteq \mA^r(i_M) \cup (q/\contracted(i_M) \ltimes J_+^{\mA^c}(i_M))
\label{eq:retract}
\end{align}
Additionally, by the definition of $\varepsilon$-goodness, if a subquery
$q'$ of $q$ has two atoms in $M$, then $q' \not\in \Gamma^1_\varepsilon$. 
Hence, we also have:
\begin{align}
J_+^{\mA^c}(i_M) \subseteq  \bm{\sigma}^{-1}(\Pi_{\vars{q/\contracted}}(J_\varepsilon^{\mA,q}(\bm{\sigma}(i))))
\label{eq:extra:bound}
\end{align}

Since $\mA^r$ send no information during the second round, we can compress it
to an algorithm $\mA'$ that uses only $r$ rounds.
Finally, since $M$ is $\varepsilon$-good, we have $\chi(q/ \contracted) = \chi(q)$
and thus $|\mA^c(i_M)| = |\mA(i_M, i_\contracted)|$. Combining everything together:
\begin{align*}
|\mA(i_M, i_\contracted)| & = |\mA^c(i_M)| \\
& \leq | \mA^r(i_M)| + |(q/\contracted(i_M) \ltimes J_+^{\mA^c}(i_M)) | \\
& \leq | \mA'(i_M)| + |(q/\contracted(i_M) \ltimes  \bm{\sigma}^{-1}(\Pi_{\vars{q/\contracted}}(J_\varepsilon^{\mA,q}( \bm{\sigma}(i))))
 | \\
&  \leq | \mA'(i_M)| + | \Pi_{\text{vars}(q/\contracted)}(q(i) \ltimes  J_\varepsilon^{\mA,q}(i))| \\
& \leq \mA'(i_M)| + | q(i) \ltimes  J_\varepsilon^{\mA,q}(i)| 
\end{align*}
This concludes the proof. \qed
\end{proof}

\begin{lemma}
\label{lemma:onebyp}
Let $q$ be a conjunctive query and $q'$ a subquery of $q$.  
Let $\mB$ be any algorithm that outputs a subset of answers to $q'$
(\ie for every database $i$, $\mB(i) \subseteq q'(i)$). Let $I$ be a 
uniformly at random chosen matching database for $q$, and $I' = I_{atoms(q')}$ 
its restriction over the atoms in $q'$. 

If  $\E[|\mB(I')|] \leq \gamma \cdot \E[|q'(I')|]$, then $\E[|q(I) \ltimes \mB(I')|] \leq \gamma \cdot \E[|q(I)|]$.
\end{lemma}

\begin{proof}
Let $\mathbf{y} = \vars{q'}$ and $d = |\mathbf{y}|$.
By symmetry, the quantity
$\E[|\sigma_{\mathbf{y} = \ba}(q(I))|]$ is independent of $\ba$,
and therefore equals $\E[|q(I)|]/n^d$.  Notice that by construction
$\sigma_{\mathbf{y}=\ba}(\mB(i')) \subseteq \set{\ba}$.  We now have:
\begin{align*}
\E[|q(I) \ltimes B(I')|] 
& = \sum_{\ba \in [n]^d} \E[|\sigma_{\mathbf{y} = \ba}(q(I)) \ltimes \sigma_{\mathbf{y}=\ba}(\mB(I'))|]\\
& =  \sum_{\ba \in [n]^d}\E[|\sigma_{\mathbf{y} = \ba}(q(I))|] \cdot P(\ba \in \mB(I'))\\
& =  \left( \E[|q(I)|] /n^d \right)  \cdot \sum_{\ba \in [n]^d}\P(\ba \in \mB(I')) \\
& = \E[|q(I)|]\cdot \E[|\mB(I')|]/n^d \\
& \leq \E[|q(I)|]\cdot (\gamma \cdot \E[|q'(I')|]) /n^d \\
& \leq \gamma \cdot \E[|q(I)|]
\end{align*}
where the last inequality follows from the fact that $\E[|q'(I')|] \leq n^d$ (since for every database
$i'$, we have $|q'(i')| \leq n^d$).
\end{proof}

\begin{proof}[ of \autoref{th:strong-multiround}]
Given an $(\varepsilon,r)$-plan $\text{atoms}(q)
=M_0\supset M_1 \supset \ldots\supset M_r$, we define $\hat
M_k=\contracted_k-\contracted_{k-1}$, for $k \ge 1$.  
Let $\mA$ be an algorithm for $q$ that uses $(r+1)$ rounds.

We start by applying \autoref{lemma:contraction} for algorithm $\mA$
and the $\varepsilon$-good set $M_1$. Then, for every matching
database $i_{\bar M_1} = i_{\hat M_1}$, there exists an algorithm $\mA^{(1)}_{i_{\hat M_1}}$
for $q / \hat M_1$ that runs in $r$ rounds such that for every matching database 
$i_{M_1}$ we have:
\begin{align*}
|\mA(i)| \leq 
|q(i) \ltimes J^{\mA,q}_\varepsilon(i)| + |\mA^{(1)}_{i_{\hat M_1}}(i_{M_1})|
\end{align*}
We can iteratively apply the same argument. For $k=1, \dots, r-1$, let us denote 
$\mathcal{B}^{k} = \mA^{(k)}_{i_{\contracted_k}}$ the inductively defined 
algorithm for query $q/\contracted_k$, and consider the $\varepsilon$-good set $M_{k+1}$. 
Then, for every matching database $i_{\hat M_{k+1}}$ there exists an algorithm 
$\mathcal{B}^{k+1} = \mA^{(k+1)}_{i_{\contracted_{k+1}}}$ for $q/\contracted_{k+1}$ 
such that for every matching database $i_{M_{k+1}}$, we have:
\begin{align*}
|\mB^{k}(i_{M_k})| \leq 
|q/{\contracted_k}(i_{M_k}) \ltimes J^{\mB^k,q/\contracted_k}_\varepsilon(i_{M_k})| + 
|\mB^{k+1}(i_{M_{k+1}})| 
\end{align*}
We can now combine all the above inequalities for $k=0, \dots, r-1$ to obtain: 
\begin{align}
|\mA(i)| 
& \leq |q(i) \ltimes J^{\mA,q}_\varepsilon(i_{M_r}, i_{\hat M_1}, \ldots, i_{\hat M_r})| \nonumber\\
&+ |q/\contracted_{1}(i_{M_1}) \ltimes J^{\mB^1,q/\contracted_1}_\varepsilon(i_{M_r}, i_{\hat M_2}, \ldots, i_{\hat M_r})| \nonumber\\
  & + \ldots \nonumber\\
  & + |q/\contracted_{r-1}(i_{M_{r-1}}) \ltimes J^{\mB^{r-1},q/\contracted_{r-1}}_\varepsilon(i_{M_r}, i_{\hat M_r})| \nonumber\\
  & +|\mB^r(i_{M_r})|\label{eq:sum}
\end{align}

We now take the expectation of~\eqref{eq:sum} over a uniformly chosen matching database $I$
and upper bound each of the resulting terms.
Observe first that for all $k =0, \dots, r$ we have $\chi(q/\contracted_k)=\chi(q)$, 
and hence, by \autoref{lem:expected_size}, we have
$\E[|q(I)|]=\E[|(q/\contracted_k)(I_{M_k})|]$.

We start by analyzing the last term of the equation, which is the expected output
of an algorithm $\mB^{r}$ that uses one round to compute $q/\contracted_r$.
By the definition of $\tau^*(\mathcal{M})$, we have $\tau^*(q/\contracted_r)\ge \tau^*(\mathcal{M})$.
Since the number of bits received by each processor in the first round of
algorithm $\mB^r$ is at most $r+1$ times the bound for the original algorithm $\mA$,
we can apply \autoref{th:lower:uniform} to obtain that:
\begin{align*}
\E[\mB^r(I_{M_r})] & \leq 
p \left( \frac{(r+1)L}{\tau^*(q/\contracted_r) M} \right)^{\tau^*(q/\contracted_r)}  \E[|(q/\contracted_r)(I_{M_r})| \\
& \leq  p \left( \frac{(r+1)L}{\tau^*(q/\contracted_r) M} \right)^{\tau^*(\mathcal{M})}  \E[|q(I)||
\end{align*}

We next bound the remaining terms.
Note that $I_{M_{k-1}}=(I_{M_r},I_{\hat M_k},\ldots, I_{\hat M_{r}})$
and consider the expected number of tuples in 
$J = J^{\mB^{k-1},q/\contracted_{k-1}}_\varepsilon(I_{M_{k-1}})$.  The algorithm
$\mB^{k-1}=\mA^{(k-1)}_{I_{\contracted_{k-1}}}$ itself depends on the choice
of $I_{\contracted_{k-1}}$; still, we show that $J$ has a small number
of tuples.  Every subquery $q'$ of $q/\contracted_{k-1}$ that
is not in $\Gamma^1_\varepsilon$ (and hence contributes to $J$) has
$\tau^*(q')\ge \tau^*(\mathcal{M})$.  For
each fixing $I_{\contracted_{k-1}}=i_{\contracted_{k-1}}$, the
expected number of tuples produced for subquery $q'$ by $B_{q'}$ ,
where $B_{q'}$ is the portion of the first round of
$\mA^{(k-1)}_{i_{\contracted_{k-1}}}$ that produces tuples for $q'$,
satisfies $\E[|B_{q'}(I_{M_{k-1}})|]
\leq \gamma(q') \cdot \E[|q'(I_{M_{k-1}})|] $, where
$$ \gamma(q') = p \left( \frac{(r+1)L}{\tau^*(q') M} \right)^{\tau^*(\mathcal{M})} $$
since each processor in a round of $\mA^{(k-1)}_{i_{\contracted_{k-1}}}$  (and
hence $B_{q'}$) receives at most $r+1$ times the communication bound for 
a round of $A$. We now apply \autoref{lemma:onebyp} to derive 
\begin{align*}
\E[|  q(I)\ltimes B_{q'}(I_{M_{k-1}})|]
 & = \E[|(q/\contracted_{k-1})(I_{M_{k-1}})\ltimes B_{q'}(I_{M_{k-1}})|]\\
&\leq \gamma(q') \cdot \E[|(q/\contracted_{k-1})(I_{M_{k-1}})|] \\
& = \gamma(q') \cdot \E[|q(I)|].
\end{align*}
Averaging over all choices of $I_{\contracted_{k-1}}=i_{\contracted_{k-1}}$
and  summing over the number of different queries $q' \in \mathcal{S}(q/\contracted_{k-1})$,
where we recall that  $\mathcal{S}_\varepsilon(q/\contracted_{k-1})$ is the set of all
minimal connected subqueries $q'$ of $q/\contracted_{k-1}$ that are not
in $\Gamma^1_\varepsilon$,
we obtain 
\begin{align*}
\E[|q(I) \ltimes &J^{A^{k-1},q/ \contracted_{k-1}}_\varepsilon(I_{M_{k-1}})|]
\leq \sum_{q' \in \mathcal{S}_\varepsilon(q/\contracted_{k-1})} \gamma(q') \cdot \E[|q(I)|]
\end{align*}
Combining the bounds obtained for the $r+1$ terms in~\eqref{eq:sum}, we conclude
that
\begin{align*}
\E[|A(I)|]&\leq \left( \left( \frac{1}{\tau^*(q/\contracted_r)} \right)^{\tau^*(\mathcal{M})} +
\sum_{k=1}^{r} \sum_{q' \in \mathcal{S}_\varepsilon(q/\contracted_{k-1})} 
\left( \frac{1}{\tau^*(q')} \right)^{\tau^*(\mathcal{M})} \right)\\
&~\times
\left( \frac{(r+1)L}{M} \right)^{\tau^*(\mathcal{M})} p \cdot  \E[|q(I)|] \\
& = \beta(q,\mathcal{M})\cdot \left( \frac{(r+1)L}{M} \right)^{\tau^*(\mathcal{M})} p \cdot  \E[|q(I)|] 
\end{align*}
which proves \autoref{th:strong-multiround}.
\end{proof}

\subsection{Application of the Lower Bound}

We show now how to apply~\autoref{th:multiround} to obtain lower bounds
for several query classes, and compare the lower bounds with the upper bounds.

The first class is the queries $L_k$, where the following corollary is a 
straightforward application of~\autoref{th:multiround} and \autoref{lem:er-plan:line}.

\begin{corollary} \label{lemma:lk:lower} 
Any tuple-based  \mpc\ algorithm that computes $L_k$ with load 
$L = O(M/p^{1-\varepsilon})$ requires at least 
$\lceil \log_{ k_{\varepsilon}} k \rceil$ rounds of computation.
\end{corollary}

Observe that this gives a tight lower bound for $L_k$, since in the previous
section we showed that there exists a query plan with depth 
$\lceil \log_{ k_{\varepsilon}} k \rceil$ and load $O(M/p^{1-\varepsilon})$.

Second, we give a lower bound for tree-like queries, and for that we
use a simple observation:
\begin{proposition} \label{prop:tree-like}
Let $q$ be a tree-like query, and $q'$ be any connected subquery of $q$.
Any algorithm that computes $q'$ with load $L$ needs at least as many
rounds to compute $q$ with the same load.
\end{proposition}
\begin{proof}
  Given any tuple-based \mpc\ algorithm $A$ for
  computing $q$ in $r$ rounds  with maximum load $L$, 
  we construct a tuple-based \mpc\ algorithm $A'$ that computes $q'$ in $r$ 
  rounds with at most load $L$.
  $A'$ will interpret each instance over $q'$ as part of
  an instance for $q$ by using the relations in $q'$ and using the
  identity permutation ($S_j =
  \set{(1,1,\ldots),(2,2,\ldots),\ldots}$) for each relation in $q
  \setminus q'$.  Then, $A'$ runs exactly as $A$ for $r$ rounds; after
  the final round, $A'$ projects out for every tuple all the variables
  not in $q'$.  The correctness of $A'$ follows from the fact that $q$
  is tree-like.
\end{proof}

Define $\text{diam}(q)$, the {\em diameter} of a query $q$, to be the
longest distance between any two nodes in the hypergraph of $q$.  In
general, $\text{rad}(q) \leq \text{diam}(q) \leq 2\ \text{rad}(q)$.
For example, $\text{rad}(L_k) = \lfloor k/2 \rfloor$,
$\text{diam}(L_k)=k$ and $\text{rad}(C_k) = \text{diam}(C_k)=\lfloor
k/2 \rfloor$. 

\begin{corollary}
  Any tuple-based \mpc\ algorithm that computes a
  tree-like query $q$ with load $L = O(M/p^{1-\varepsilon})$ needs at least 
  $\lceil \log_{k_{\varepsilon}} (\text{diam}(q)) \rceil$ rounds.
\end{corollary}

\begin{proof}
Let $q'$ be the subquery of $q$ that corresponds to the diameter of $q$.
Notice that $q'$ is a connected query, and moreover, it behaves exactly
like $L_{\text{diam}(q)}$. Hence, by~\autoref{lemma:lk:lower} any algorithm
needs at least $\lceil \log_{k_{\varepsilon}} (\text{diam}(q)) \rceil$ to compute
$q'$. By applying~\autoref{prop:tree-like}, we have that $q$ needs at least
that many rounds as well.
\end{proof}

Let us compare the lower bound $r_{\text{low}}=\lceil
\log_{k_{\varepsilon}} (\text{diam}(q)) \rceil$ and the upper bound
$r_{\text{up}}=\lceil \log_{k_{\varepsilon}} (\text{rad}(q)) \rceil
+1$ from~\autoref{lem:multitree}.
Since $\text{diam}(q) \leq 2\text{rad}(q)$, we have that
$r_{\text{low}} \leq r_{\text{up}}$. Additionally, $\text{rad}(q)
\leq \text{diam}(q)$ implies $r_{\text{up}} \leq r_{\text{low}} +1$.
Thus, the gap between the lower bound and the upper bound  on the number
of rounds is at most 1 for tree-like queries.
When $\varepsilon < 1/2$, these bounds
are matching, since $k_\varepsilon = 2$ and $2\text{rad}(q)-1\leq
\text{diam}(q)$ for tree-like queries.  


Third, we study one instance of a non tree-like query, namely the cycle
query $C_k$. The lemma is a direct application of~\autoref{lem:er-plan:cycle}.

\begin{lemma}
Any tuple-based \mpc\ algorithm that computes the query $C_k$
with load $L = O(M/p^{1-\varepsilon})$ 
requires at least $ \lfloor \log_{k_{\varepsilon}} (k/(m_{\varepsilon}+1))
\rfloor + 2$ rounds, where $m_{\varepsilon} = \lfloor
2/(1-\varepsilon)\rfloor$.
\end{lemma}

For cycle queries we also have a gap of at most 1 between this lower bound and the upper
bound in \autoref{lem:multitree}. 

\begin{example}
Let $\varepsilon=0$ and consider two queries, $C_5$ and $C_6$. In this case,
we have $k_\varepsilon=m_\varepsilon=2$, and $\text{rad}(C_5) = \text{rad}(C_6) = 2$.

For query $C_6$, the lower bound is then $\lfloor \log_2 (6/3) \rfloor + 2 = 3$ rounds,
while the upper bound is $\lceil \log_2 (3) \rceil + 1= 3$ rounds. Hence, in the case of
$C_6$ we have tight upper and lower bounds.
For query $C_5$, the upper bound is again $\lceil \log_2 (3) \rceil + 1
= 3$ rounds, but the lower bound becomes $\lfloor \log_2 (5/3)
\rfloor + 2 = 2$ rounds. The exact number of rounds necessary to compute $C_5$
is thus open.
\end{example}

Finally, we show how to apply \autoref{lemma:lk:lower}
to show that transitive closure requires many rounds. In particular, we consider
the problem \textsc{Connected-Components}, for which, given an undirected graph 
$G = (V,E)$ with input a set of edges, the requirement is to label the nodes of 
each connected component with the same label, unique to that component. 

\begin{theorem}
\label{th:connected-comps}
Let $G$ be an input graph of size $M$. For any $\varepsilon < 1$, 
there is no algorithm in the tuple-based \mpc\ model that computes 
\textsc{Connected-Components} with $p$ processors and load
$L = O(M/p^{1-\varepsilon})$ in fewer than $o(\log p)$ rounds.
\end{theorem}

The idea of the proof is to construct input graphs
for \textsc{Connected-Components} whose components correspond to the output
tuples for $L_k$ for $k=p^\delta$ for some small constant $\delta$ depending
on $\varepsilon$ and use the round lower bound for solving $L_k$.
Notice that the size of the query $L_k$ is not fixed, but 
depends on the number of processors $p$.  


\begin{proof}
Since larger $\varepsilon$ implies a more powerful algorithm, we
assume without loss of generality that $\varepsilon=1-1/t$ for some
integer constant $t > 1$. Let $\delta=1/(2t (t+2))$.
The family of input graphs and the initial
distribution of the edges to servers will look like an input to
$L_k$, where $k =\lfloor p^\delta\rfloor$.
In particular, the vertices of the input graph 
$G$ will be partitioned into $k+1$ sets $P_1, \dots, P_{k+1}$, 
each partition containing $m/k$ vertices. The edges of $G$
will form permutations (matchings) between adjacent partitions, $P_i, P_{i+1}$,
for $i=1, \dots, k$. Thus, $G$ will contain exactly $k \cdot (m/k) = m$ edges. 
This  construction creates essentially $k$ binary relations, each with
$m/k$ tuples and size $M_k = (m/k) \log(m/k)$.
  
Since $k<p$, we can assume that the adversary initially places the edges
of the graph so that each server is given edges only from one relation. 
It is now easy to see that any tuple-based algorithm in \mpc\ that solves
\textsc{Connected-Components} for an arbitrary graph $G$ of the above family
in $r$ rounds with load $L$ implies 
an $(r+1)$-round tuple-based algorithm with the same load that solves 
$L_k$ when each relation has size $M$.
Indeed, the new algorithm runs the algorithm for connected 
components for the first $r$ rounds, and
then executes a join on the labels of each node. 
Since each tuple in $L_k$ corresponds
exactly to a connected component in $G$, the join will recover all
the tuples of $L_k$.
  
Since the query size is not independent of the number of servers $p$,
we have to carefully compute the constants for our lower bounds.
Consider an algorithm
for $L_k$ with load $L \leq c M/p^{1-\varepsilon}$, where $M = m \log(m)$.
Let $r=\lceil \log_{k_{\varepsilon}} k \rceil -2$.
Observe also that $k_\varepsilon=2t$ since $\varepsilon=1-1/t$.

We will use the $(\varepsilon,r)$-plan $\mathcal{M}$ for $L_k$
presented in the proof of \autoref{lem:er-plan:line}, apply
\autoref{th:strong-multiround},
and compute the factor $\beta(L_k,\mathcal{M})$.
First, notice that each query $L_k/\contracted_j$ for $j=0, \dots, r$ is
isomorphic to $L_{k/k_{\varepsilon}^j}$.
Then, the set $\mathcal{S}_{\varepsilon}(L_{k/k_{\varepsilon}^j})$
consists of at most $k/k_{\varepsilon}^j$ paths $q'$ of
length $k_{\varepsilon}+1$.
By the choice of $r$, $L_k/\contracted_r$ is isomorphic to $L_\ell$ where
$k_\varepsilon+1 \leq \ell  < k_\varepsilon^2$. Further, we have that 
$\tau^*(\mathcal{M})=\tau^*(L_{k_\varepsilon+1})=\lceil (k_\varepsilon+1)/2\rceil=t+1$
since $k_\varepsilon=2t$.

Thus, we have
\begin{align*}
\beta(L_k,\mathcal{M})  &= 
\left( \frac{1}{\tau^*(L_k/\contracted_r)} \right)^{\tau^*(\mathcal{M})} +
\sum_{j=1}^{r} \sum_{q' \in \mathcal{S}_\varepsilon(q/\contracted_{k-1})} 
\left( \frac{1}{\tau^*(q')} \right)^{\tau^*(\mathcal{M})} \\
& \leq (1-\varepsilon)^{\tau^*(\mathcal{M})} \left( 1 + \sum_{j=1}^{r} \frac{k}{k_{\varepsilon}^{j-1}} 
\right) \\
&   \leq (2k+1) (1- \varepsilon)^{\tau^*(\mathcal{M})}.
\end{align*}  
%
  %

Observe now that $M/M_k = 1/(1/k - \log(k)/(k \log(m))) \leq 2k$, assuming that $m \geq p^{2\delta}$.
Consequently, \autoref{th:strong-multiround} implies that any tuple-based 
\mpc\ algorithm using at most $\lceil \log_{k_{\varepsilon}} k \rceil-1$ rounds 
reports at most the following fraction of the
required output tuples for the $L_k$ query:
\begin{align*}
\beta(L_k, \mathcal{M}) \cdot p \left( \frac{(r+1) L}{M_k} \right)^{\tau^*(\mathcal{M})} 
& \leq (2k+1) (2ck(r+1)/t )^{\tau^*(\mathcal{M})} \cdot p^{1 - \tau^*(\mathcal{M})(1-\varepsilon)}  \\
& \leq  c' k^{t+2} (\log_2 k )^{c''} \cdot p^{1 - (1+t)(1-\varepsilon)} \\
& \leq c' (\delta \log_2 p)^{c''} \cdot p^{\delta (t+2) +1 -(1+t)(1-\varepsilon) } \\
& = c' (\delta \log_2 p)^{c''} \cdot p^{\delta (t+2) -1/t} \\
& = c' (\delta \log_2 p)^{c''} \cdot p^{-1/2t}
\end{align*}
where $c,c''$ are constants. Since $t> 1$, the fraction of the output tuples 
is $o(1)$ as a function of the number of processors $p$.
This implies that any algorithm that computes \textsc{Connected-Components} 
on $G$ requires at least  $\lceil \log_{k_\varepsilon} \lfloor p^\delta\rfloor \rceil -2 = \Omega(\log p)$ rounds.
\end{proof}

\section{Related Work}
\label{sec:related}

\paragraph{MapReduce-Related Models}

Several computation models have been proposed in order to understand
the power of MapReduce and related massively parallel programming
methods~\cite{DBLP:journals/talg/FeldmanMSSS10,DBLP:conf/soda/KarloffSV10,DBLP:conf/pods/KoutrisS11,DBLP:journals/corr/abs-1206-4377}.
These all identify the number of communication rounds as a
main complexity parameter, but differ in their treatment of the
communication.

The first of these models 
was the MUD (Massive, Unordered,
Distributed) model of Feldman et
al.~\cite{DBLP:journals/talg/FeldmanMSSS10}.  It takes as input a
sequence of elements and applies a binary merge operation repeatedly,
until obtaining a final result, similarly to a User Defined Aggregate
in database systems.  The paper compares MUD with streaming
algorithms: a streaming algorithm can trivially simulate MUD, and the
converse is also possible if the merge operators are computationally
powerful (beyond PTIME).

Karloff et al.~\cite{DBLP:conf/soda/KarloffSV10} define
$\mathcal{MRC}$, a class of multi-round algorithms based on
using the MapReduce primitive as the sole building block, and fixing
specific parameters for balanced processing.  The
number of processors $p$ is $\Theta(N^{1-\epsilon})$, and each can
exchange MapReduce outputs expressible in $\Theta(N^{1-\epsilon})$
bits per step, resulting in $\Theta(N^{2-2\epsilon})$ total storage
among the processors on a problem of size $N$.  Their focus was
algorithmic, showing simulations of other parallel models by $\mathcal{MRC}$,
as well as the power of two round algorithms for specific problems.

Lower bounds for the single round MapReduce model are first discussed
by Afrati et al.~\cite{DBLP:journals/corr/abs-1206-4377}, who derive
an interesting tradeoff between reducer size and replication rate.
This is nicely illustrated by Ullman's drug interaction
example~\cite{DBLP:journals/crossroads/Ullman12}.  There are $n$
($=6,500$) drugs, each consisting of about 1MB of data about patients
who took that drug, and one has to find all drug interactions, by
applying a user defined function (UDF) to all pairs of drugs.  To see
the tradeoffs, it helps to simplify the example, by assuming we are
given {\em two} sets, each of size $n$, and we have to apply a UDF to
every pair of items, one from each set, in effect computing their
Cartesian product.  There are two extreme ways to solve this. One can
use $n^2$ reducers, one for each pair of items; while each reducer has
size $2$, this approach is impractical because the entire data is
replicated $n$ times.  At the other extreme one can use a single
reducer that handles the entire data; the replication rate is 1, but
the size of the reducer is $2n$, which is also impractical.  As a
tradeoff, partition each set into $g$ groups of size $n/g$, and use
one reducer for each of the $g^2$ pairs of groups: the size of a
reducer is $2n/g$, while the replication rate is $g$.  Thus, there is
a tradeoff between the replication rate and the reducer size, which
was also shown to hold for several other classes of
problems~\cite{DBLP:journals/corr/abs-1206-4377}.

There are two significant limitations of this prior work: 
(1) As powerful and as convenient as the MapReduce framework is, the operations
it provides may not be able to take full advantage of the resource
constraints of modern systems.   The lower bounds say nothing about
alternative ways of structuring the computation that send and receive the same
amount data per step.  
(2) Even within the MapReduce framework, the only lower bounds apply to a 
single communication round, and say nothing about the limitations of multi-round
MapReduce algorithms.

While it is convenient that MapReduce hides the number of servers
from the programmer, when considering the most efficient way to use resources
to solve problems it is natural to expose information about those resources
to the programmer.
In this paper, we take the view that the number of servers $p$
should be an explicit parameter of the model, which allows us to focus
on the tradeoff between the amount of communication and the number of rounds.  
For example, going back to our Cartesian product problem, if the
number of servers $p$ is known, there is one optimal way to solve the
problem: partition each of the two sets into $g = \sqrt{p}$ groups,
and let each server handle one pair of groups.  

A model with $p$ as explicit parameter was proposed by Koutris and
Suciu~\cite{DBLP:conf/pods/KoutrisS11}, who showed both lower and
upper bounds for one round of communication.  In this model
only tuples are sent and they must be routed independent
of each other.  For example, \cite{DBLP:conf/pods/KoutrisS11} proves
that multi-joins on the same attribute can be computed in one round,
while multi-joins on different attributes, like $R(x),S(x,y),T(y)$
require strictly more than one round. The study was mostly focused
on understanding data skew, the model was limited, and the results do
not apply to more than one round.  

The \mpc\ model we introduce in this paper is much more general than the
above models, allowing arbitrary bits to represent communicated data, rather
than just tuples, and unbounded computing power of servers so the lower bounds
we show for it apply more broadly.  Moreover, we
establish lower bounds that hold even in the absence of skew.

\paragraph{Other Parallel Models}

The prior parallel model that is closest to the \mpc\ model is Valiant's 
Bulk Synchronous Parallel (BSP) model~\cite{DBLP:journals/cacm/Valiant90}.
The BSP model, operates in synchronous rounds of supersteps consisting of
possibly asynchronous steps.
In addition to the number of processors, there is a 
superstep size, $L$, there is the notion of an $h$-relation, a mapping in
which each processor sends and receives at most $h$
bits, as well as an architecture-dependent bandwidth parameter $g$ which says
that a superstep has to have at least $gh$ steps in order for the processors to deliver an $h$-relation during a superstep.

In the \mpc\ model, the notion of load $L$ largely parallels the notion
of $h$-relation (though we technically only need to bound the number of bits
each processor receives to obtain our lower bounds) but the other parameters
are irrelevant because we strengthen the model to allow unbounded local
computation and hence the only notion of time in the model is the number of 
synchronous rounds (supersteps).

The finer-grained LogP
model~\cite{DBLP:journals/cacm/CullerKPSSSSE96} does away with the
synchronization barriers and the notion of $h$-relations 
inherent in the BSP model and has a more continuous
notion of relaxed synchrony based on a latency bound and bound on processor
overhead for setting up communication, rather than based on 
supersteps.
While its finer grain computation and relaxed asyncrhony allowed tighter
modeling of a number of parallel architectures, it seems less well matched to
system architectures for MapReduce-style computations than either the BSP
or \mpc\ models.

\paragraph{Communication complexity}

The results we show belong to the study of communication complexity,
for which there is a very large body of existing research~\cite{kn97}.
Communication complexity considers the number of bits that need to be
communicated between cooperating agents in order to solve
computational problems when the agents have unlimited computational
power.  Our model is related to the so-called number-in-hand
multi-party communication complexity, in which there are multiple
agents and no shared information at the start of communication.  This
has already been shown to be important to understanding the processing
of massive data: Analysis of number-in-hand (NIH) communication
complexity has been the main method for obtaining lower bounds on the
space required for data stream algorithms~(e.g. \cite{ams:freq}).

However, there is something very different about the results that we prove here.
In almost all prior lower bounds, there is at least one agent that has
access to all communication between agents\footnote{\footnotesize Though private-messages
models have been defined before, we are aware of only two lines of work where
lower bounds make use of the fact that no single agent has access to all
communication:  
(1) Results of \cite{DBLP:conf/focs/GalG07,DBLP:conf/icalp/GuhaH09}
use the assumption that communication is both private and (multi-pass) one-way,
but unlike the bounds we prove here, their lower bounds 
are smaller than the total input size;
(2) Tiwari~\cite{tiw87} defined a distributed model of communication 
complexity in networks in which in input is given to two processors that
communicate privately using other helper processors.
However, this model is equivalent to ordinary public two-party
communication when the network allows direct private communication between any
two processors, as our model does.}.
(Typically, this is either
via a shared blackboard to which all agents have access or a referee who
receives all communication.)   
In this case, no problem on $N$ bits whose answer is
$M$ bits long can be shown to require more than $N+M$ bits of communication.

In our \mpc\ model, all communication between servers is {\em private} and
we restrict the communication per processor per step, rather than the total
communication.  
Indeed, the privacy of communication is essential to our lower bounds, since
we prove lower bounds that apply when the total communication is much larger
than $N+M$. (Our lower bounds for some problems apply when the total
communication is as large as $N^{1+\delta}$.)

\section{Conclusion}
\label{sec:conclusion}

In this paper, we introduce a simple but powerful model, the \mpc\ model, 
that allows us to analyze query processing in massively parallel systems. 
The \mpc\ model captures two important parameters: the number of communication
rounds, and the maximum load that a server receives during the computation.
We prove the first tight upper and lower bounds for the maximum load in the case of one
communication round and input data without skew. Then, we show how to handle skew
for several classes of queries. Finally, we analyze the precise tradeoff between
the number of rounds and maximum load for the case of multiple rounds.

Our work leaves open many interesting questions. The analysis for multiple rounds works
for a limited class of inputs, since we have to assume that all relations have
the same size. Further, our lower bounds are (almost) tight only for a specific class of queries
(tree-like queries), so it remains open how we can obtain lower bounds for any conjunctive
query, where relations have different size.

The effect of data skew in parallel computation is another exciting research direction. 
Although we have some understanding on how to handle skew in a single round,
it is an open how skew influences computation when we have multiple rounds, and
what are the tradeoffs we can obtain in such cases.

\bibliographystyle{plain}
\bibliography{bib}

\begin{thebibliography}{10}

\bibitem{DBLP:journals/corr/abs-1206-4377}
F.~N. Afrati, A.~D. Sarma, S.~Salihoglu, and J.~D. Ullman.
\newblock Upper and lower bounds on the cost of a map-reduce computation.
\newblock {\em CoRR}, abs/1206.4377, 2012.

\bibitem{DBLP:conf/edbt/AfratiU10}
F.~N. Afrati and J.~D. Ullman.
\newblock Optimizing joins in a map-reduce environment.
\newblock In {\em EDBT}, pages 99--110, 2010.

\bibitem{ams:freq}
N.~Alon, Y.~Matias, and M.~Szegedy.
\newblock The space complexity of approximating the frequency moments.
\newblock {\em JCSS}, 58(1):137--147, 1999.

\bibitem{DBLP:conf/focs/AtseriasGM08}
A.~Atserias, M.~Grohe, and D.~Marx.
\newblock Size bounds and query plans for relational joins.
\newblock In {\em FOCS}, pages 739--748, 2008.

\bibitem{BKS13}
Paul Beame, Paraschos Koutris, and Dan Suciu.
\newblock Communication steps for parallel query processing.
\newblock In Richard Hull and Wenfei Fan, editors, {\em Proceedings of the 32nd
  {ACM} {SIGMOD-SIGACT-SIGART} Symposium on Principles of Database Systems,
  {PODS} 2013, New York, NY, {USA} - June 22 - 27, 2013}, pages 273--284.
  {ACM}, 2013.

\bibitem{BeameKS14}
Paul Beame, Paraschos Koutris, and Dan Suciu.
\newblock Skew in parallel query processing.
\newblock In Richard Hull and Martin Grohe, editors, {\em Proceedings of the
  33rd {ACM} {SIGMOD-SIGACT-SIGART} Symposium on Principles of Database
  Systems, PODS'14, Snowbird, UT, USA, June 22-27, 2014}, pages 212--223.
  {ACM}, 2014.

\bibitem{bennett}
George Bennett.
\newblock Probability inequalities for the sum of independent random variables.
\newblock {\em Journal of the American Statistical Association},
  57(297):33--45, March 1962.

\bibitem{DBLP:conf/pods/Chaudhuri12}
S.~Chaudhuri.
\newblock What next?: a half-dozen data management research goals for big data
  and the cloud.
\newblock In {\em PODS}, pages 1--4, 2012.

\bibitem{DBLP:journals/cacm/CullerKPSSSSE96}
David~E. Culler, Richard~M. Karp, David~A. Patterson, Abhijit Sahay, Eunice~E.
  Santos, Klaus~E. Schauser, Ramesh Subramonian, and Thorsten von Eicken.
\newblock Logp: {A} practical model of parallel computation.
\newblock {\em Commun. {ACM}}, 39(11):78--85, 1996.

\bibitem{DBLP:conf/osdi/DeanG04}
J.~Dean and S.~Ghemawat.
\newblock Mapreduce: Simplified data processing on large clusters.
\newblock In {\em OSDI}, pages 137--150, 2004.

\bibitem{datasciencesurvey}
{EMC Corporation}.
\newblock Data science revealed: A data-driven glimpse into the burgeoning new
  field.
\newblock
  \url{http://www.emc.com/collateral/about/news/emc-data-science-study-wp.pdf}.

\bibitem{DBLP:journals/talg/FeldmanMSSS10}
J.~Feldman, S.~Muthukrishnan, A.~Sidiropoulos, C.~Stein, and Z.~Svitkina.
\newblock On distributing symmetric streaming computations.
\newblock {\em ACM Transactions on Algorithms}, 6(4), 2010.

\bibitem{friedgut2004hypergraphs}
E.~Friedgut.
\newblock Hypergraphs, entropy, and inequalities.
\newblock {\em American Mathematical Monthly}, pages 749--760, 2004.

\bibitem{DBLP:conf/focs/GalG07}
A.~G{\'a}l and P.~Gopalan.
\newblock Lower bounds on streaming algorithms for approximating the length of
  the longest increasing subsequence.
\newblock In {\em FOCS}, pages 294--304, 2007.

\bibitem{DBLP:journals/jlp/GangulyST92}
S.~Ganguly, A.~Silberschatz, and S.~Tsur.
\newblock Parallel bottom-up processing of datalog queries.
\newblock {\em J. Log. Program.}, 14(1{\&}2):101--126, 1992.

\bibitem{DBLP:conf/soda/GroheM06}
M.~Grohe and D.~Marx.
\newblock Constraint solving via fractional edge covers.
\newblock In {\em SODA}, pages 289--298, 2006.

\bibitem{DBLP:conf/icalp/GuhaH09}
Sudipto Guha and Zhiyi Huang.
\newblock Revisiting the direct sum theorem and space lower bounds in random
  order streams.
\newblock In {\em ICALP}, volume 5555 of {\em LNCS}, pages 513--524. Springer,
  2009.

\bibitem{DBLP:conf/approx/ImpagliazzoK10}
Russell Impagliazzo and Valentine Kabanets.
\newblock Constructive proofs of concentration bounds.
\newblock In {\em Proceedings, Approximation, Randomization, and Combinatorial
  Optimization. Algorithms and Techniques, 14th International Workshop,
  {RANDOM}}, volume 6302 of {\em Lecture Notes in Computer Science}, pages
  617--631, Barcelona, Spain, 2010. Springer.

\bibitem{DBLP:conf/soda/KarloffSV10}
H.~J. Karloff, S.~Suri, and S.~Vassilvitskii.
\newblock A model of computation for mapreduce.
\newblock In {\em SODA}, pages 938--948, 2010.

\bibitem{DBLP:conf/pods/KoutrisS11}
P.~Koutris and D.~Suciu.
\newblock Parallel evaluation of conjunctive queries.
\newblock In {\em PODS}, pages 223--234, 2011.

\bibitem{kn97}
E.~Kushilevitz and N.~Nisan.
\newblock {\em Communication Complexity}.
\newblock Cambridge University Press, Cambridge, England ; New York, 1997.

\bibitem{DBLP:conf/sigmod/KwonBHR12}
YongChul Kwon, Magdalena Balazinska, Bill Howe, and Jerome~A. Rolia.
\newblock Skewtune: mitigating skew in mapreduce applications.
\newblock In {\em Proceedings of the {ACM} {SIGMOD} International Conference on
  Management of Data, {SIGMOD} 2012, Scottsdale, AZ, USA, May 20-24, 2012},
  pages 25--36, 2012.

\bibitem{DBLP:journals/pvldb/MelnikGLRSTV10}
S.~Melnik, A.~Gubarev, J.~J. Long, G.~Romer, S.~Shivakumar, M.~Tolton, and
  T.~Vassilakis.
\newblock Dremel: Interactive analysis of web-scale datasets.
\newblock {\em PVLDB}, 3(1):330--339, 2010.

\bibitem{DBLP:conf/pods/NgoPRR12}
H.~Q. Ngo, E.~Porat, C.~R{\'e}, and A.~Rudra.
\newblock Worst-case optimal join algorithms: [extended abstract].
\newblock In {\em PODS}, pages 37--48, 2012.

\bibitem{DBLP:conf/sigmod/OlstonRSKT08}
C.~Olston, B.~Reed, U.~Srivastava, R.~Kumar, and A.~Tomkins.
\newblock Pig latin: a not-so-foreign language for data processing.
\newblock In {\em SIGMOD Conference}, pages 1099--1110, 2008.

\bibitem{DBLP:conf/www/SuriV11}
Siddharth Suri and Sergei Vassilvitskii.
\newblock Counting triangles and the curse of the last reducer.
\newblock In {\em WWW}, pages 607--614, 2011.

\bibitem{TSJSCALWM09}
A.~Thusoo, J.~S. Sarma, N.~Jain, Z.~Shao, P.~Chakka, S.~Anthony, H.~Liu,
  P.~Wyckoff, and R.~Murthy.
\newblock Hive - a warehousing solution over a map-reduce framework.
\newblock {\em PVLDB}, 2(2):1626--1629, 2009.

\bibitem{tiw87}
P.~Tiwari.
\newblock Lower bounds on communication complexity in distributed computer
  networks.
\newblock {\em JACM}, 34(4):921--938, October 1987.

\bibitem{DBLP:journals/crossroads/Ullman12}
J.~D. Ullman.
\newblock Designing good mapreduce algorithms.
\newblock {\em ACM Crossroads}, 19(1):30--34, 2012.

\bibitem{DBLP:journals/cacm/Valiant90}
Leslie~G. Valiant.
\newblock A bridging model for parallel computation.
\newblock {\em Commun. {ACM}}, 33(8):103--111, 1990.

\bibitem{Shark}
Reynold~S. Xin, Josh Rosen, Matei Zaharia, Michael~J. Franklin, Scott Shenker,
  and Ion Stoica.
\newblock Shark: {SQL} and rich analytics at scale.
\newblock In Kenneth~A. Ross, Divesh Srivastava, and Dimitris Papadias,
  editors, {\em Proceedings of the {ACM} {SIGMOD} International Conference on
  Management of Data, {SIGMOD} 2013, New York, NY, USA, June 22-27, 2013},
  pages 13--24. {ACM}, 2013.

\bibitem{yao83}
A.~C. Yao.
\newblock Lower bounds by probabilistic arguments.
\newblock In {\em FOCS}, pages 420--428, Tucson, AZ, 1983.

\end{thebibliography}

\newpage
\appendix

\section{Hashing}
\label{sec:hashing}

In this section, we present a detailed analysis of the behavior of the HyperCube
algorithm for input distributions with various guarantees. 
Throughout this section, we assume that a hash function is chosen 
randomly from a {\em strongly universal family} of hash functions.
Recall that a strongly universal set of hash function is a set $\mathcal{H}$ of functions with 
range $[K]$ such that, for any $n \geq 1$, any distinct values $a_1, \dots ,a_n$ and any 
bins $B_1, \dots,B_n \in [K]$, we have that 
$\P(h(a_1) = B_1\wedge \dots \wedge h(a_n) = B_n) = 1/K^n$, 
where the probability is over the random choices of $h \in \mathcal{H}$.

\subsection{Basic Partition}

We start by examining the following scenario. Suppose that we have a set of weighted
balls which we hash-partition into $K$ bins; what is the maximum load among all the bins?
Assuming that the sum of the weights is $m$, it is easy to see that the expected load for
each bin is $m/K$. 
However, this does not tell us anything about the maximum load. In particular,
in the case where we have one ball of weight $m$, the maximum load will always be $m$,
which is far from the expected load.

In order to obtain meaningful bounds on the distribution of the maximum load, we thus have
to put a restriction on the maximum weight of a ball. The following theorem provides such
a tail bound on the probability distribution.

\begin{theorem}[Weighted Balls in Bins]
\label{th:balls:in:bins} 
Let $S$ be a set where each element $i$ has weight $w_i$ and
$\sum_{i\in S} w_i\le m$.
Let $K>0$ be an integer.
Suppose that for some $\beta>0$, $\max_{i \in S} \{w_i\} \leq \beta m/K$.
We hash-partition $S$ into $K$ bins.
Then for any $\delta>0$ 
\begin{align}
\P(\text{maximum weight of any bin}\ge (1+\delta) m/K) \leq K\cdot e^{-h(\delta)/\beta}\label{eq:strong}
\end{align}
where $h(x) = (1+x)\ln(1+x) - x$. 
\end{theorem}

A stronger version of \autoref{th:balls:in:bins} is also true, where we
replace $h(\delta)$ by $K \cdot D(\frac{1+\delta}{K}||\frac{1}{K})$ where
$D(q'||q)=q'\ln (\frac{q'}{q}) + (1-q') \ln(\frac{1-q'}{1-q})$
is the relative entropy (also known as the KL-divergence) of
Bernoulli indicator variables with probabilities $q'$ and $q$. 
This strengthening\footnote{Note that $K \cdot D(\frac{1+\delta}{K}||\frac{1}{K})
=(1+\delta)\ln(1+\delta)+(K-1-\delta))\ln(1-\frac{\delta}{K-1})\ge (1+\delta)\ln(1+\delta)-\delta$.} is immediate from 
the following theorem with $t=1+\delta$ and $m_1=m$, which implies the bound
for a single bin, together with a union bound over all $K$ choices of bins.
The statement above follows from a weaker form of
Theorem~\ref{th:vector:in:bins} that
can also be derived using using Bennett's
inequality~\cite{bennett}.

\begin{theorem} \label{th:vector:in:bins} 
Let $K\ge 2$.
Let $w\in \mathbb{R}^n$ satisfy $w\ge 0$, 
$||w||_1\le m_1$, and $||w||_\infty \le m_\infty = \beta m_1/K$.
Let $(Y_i)_{i\in [n]}$ be a vector of
i.i.d.~random indicator variables with $\P(Y_i=1)=1/K$ and $\P(Y_i=0)=1-1/K$.
Then $$\P(\sum_{i\in [n]} w_i Y_i > tm_1/K)\le e^{-K\cdot D(t/K||1/K)/\beta}.$$
\end{theorem}

\begin{proof}
The proof follows along similar lines to constructive proofs of Chernoff
bounds in~\cite{DBLP:conf/approx/ImpagliazzoK10}:
Choose a random
$S \subseteq [n]$ by including each element $i$ independently with
probability $q_i$:
\begin{align*}
  \P(i \in S) = q_i = & 1- \left(1-q\right)^{\frac{w_i}{\beta m_1/K}}
\end{align*}
Let $\mathcal{E}$ denote the event that
  $\sum_{i \in [n]} w_i Y_i \geq t m_1/K$.
Then
\begin{align*}
  \E\left[\bigwedge_{i \in S} Y_i = 1\right ] \geq \E\left[\bigwedge_{i \in S} Y_i = 1 \mid \mathcal{E}\right]\cdot \P(\mathcal{E}).
\end{align*}
We bound both expectations.  First we see that
\begin{align}
\begin{split}
  \E\left[\bigwedge_{i \in S} Y_i = 1\right] =
  & \sum_{S\subseteq [n]} (1/K^{|S|}) \prod_{i\in S} q_i \prod_{i\notin S} (1-q_i)\\
  =& \sum_{S\subseteq [n]} \prod_{i\in S} (q_i/K) \prod_{i\notin S} (1-q_i)\\
  = & \prod_{i \in [n]} \left(q_i/K + (1-q_i)\right)
  =  \prod_{i \in [n]} \left(1/K + (1-1/K)(1-q_i)\right) \\
  = & \prod_{i \in [n]} \left(1/K + (1-1/K)(1-q)^{\frac{w_i}{\beta m_1/K}}\right) \\
\leq&  \prod_{i \in [K/\beta]}\left(1/K +  (1-1/K)(1-q)\right)\\
 = & \left(1/K +  (1-1/K)(1-q)\right)^{K/\beta} \\
 =  &\left(1 - q(1-1/K)\right)^{K/\beta}. 
\end{split}\label{eq:e1}
\end{align}
The inequality
follows from the fact the function $f(w) = 1/K +
(1-1/K)(1-q)^{Kw/(\beta m_1)}$ is log-convex\footnote{Write $g(w) =
  \log f(w)$.  Then $g(w) = \log(a + b e^{-cw})$ for some $a,b,c>0$;
  hence $g'(w) = -bce^{-cw}/(a + b e^{-cw}) = -bc/(ae^{cw} + b)$ is
  increasing and so $g$ is convex.} and therefore $\prod_i f(w_i)$ is
maximized on a vertex of the polytope given by $0 \leq w_i \leq \beta m_1/K$
and $\sum_i w_i \leq m_1$.

Second,
for any outcome of $Y_1, \ldots, Y_n$ that satisfies $\mathcal{E}$, the
probability that $S$ misses all indices $i$ such that $Y_i=0$ is
\begin{align*}
\prod_{i: Y_i=0} (1-q)^{\frac{w_i}{\beta m_1/K}}=(1-q)^{\frac{\sum_{i:Y_i=0} w_i}{\beta m_1/K}} \geq (1-q)^{\frac{m_1-tm_1/K}{\beta m_1/K}}
\end{align*}
since $\mathcal{E}$ implies that
$\sum_{i:Y_i=0} w_i \le m_1-\sum_{i:Y_i=1} w_i\le m_1-tm_1/K$.
Hence
\begin{align}
  \E\left[\bigwedge_{i \in S} Y_i = 1 \mid \mathcal{E}\right ] \geq  (1-q)^{(K - t)/\beta}. \label{eq:e2}
\end{align}
Combining Eq(\ref{eq:e1}) and (\ref{eq:e2}), we obtain:
\begin{align*}
  \P(\mathcal{E}) \leq \left(\frac{1-q(1-1/K)}{(1-q)^{1-t/K}}\right)^{K/\beta}
\end{align*}
As noted in~\cite{DBLP:conf/approx/ImpagliazzoK10}, by looking at its first
derivative
one can show that the function 
$f_{\delta,\gamma}(q)=\frac{1-q(1-\delta)}{(1-q)^{1-\gamma}}$ takes its
minimum at $q = q^*=\frac{\gamma-\delta}{\gamma(1-\delta)}$ where it has value
$e^{-D(\gamma||\delta)}$.
Plugging in $\delta=1/K$ and $\gamma=t/K$, we obtain:
\begin{align*}
  \P(\mathcal{E}) \leq e^{-K \cdot D(t/K || 1/K)/\beta}
\end{align*}
as required.
\end{proof}

We also will find the following extension of~\autoref{th:vector:in:bins}
to be useful.

\begin{theorem}
\label{th:multivector:in:bins}
Let $K\ge 2$.
Let $(w^{(j)})_j$ be a sequence of vectors in $\mathbb{R}^n$ satisfying
$w^{(j)}\ge 0$, $||w^{(j)}||_1\le m_1$, and
$||w^{(j)}||_\infty \le m_\infty = \beta m_1/K$.
Suppose further that $||\sum_j w^{(j)}||_1\le k m_1$.
Let $(Y_i)_{i\in [n]}$ be a vector of
i.i.d.~random indicator variables with $\P(Y_i=1)=1/K$ and $\P(Y_i=0)=1-1/K$.
Then 
$$\P(\exists j\ \sum_{i\in [n]} w^{(j)}_i Y_i > (1+\delta)\frac{m_1}{K})\le 2k\cdot e^{-h(\delta)/\beta}$$
where $h(x) = (1+x)\ln(1+x) - x$. 
\end{theorem}

The proof of this theorem follows easily from the following lemma.

\begin{lemma}
\label{lem:bin-packing}
Let $(w^{(j)})_j$ be a sequence of vectors in $\mathbb{R}^n$ satisfying
$w^{(j)}\ge 0$, $||w^{(j)}||_1\le m_1$, and
$||w^{(j)}||_\infty \le m_\infty$.
Suppose further that $||\sum_j w^{(j)}||_1\le k m_1$.
Then there is a sequence of at most $2k$ vectors 
$u^{(1)},\ldots, u^{(2k)} \in \mathbb{R}^n$ such that 
each $u^{(\ell)}\ge 0$, $||u^{(\ell)}||_1\le m_1$, and
$||u^{(\ell)}||_\infty\le m_\infty$, and
for every $j$, there is some $\ell\in [2k]$ such that
$w^{(j)}\le u^{(\ell)}$, where the inequality holds only if it holds for every
coordinate.
\end{lemma}

\begin{proof}
The construction goes via the first-fit decreasing algorithm for bin-packing.
Sort the vectors $w^{(j)}$ in decreasing order of $||w^{(j)}||_1$.
Then greedily group them in bins of capacity $m_1$.  That is, 
we begin with $w^{(1)}$ and continue to add vectors until we find the largest
$j_1$ such that $\sum_{j=1}^{j_1} ||w^{(j)}||_1 \le m_1$.
Define $u^{(1)}$ by $u^{(1)}_i=\max_{1\le j\le j_1} w^{(j)}_i$ for each $i\in [n]$.
Now $||u^{(1)}||\le m_1$ and
$||u^{(1)}||_\infty \le \max_j ||w^{(j)}||_\infty\le m_\infty$.
Moreover, for each $j\in [1,j_1]$, $w^{(j)}\le u^{(1)}$ by definition.
Then repeat beginning with $w^{(j_1+1)}$ until the largest $j_2$ 
such that $\sum_{j=j_1+1}^{j_2} ||w^{(j)}||_1 \le m_1$, and define
$u^{(2)}$ by $u^{(2)}_i=\max_{j_1+1\le j\le j_2} w^{(j)}_i$ for each $i\in [n]$ as
before, and so on.
Since the contribution of each subsequent $||w^{(j)}||_1$ is at most that
of its predecessor, if it cannot be included in a bin, then that bin is more
than half full so we have $||u^{(\ell)}||_1> m_1/2$ for all $\ell$.  Since 
$\sum_{\ell} ||u^{(\ell)}||_1 \le \sum_j ||w^{(j)}||_1\le k m_1$, there must
be at most $2k$ such $u^{(\ell)}$.
\end{proof}

\begin{proof}[of \autoref{th:multivector:in:bins}]
We apply \autoref{lem:bin-packing} to the vectors $w^{(j)}$ to construct
$u^{(1)},\ldots u^{(2k)}$.
We then apply a union bound to the application of
\autoref{th:vector:in:bins} to each of $u^{(\ell)}$.
The total probability that there exists some $\ell\in [2k]$ such
that $\sum_{i\in [n]} u^{(\ell)}_i Y_i> (1+\delta)\beta m_1/K$ is at most
$2k\cdot e^{-h(\delta)/\beta}$.
Now for each $j$, there is some $\ell$ such that $w^{(j)}\le u^{(\ell)}$ and
hence $\sum_{i\in [n]} w^{(j)}_i Y_i \le \sum_{i\in [n]} u^{(\ell)}_i Y_i$.
Therefore if there exists a $j$ such that
$\sum_{i\in [n]} w^{(j)}_i Y_i > (1+\delta)m_1/K$
then there exists an $\ell$ such that
$\sum_{i\in [n]} u^{(\ell)}_i Y_i>(1+\delta)m_1/K$.
\end{proof}

\subsection{HyperCube Partition}

Before we analyze the load of the HC algorithm, we present some
useful notation. Even though the analysis in the main paper assumes
that relations are sets, here we will give a more general analysis
for {\em bags}. 

Let a $U$-tuple $J$ be a function $J : U \rightarrow [n]^{|U|}$, where $[n]$ is
the domain and $U \subseteq [r]$ a set of attributes.  If $J$ is a
$V$-tuple and $U \subseteq V$ then $\pi_U(J)$ is the projection of $J$
on $U$.  Let $S$ be a bag of $[r]$-tuples.  Define:
\begin{align*}
  m(S) = & |S| & \mbox{the size of the bag $S$, counting duplicates} \\
  \proj_U(S) = &  \setof{\pi_U(J)}{J \in S} & \mbox{duplicates are kept, thus $|\proj_U(S)|=|S|$} \\
  \sigma_{J}(S) = & \setof{K \in S}{\pi_U(K) = J} & \mbox{bag of tuples that contain $J$} \\
  d_{J}(S) = & |\sigma_{J}(S)|  & \mbox{the degree of the tuple $J$} \\
\end{align*}

Given shares $p_1, \dots, p_r$, such that $\prod_u p_u = p$,
 let $p_U = \prod_{u \in U} p_u$ for any attribute set $U$.  Let $h_1, \ldots, h_r$ be
independently chosen hash functions, with ranges $[p_1],\ldots,[p_r]$,
respectively.  
The {\em hypercube hash-partition of $S$} sends each element
$(i_1, \ldots, i_r)$ to the bin $(h_1(i_1), \ldots, h_r(i_r))$. 

\subsubsection{HyperCube Partition without Promise}

We prove the following:

\begin{theorem} \label{th:balls:in:bins:skewed} 
Let $S$ be a bag of tuples of $[n]^r$ such that each tuple in $S$ 
occurs at most $\beta^r m/p$ times, for some constant $\beta > 0$. 
Then for any $\delta > 0$:
\begin{align*}
 P\left(\text{maximum size any bin} > (1+\delta)\frac{m(S)}{\min_u p_u}\right) 
      \leq r \cdot p \cdot e^{- h(\delta)/\beta}
\end{align*}
where the bin refers to the HyperCube partition of $S$ using shares
$p_1, \ldots, p_r$.
\end{theorem}

Notice that there is no promise on how large the degrees can be. 
The only promise is on the number of repetitions in the bag $S$, which is 
automatically satisfied when $S$ is a set, since it is at most one.

\begin{proof}
We prove the theorem by induction on $r$.  If $r=1$ then it follows
immediately from~\autoref{th:balls:in:bins} by letting the weight of 
a ball $i$ be the number of elements in $S$ containing it. 
Assume now that $r > 1$.
We partition the domain $[n]$ into two sets:
\begin{align*}
    D_{small} = \setof{i}{d_{r\mapsto i}(S) \leq  \beta m/p_r}\mbox{ and
 }   D_{large} = \setof{i}{d_{r\mapsto i}(S) >  \beta m/p_r}
\end{align*}
Here $r \mapsto i$ denotes the tuple $(i)$; in other words
$\sigma_{r \mapsto i}(S)$ returns the tuples in $S$ whose last
($r$-th) attribute has value $i$.  
We then partition the bag $S$ into two sets $S_{small}, S_{large}$, where $S_{small}$
consists of tuples $t$ where $\pi_r(t) \in D_{small}$, and $S_{large}$ consists
of those where $\pi_r(t)\in D_{large}$.
The intuition is that we can apply \autoref{th:balls:in:bins} directly to show
that $S_{small}$ is distributed well by the hash function $h_r$. On the other hand, there
cannot be many $i\in D_{large}$, in particular $|D_{large}|\le p_r/\beta$,
and hence the projection of any tuple in $S_{large}$ onto $[r-1]$ has at most
$D_{large}$ extensions in $S_{large}$. Thus, we can obtain a good inductive distribution of 
$S_{large}$ onto $[r-1]$ using $h_1,\ldots,h_{r-1}$.

Formally, for $U\subset [r]$ and $T\subseteq S$, let
$M_U(T)$ denote the maximum number of tuples of $T$ that have any particular
fixed value under $h_U=\times_{j\in U}h_j$. With this notation, $M_{[r]}(S)$ denotes the
the maximum number of tuples from $S$ in any bin.
Hence, our goal is to show that $P(M_{[r]}(S)>(1+\delta)m(S)/\min_{u\in [r]} p_u)< r\cdot p \cdot e^{- h(\delta)/\beta}$.
Now, by \autoref{th:balls:in:bins},
$$P(M_{\set{r}}(S_{small})>(1+\delta)m(S_{small})/p_r)\le p\cdot e^{-h(\delta)/\beta}$$
and consequently
$$P(M_{[r]}(S_{small})>(1+\delta)m(S_{small})/\min_{u\in [r]} p_u)\le p\cdot e^{-h(\delta)/\beta}.$$
Let $S' = \proj_{[r-1]}(S_{large})$. Since projections keep duplicates, we have
$m(S')=m(S_{large})$ and $M_{[r-1]}(S')=M_{[r-1]}(S_{large})$.
By the assumption in the theorem statement, each tuple in $S$, and hence in
$S_{large}$, occurs at most $\beta^r m/p$ times. Then, since
$|D_{large}|\le p_r/\beta$, each tuple in $S'$ occurs at most
$\beta^{r-1} m/p')$ times where $p'=\prod_{u\in [r-1]} p_u$.
Therefore we can apply the inductive hypothesis to $S'$ to yield
$$P(M_{[r-1]}(S')>(1+\delta)m(S')/\min_{u\in [r-1]} p_u)\le (r-1)\cdot p\cdot e^{-h(\delta)/\beta}$$
and hence
$$P(M_{[r]}(S_{large})>(1+\delta)m(S_{large})/\min_{u\in [r]} p_u)\le (r-1)\cdot p\cdot e^{-h(\delta)/\beta}.$$
Since $m(S)=m(S_{small})+m(S_{large})$ and $M_{[r]}(S)=M_{[r]}(S_{small})+M_{[r]}(S_{large})$,
$$P(M_{[r]}(S)>(1+\delta)m(S)/\min_{u\in [r]} p_u))\le p\cdot e^{-h(\delta)/\beta}+(r-1)\cdot e^{-h(\delta)/\beta}= r\cdot p\cdot e^{-h(\delta)/\beta}$$
as required.
\end{proof}

\subsubsection{HyperCube Partition with Promise}

The following theorem extends ~\autoref{th:balls:in:bins:skewed} to the case
when we have a promise on the degrees in the bag (or set) $S$.  

\begin{theorem} \label{th:balls:in:hypercube} Let $S$ be a bag of tuples
of $[n]^r$, and suppose that for every $U$-tuple $J$ we have
 $d_{J}(S) \leq  \frac{\beta^{|U|}\cdot m}{p_U}$ where $\beta > 0$.
Consider a hypercube hash-partition of $S$ into
$p$ bins.  Then, for any $\delta \geq 0$:
\begin{align*}
  \P\left(\text{maximum size of any bin} > (1+\delta)^r \frac{m(S)}{p}\right) \leq
&
  f(p,r,\beta) \cdot e^{ -h(\delta)/\beta}
\end{align*}
where the bin refers to the HyperCube partition of $S$ using shares
$p_1, \ldots, p_r$
and 
\begin{align}
f(p,r,\beta)=2p\ \sum_{j=1}^{r} \prod_{u\in [j-1]} (1/\beta+1/p_u)
\le 2p\ \frac{(1/\beta+\varepsilon)^r-1}{1/\beta+\varepsilon-1},  \label{eq:bound}
\end{align}
where $\varepsilon=1/\min_{u\in [r-1]} p_u$.
\end{theorem}

We will think of $r$ as a constant, $p_u$ as being 
relatively large, and $\beta$ as $\log^{-O(1)}p$.

\begin{proof}
We prove the theorem by induction on $r$. The base case $r=1$ follows
immediately from~\autoref{th:balls:in:bins} since an empty product evaluates
to 1 and hence $f(p,1,\beta)=2p$.

Suppose that $r > 1$.  
There is one bin for each $r$-tuple in $[p_1]\times\cdots\times [p_r]$.
We analyze cases based on the value $b\in [p_r]$. Define
\begin{align*}
S^{r\rightarrow b}(h_r) = & \bigcup_{i \in [n]: h_r(i) = b} \sigma_{r \mapsto i}(S)  &\mbox{ and }&
S'(b,h_r) = \proj_{[r-1]}(S^{r\rightarrow b})
\end{align*}
Here $r \mapsto i$ denotes the tuple $(i)$.  
$S^{r\rightarrow b}(h_r)$ is a
random variable depending on the choice of the hash function $h_r$ 
that represents the bag of tuples sent to bins whose first projection is $b$.  
$S'(b,h_r)$ is essentially the same bag where we drop the last coordinate, which,
strictly speaking, we need to do to apply induction. 
Then $m(S'(b,h_r))=m(S^{r\rightarrow b}(h_r))$.

Since the promise with $U=\set{r}$ implies that 
$d_r(S)\le \beta\cdot m/p_r$, by \autoref{th:balls:in:bins},
$$\P(m(S'(b,h_r)>(1+|delta)m(S)/p_r\le e^{-h(\delta)/\beta}.$$
We handle the bins corresponding to each value of $b$ separately via induction.  However, in order to do this we need to argue
that the recursive version of the promise on coordinates holds for every
$U\subseteq [r-1]$ with $S'(b,h_r)$ and $m'=(1+\delta)m(S)/p_r$ instead of $S$ and $m$.
More precisely, we need to argue that, {\em with high probability, for every  
$U \subseteq [r-1]$ and every $U$-tuple $J$}, 
\begin{align}
d_J(S'(b,h_r))\le \frac{\beta^{|U|}\cdot m'}{p_U}=(1+\delta)\frac{\beta^{|U|}\cdot m}{p_U p_r} \label{eq:U}
\end{align}

Fix such a subset $U\subseteq [r-1]$.  
The case for $U=\emptyset$ is precisely the bound for the size $m(S'(b,h_r))$
of $S'(b,h_r)$.
Since the promise of the theorem statement with $U=\set{r}$ implies that
$d_{\{r\}}(S)\le \beta m/p_r$, by~\autoref{th:balls:in:bins} we have
that $P(m(S'(b))> m' ) \le e^{-h(\delta)/\beta}$.

Assume next that $U \ne \emptyset$.
Observe that $d_J(S'(b,h_r))$ is precisely the number of tuples of $S$ consistent with 
$(J,i)$ such that $h_r(i)=b$.
Using~\autoref{th:multivector:in:bins}, we upper bound the probability
that there is some $U$-tuple $J$ such that~\eqref{eq:U} fails.
Let $k(U)=p_U/\beta^{|U|}$.
For each fixed $(J,i)$, the promise for coordinates $U\cup\set{r}$ implies
that there are at most
$\frac{\beta^{|U|+1}\cdot m}{p_U p_r}=\frac{\beta m}{p_r k(U)}$ tuples in $S$
consistent with $(J,i)$.
Further, the promise for coordinates $U$ implies that there are at most
$\frac{\beta^{|U|} m}{p_U}=\frac{m}{k(U)}$ tuples in $S$ consistent with $J$.
For each such $J$ define vector $w^{(J)}$ by letting $w^{(J)}_i$ be the number
of tuples consistent with $(J,i)$.
Thus $||w^{(J)}||_\infty\le \frac{\beta m}{p_r k(U)}$ for all $J$ and
$||w^{(J)}||_1\le \frac{m}{k(U)}$ for all $J$.
Finally note that since there are $m=m(S)$ tuples in $S$,
$\sum_{J} ||w^{(J)}||_1 \le m$.
We therefore we can apply~\autoref{th:multivector:in:bins} with $k=k(U)$,
$m_1=m/k(U)$ and $m_\infty=\beta m_1/p_r$ to say that the probability that
there is some $U$-tuple $J$ such that 
$d_J(S'(b))>(1+\delta) m_1/p_1=(1+\delta)m/(p_r k(U))$ is at most
$2k(U)\cdot e^{-h(\delta)/\beta}$.

For a fixed $b$, we now use a union bound over the possible sets
$U\subseteq [r-1]$ to obtain a
total probability that~\eqref{eq:U} fails  for some set $U$ and some $U$-tuple
$J$ of at most 
\begin{align*}
2\ \sum_{U \subseteq [r-1]}
  \beta^{-|U|} p_U \cdot e^{-h(\delta)/\beta}&= 2\ \prod_{u\in [r-1]} (1 +
  p_u/\beta)\cdot e^{-h(\delta)/\beta}\\
  &= 2 (p / p_r) \prod_{u\in [r-1]} (1/\beta + 1/p_u)\cdot e^{-h(\delta)/\beta}
\end{align*}
If $m(S'(b))\le m'$ and~\eqref{eq:U} holds for all $U\subseteq[r-1]$ and
$U$-tuples $J$, then
we apply the induction hypothesis~\eqref{eq:bound} to derive that the
probability that some bin that has $b$ in its last coordinate has more than
$(1+\delta)^{r-1} m'/(p/p_r)=(1+\delta)^r m/p$ tuples is at most
$f(p/p_r,r-1,\beta)\cdot e^{-h(\delta)/\beta}$.

Since there are $p_r$ choices for $b$, we obtain a total failure probability
at most 
$f(p,r,\beta) \cdot e^{-h(\delta)/\beta}$ where
\begin{align*}
f(p,r,\beta) & \le p_r \left( 2p_{[r-1]}\prod_{u\in [r-1]} (1/\beta+1/p_u)
+ f(p/p_r,r-1,\beta) \right) \\
& = 2p\ \prod_{u\in [r-1]} (1/\beta+1/p_u) + p_r\ f(p/p_r,r-1,\beta)\\
& = 2p\ \prod_{u\in [r-1]} (1/\beta+1/p_u) + p_r (2p/p_r)\ \sum_{j=1}^{r-1}\prod_{u\in [j-1]}(1/\beta+1/p_u)\\
& = 2p\ \sum_{j=1}^{r}\prod_{u\in [j-1]}(1/\beta+1/p_u)\\
& = f(p,r,\beta)
\end{align*}
 The final bound uses geometric series sum upper bound.
\end{proof}

\section{Probability Bounds}
\label{sec:prob-bounds}

In this section, we show how to obtain lower bounds on the probability of failure using
bounds on the expected output.
We start by proving a lemma regarding the distribution of the query output for random
matching databases.

\begin{lemma} \label{lem:size:bound}
Let $I$ be a random matching database for a connected conjunctive query $q$, and let
$\mu = \E[|q(I)|]$. Then, for any $\alpha \in [0,1)$ we have:
$$ P(|q(I)| > \alpha \mu) \geq (1-\alpha)^2 \frac{\mu}{\mu + 1} $$
\end{lemma}

\begin{proof}
To prove the bound, we will use the Paley-Zygmund inequality for the random variable $|q(I)|$:
$$ P(|q(I)| > \alpha \mu ) \geq (1-\alpha)^2 \frac{\mu^2}{\E[|q(I)|^2]} $$

To bound $\E[|q(I)|^2] $, we construct a query $q'$ that consists of $q$ plus a copy of $q
$ with new variables. For example, if $q = R(x,y), S(y,z)$, 
we define $q' = R(x,y), S(y,z), R(x',y'), S(y',z')$. We now have:
\begin{align*}
\E[|q(I)|^2] = \E[|q'(I)|] & = 
\sum_{\ba, \ba' \in [n]^k} \prod_{j=1}^{\ell} P(\ba_j \in S_j  \wedge \ba'_j \in S_j) \\
& =  \sum_{\ba \neq \ba' \in [n]^k} \prod_{j=1}^{\ell} P(\ba_j \in S_j  \wedge \ba'_j \in S_j)
+  \sum_{\ba \in [n]^k} \prod_{j=1}^{\ell} P(\ba_j \in S_j)  \\
& = \sum_{\ba \neq \ba' \in [n]^k} \prod_{j=1}^{\ell} P(\ba_j \in S_j) P(\ba'_j \in S_j \mid \ba_j \in S_j) + \mu
\end{align*}

Now, observe that when $\ba,\ba'$ differ in all positions, since the database is a 
matching, the event $\ba'_j \in S_j$ is independent of the event $\ba_j \in S_j$ for every relation $S_j$;  in this case, $P(\ba'_j \in S_j  \mid \ba_j \in S_j) = P(\ba'_j \in S_j)$ for every $S_j$. On the other hand, if $\ba, \ba'$ agree in at least one position, then since $q$ is connected it will be that $P(\ba'_j \in S_j  \mid \ba_j \in S_j) = 0$ for some relation $S_j$. Thus, we can write:
\begin{align*}
\E[|q(I)|^2] & \leq 
 \sum_{\ba \neq \ba' \in [n]^k} \prod_{j=1}^{\ell} P(\ba_j \in S_j) P(\ba'_j \in S_j) + \mu \\
 & = (n^{2k}-n^k)  \prod_j (m_j/n^{a_j})^2  + \mu \\
 & = \left(1-n^{-k} \right) \mu^2 + \mu \\
 & \leq  \mu^2 + \mu
\end{align*}
\end{proof}

For a deterministic algorithm $A$ that computes the answers to a query
$q$ over a randomized instance $I$, let $fail$ denote the event that $|q(I) \setminus A(I)|>0$,
\ie the event that the algorithm $A$ fails to return all the output tuples.
The next lemma shows how we can use a bound on the expectation to obtain a 
bound on the probability of failure.

\begin{lemma} \label{lem:conditional:bound}
Let $I$ be a random matching database for a connected query $q$. Let $A$ be a deterministic 
algorithm such that $\E[|A(I)|] \leq f \E[|q(I)|]$, where $f \leq 1$. Let $\mu = \E[|q(I)|]$ and 
let $C_\alpha$ denote the event that $|q(I)| >  \alpha \mu$. Then,
$$ P(fail \mid C_{1/3}) \geq 1-9 f$$
\end{lemma}

\begin{proof}
We start by writing
\begin{align*} 
P(fail \mid C_\alpha) & = P(|q(I) \setminus A(I)| >0 \mid C_\alpha) \\
& \geq P( |A(I)| \leq \alpha \mu \mid C_\alpha) \\
& = 1 - P( |A(I)| > \alpha \mu \mid C_\alpha) 
\end{align*}
Additionally, we have:
\begin{align*}
\E[|A(I)] & =  \E[|A(I)| \mid C_\alpha] \cdot P(C_\alpha) + 
                            \E[|A(I)| \mid \neg C_\alpha] \cdot P(\neg C_\alpha) \\
& \geq    \E[|A(I)| \mid C_\alpha] \cdot P(C_\alpha) \\       
&  = P(C_\alpha) \sum_{t= \lfloor \alpha \mu \rfloor +1}^{\infty}  t \cdot P(|A(I)|=t \mid C_\alpha) \\
& \geq  P(C_\alpha) (\lfloor \alpha \mu \rfloor +1)  P(|A(I)| > \alpha \mid C_\alpha)     
\end{align*}
Combining the above two inequalities, we can now write
\begin{align*} 
P(fail \mid C_\alpha) & \geq 1 - \frac{\E[|A(I)|]}{(\lfloor \alpha \mu \rfloor +1) P(C_\alpha)} 
\geq 1 - \frac{f \mu}{(\lfloor \alpha \mu \rfloor +1) P(C_\alpha)} 
\end{align*}
We can now apply \autoref{lem:size:bound} to obtain $P(C_\alpha) = P(|q(I)| > \alpha \mu) \geq (1-\alpha)^2\mu/(\mu+1)$. Thus,
\begin{align*} 
P(fail \mid C_\alpha)  \geq 1 - \frac{f \mu}{\lfloor \alpha \mu \rfloor +1} \cdot \frac{\mu+1}{\mu (1-\alpha)^2}
=  1 - \frac{f (\mu+1)}{(\lfloor \alpha \mu \rfloor +1) (1-\alpha)^2}
\end{align*}
We can now choose $\alpha = 1/3$ to obtain that
$$ P(fail \mid C_{1/3}) \geq 1- (9/4) f \frac{\mu+1}{\lfloor  \mu/3 \rfloor+1}$$
The final step is to show that the quantity $\frac{\mu+1}{\lfloor  \mu/3 \rfloor+1}$ is upper bounded by 4 for any (positive) value of $\mu$. We distinguish here two cases:
\begin{itemize}
\item If $\mu<3$, then $\lfloor  \mu/3 \rfloor  = 0$. Thus,  
$\frac{\mu+1}{\lfloor  \mu/3 \rfloor+1} = \mu+1 < 4$.
\item If $\mu \geq 3$, we use the fact $\mu/3 \leq \lfloor \mu/3 \rfloor +1$ to obtain that
$\frac{\mu+1}{\lfloor  \mu/3 \rfloor+1} \leq (\mu+1)/(\mu/3) = 3(1+1/\mu) \leq 3(1+1/3) = 4$. 
\end{itemize}
This concludes the proof of the lemma.
\end{proof}

\end{document}